\documentclass[twoside,11pt]{article}

\usepackage[abbrvbib, preprint]{jmlr2e}

\jmlrheading{}{2026}{}{8/26}{NA}{21-0000}{Andrea Mascaretti and Daniel R. Kowal}

\ShortHeadings{Graph-Dependent Shrinkage Priors}{Mascaretti and Kowal}
\firstpageno{1}

\usepackage{amsmath}

\usepackage{color}
\newcounter{subfigure}[figure]

\newcommand{\R}{\mathbb{R}}

\newcommand{\rank}{\mathrm{rank}}
\newcommand{\kernel}{\mathrm{ker}}

\newcommand{\Normal}{\mathcal{N}}
\newcommand{\bfbeta}{\pmb{\beta}}

\newcommand{\bfeta}{\pmb{\eta}}
\newcommand{\bfy}{\pmb{y}}
\newcommand{\bfh}{\pmb{h}}
\newcommand{\bfell}{\pmb{\ell}}
\newcommand{\bfone}{\pmb{1}}

\newcommand{\dd}{\mathrm{d}}
\newcommand{\hzero}{h_0}


\newtheorem{lemma}{Lemma}

\usepackage{verbatim}
\usepackage{booktabs}
\usepackage{placeins}
\newcommand{\muG}{\mu_{k}}

\begin{document}

\title{Graph-dependent shrinkage priors for Bayesian trend filtering}

\author{\name Andrea Mascaretti \email amascare@sissa.it\\
    \addr Scuola Internazionale Superiore di Studi Avanzati (SISSA)\\
    Trieste, Italy
    \AND
    \name Daniel R. Kowal \email dan.kowal@cornell.edu\\
    \addr Department of Statistics and Data Science, Cornell University\\
    Ithaca, NY 14850, USA}

\editor{TBD}

\maketitle

\begin{abstract}
  Many common data dependencies can be characterized by graphs: time
  series data are sequential (chain graph), images appear as pixels
  (lattice graph), areal data are defined by neighboring units
  (spatial adjacency graph), etc. Graph trend filtering seeks to
  smooth and predict such data. However, classical trend filtering
  only incorporates the graph for estimation of the trend, which
  limits its adaptivity, and is brittle in the presence of missing
  data. Further, it lacks uncertainty quantification and faces certain
  computing challenges. We address these limitations with a
  comprehensive Bayesian framework for (graph-) dependent data. Our
  approach leverages the graph at three critical junctures: 1) the
  trend, to enable smoothing, imputation, and prediction; 2) the
  local shrinkage, to enhance adaptivity and precision; and 3) the
  MCMC sampling algorithm, to deliver scalable posterior (predictive)
  inference via sparse and banded operations. For the proposed graph-dependent shrinkage priors, we study the local concentration and adaptivity properties and establish conditions for posterior propriety.  Simulation studies
  demonstrate that, relative to state-of-the-art frequentist and
  Bayesian alternatives, this framework provides more accurate point
  estimates, more precise interval estimates, and highly competitive
  computing. We apply our methods for
  spatio-temporal modeling and forecasting of local area unemployment data
  for every county in the continental U.S. during the 2020 COVID-19 unemployment shock.
\end{abstract}

\begin{keywords}
    Areal data; finite-difference operator; horseshoe prior; Markov chain Monte Carlo; time series.
\end{keywords}

\section{Introduction}\label{sec:introduction}
Dependent data are prevalent and diverse: prominent settings include time series data, image data,
spatial data, and other cases where there is some natural
ordering, grouping, or structure to the data. Despite the many
differences among these cases, several unifying themes persist.
Primary modeling tasks---prediction, imputation, (trend) estimation,
and uncertainty quantification---rely critically on careful modeling
of data dependencies. These dependencies may be relevant for multiple
aspects of the data-generating process, including expectations,
variances, priors, etc.  This is especially important for
Bayesian methods, which explicitly make distributional assumptions to
obtain model-based (posterior) inference. Additionally, dependent data
models typically incur a greater computational cost than independent counterparts, and often require specialized algorithms for
prediction and inference.  Finally, many types of dependencies may be
described by a graph, where vertices denote observations and edges
denote potential data dependencies. Examples include chain graphs to encode
the sequential nature of time series data, lattice (or grid) graphs to
represent adjacent pixels in image data, and spatial adjacency graphs
to identify the neighboring units among areal (spatial)
data. Consequently, models and algorithms that can be expressed in
terms of graphs gain substantial generalizability across different
types of dependent data.

Within this broad context, a central task is \emph{trend filtering},
which seeks to predict and smooth data observed over a
graph. ``Ordinary" trend filtering considers sequential (time series)
data and seeks to construct locally-adaptive trend estimates via
sparse $\ell_1$-penalization of $k$th order differences
\citep{kim$ell_1$TrendFiltering2009,tibshiraniAdaptivePiecewisePolynomial2014}.
That work sparked many extensions, including for
unequally-spaced data
\citep{madridpadillaAdaptiveNonparametricRegression2020} and several
Bayesian adaptions
\citep{faulknerLocallyAdaptiveSmoothing2018,kowalDynamicShrinkageProcesses2019,hengBayesianTrendFiltering2023,kangObjectiveBayesianTrend2025}. \cite{wangTrendFilteringGraphs2016}
provided the critical generalizations for trend filtering on a
graph. Specifically, \cite{wangTrendFilteringGraphs2016} designed
graph difference operators to embed within an $\ell_1$-penalty,
leveraging the induced sparsity patterns to gain smooth yet
locally-adaptive trend estimates. It is this sparsity
that distinguishes graph trend filtering from other approaches, such
as Laplacian smoothing \citep{smolaKernelsRegularizationGraphs2003}
and Gaussian Markov Random Fields \citep{rueFastSamplingGaussian2001},
which instead rely on $\ell_2$-penalties and Gaussian priors,
respectively. Subsequent work included total variation denoising
($k=0$) with $d$-dimensional grid graphs
\citep{sadhanalaHigherOrderTotalVariation2017,sadhanalaMultivariateTrendFiltering2024},
exponential families \citep{sadhanalaExponentialFamilyTrend2024}, and
multivariate data and non-convex penalties
\citep{varmaVectorValuedGraphTrend2020}. Graphs also appear in a
variety of other statistical and machine learning tasks; see
\cite{khoshraftarSurveyGraphRepresentation2024}.

These contributions clearly established the utility of (ordinary or
graph) trend filtering and prioritized the theoretical guarantees for trend estimation. However, they focused exclusively on point
estimation without accompanying uncertainty quantification. The Bayesian
approaches are an exception, 
but only considered ordinary trend filtering for time series data and do not apply for general graph dependencies. Furthermore, nearly all  
previous approaches only incorporated the graph dependencies in the
model for the trend, and
often only for total variation denoising or the fused lasso 
($k=0$). However, this regularization is controlled by a single tuning
parameter that 1) is constant over graph, which limits adaptivity, and 2) must be estimated 
such as through cross-validation, which adds to the  already substantial computational costs for graph-dependent data.
Lastly, these approaches did not address missing
data.  Even in
the simpler case of unequally-spaced time series data, the
\cite{madridpadillaAdaptiveNonparametricRegression2020} modification of ordinary trend
filtering required a nearest-neighbors pre-processing step that
redefined the graph and added a tuning
parameter. Thus, the challenges of missing data for general graph
dependencies are substantial.

Motivated by these challenges, we propose a Bayesian framework for
modeling graph-dependent data. Our approach leverages the graph
dependencies at three junctures. First, we adapt the
\cite{wangTrendFilteringGraphs2016} graph difference operators to define a prior
(rather than penalty) for the trend, thus encouraging graph-based
smoothness across connected vertices. Second, we introduce
\emph{graph-dependent shrinkage priors} to enable locally-adaptive
regularization informed by the graph. These priors substantially
improve trend estimation, especially for signals that exhibit both
smoothness and abrupt changes over the graph, and provide more precise
yet well-calibrated uncertainty quantification. And third, we develop
a scalable MCMC algorithm that combines sparse and banded matrix
operations with convenient parameter expansions to provide efficient
posterior predictive inference and imputation. Notably, both the
graph-dependent trend parameters and the graph-dependent shrinkage
parameters are each sampled jointly and all without a loop. These
methods are developed for general graphs, and thus apply across many
different types of dependent data.

\begin{figure}[h]
    \centering
    \includegraphics[width=.8\textwidth]{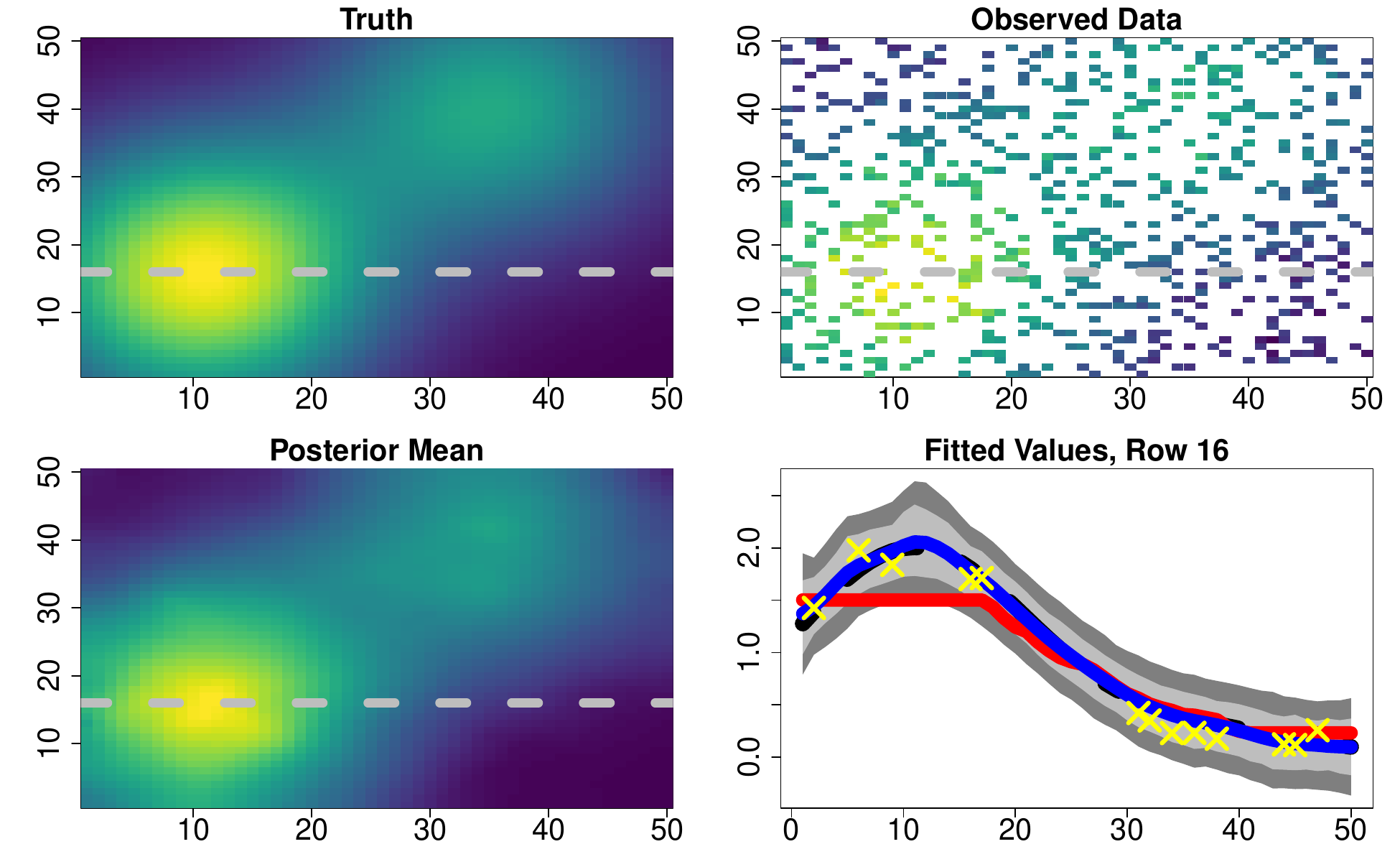}
    \caption{Bayesian trend filtering with graph-dependent shrinkage for a $50\times 50$ lattice. The true trend (top left) is corrupted by measurement error and 75\% missingness to obtain the observed data (top right). The fitted trend from the proposed approach (bottom left) accurately recovers the truth. Inspecting one row of the lattice (dashed gray horizontal lines), the bottom right plot shows the true trend (dashed black), observed data (yellow x's), fitted values (blue) and 95\% pointwise intervals (light gray) and simultaneous bands (dark gray) for the proposed approach and frequentist trend filtering estimates \citep{wangTrendFilteringGraphs2016} (red).}
    \label{fig:sim_lattice}
\end{figure}

To illustrate the challenges in this setting---and how our methods
address them---we provide an example in
Figure~\ref{fig:sim_lattice}. These are simulated data on a lattice
with a smooth, nonlinear trend (top left) corrupted by measurement
error and 75\% missingness (top right). Visually, it is difficult to
identify much of a trend from the observed data; yet the proposed
approach recovers it cleanly (bottom left) by borrowing information
across neighboring pixels for both the trend and the
shrinkage. Inspecting one row of the lattice (bottom right), our
Bayesian approach accurately captures the trend, reliably imputes the
missing data, and provides suitable uncertainty estimates. By
contrast, the estimated trend from \cite{wangTrendFilteringGraphs2016} struggles
due to the missing data and does not natively provide uncertainty
quantification.

Central to our modeling efforts is a new family of graph-dependent
shrinkage priors (GDSP). We operate within the class of continuous, global-local shrinkage
priors 
\citep{strawdermanProperBayesMinimax1971,griffinAlternativePriorDistributions2005,carvalhoHorseshoeEstimatorSparse2010,polsonShrinkGloballyAct2011}. Unlike
sparsity-inducing (e.g., spike-and-slab) priors, which explicitly
model variable inclusion yet face substantial computational burdens,
continuous shrinkage priors are designed for aggressive shrinkage of
noise, minimal shrinkage of signals, and scalable posterior computing. The novelty is that the GDSP local scale parameters are \emph{graph-dependent}, so that neighboring parameters borrow shrinkage information across the graph. We study the local concentration and adaptivity properties of the  GDSP for arbitrary graphs. Further, we construct global shrinkage priors that learn the overall level of smoothness over the graph and, even more critically, ensure posterior propriety of the GDSP---which is not guaranteed simply by assigning shrinkage priors to graph differences. 

Although shrinkage parameters are traditionally assigned iid priors, recent work has considered the advantages of dependent priors in
different contexts. For linear regression analysis,
\cite{liuBayesianRegularizationGraph2014} and
\cite{chakrabortyGraphLaplacianPrior2019} introduced a ``graph
Laplacian prior" that generalized the Bayesian lasso to model the
precision matrix among regression coefficients. Similarly,
\cite{griffinStructuredShrinkagePriors2024} introduced general priors
that include both shrinkage and dependence among regression
coefficients. However, these shrinkage priors were not designed for
dependent \emph{data} or for trend filtering, prediction, and
imputation.  \cite{kalliTimevaryingSparsityDynamic2014} and
\cite{kowalDynamicShrinkageProcesses2019} introduced dynamic shrinkage
priors with time series models for shrinkage parameters. These methods
are restricted to sequential data and do not apply for general graph
dependencies.

For completeness, we acknowledge the literature on horseshoe priors
for Gaussian graphical models
\citep{liGraphicalHorseshoeEstimator2019,lingjaerdeScalableMultipleNetwork2023,sagarPrecisionMatrixEstimation2024}. Despite
the similarities in terminology, that context is quite different:
their priority is to learn sparse precision matrices from multivariate
Gaussian data, rather than to model data dependencies from a (known)
graph.

The remainder of the paper is organized as follows. In
Section~\ref{sec:ad_dep_data}, we introduce our Bayesian model for
graph trend filtering with graph-dependent shrinkage priors, including
theoretical analysis of this prior.  Section~\ref{sec:mcmc} details
our Markov chain Monte Carlo (MCMC) sampling algorithm. Simulation
studies are presented in Section~\ref{sec:sim-data}. In
Section~\ref{sec:us_unemployment}, the methods are applied for
space-time modeling and  forecasting of local area (county-level)
unemployment data for  3,108 counties in the continental U.S.\ during spring 2020. We
conclude in Section~\ref{sec:conclusions}. The appendix includes proofs of all results, additional theory, computational details, and supporting empirical results for simulations and real data analysis. 

\section{Bayesian trend filtering with graph-dependent shrinkage}\label{sec:ad_dep_data}

\subsection{Trend filtering on a graph}\label{sec-gtf}
Graphs are an effective and general tool to characterize dependencies among observations. Suppose we observe dependent data $y_i \in \mathbb{R}$ for $i=1,\ldots,n$. Dependencies may be modeled via  graph $G = \left(V, E\right)$  with vertices (or nodes)
$V = \left\{1, \dots, n\right\}$ corresponding to observations  and edges  $E = \left\{e_1, \dots, e_m\right\}$ that list the potential dependencies among them.  Prominent examples include:
\begin{enumerate}
    \item Time series data: the \emph{chain graph} connects sequential points, $E = \{(i-1, i), \ i=2,\ldots,n\}$, with $m=n-1$ edges;
    \item Image data: the \emph{lattice graph} connects neighboring pixels on a $d_1 \times d_2$ grid (Figure~\ref{fig:sim_lattice}) with $n=d_1d_2$ vertices and $m=2n - d_1 - d_2$ edges;
    \item Areal data: \emph{spatial adjacency graphs} connect neighboring areal units, such as U.S. counties (Section~\ref{sec:us_unemployment}). 
\end{enumerate}
These graphs also may be combined to account for multiple dependencies simultaneously, such as a time series of images (Section~\ref{sec:st_sims}) or spatio-temporal areal data (Section~\ref{sec:us_unemployment}). Independent observations are described by the empty graph $E = \emptyset$. Although seemingly trivial, this setting is worth revisiting because it clarifies how our dependent data models and priors  generalize beyond the familiar independent versions. 

The goal of graph trend filtering is to extract a signal  from noisy data observed over the graph: 
\begin{equation}
  \label{eq:likelihood}
    y_i = \beta_i + \epsilon_i,\quad \epsilon_i \mid \sigma_{\epsilon}^2 \overset{\mathrm{iid}}{\sim} \mathcal{N}\left(0, \sigma_{\epsilon}^2\right), \quad i \in V.
\end{equation}
Broadly, we seek 1) point estimates and uncertainty quantification for
the signal $\{\beta_i\}_{i \in V}$ and 2) prediction and imputation
for future or unobserved data, say $\{\tilde y_i\}_{i\in V}$. Central
to each task is the graph dependence: connected observations
$(i,j) \in E$ are expected to be similar, e.g., with
$\vert \beta_i - \beta_{j}\vert$ small. This \emph{smoothness} over
the graph is critical in each of those goals---point estimation,
uncertainty quantification, prediction, and imputation. Yet still, stark differences may arise even between connected
observations. For instance, time series undergo change points, with
consecutive observations that differ substantially; images change
abruptly at object boundaries; and neighboring areal units may diverge
due to structural differences such as local policies and
demographics. Thus, our models, estimates, and uncertainties cannot
oversmooth on the graph, but rather must be \emph{adaptive} to sharp
changes between connected observations.

In order to convert general graph dependencies into viable statistical
models and estimators, we first define \emph{graph difference
  operators}. Let $\delta:E \to \left\{-1, 0, 1\right\}^{n}$ be the
(oriented) incidence function, so that
$\delta \left(e\right) = \left(0, \dots, -1, 0, \dots, 1, \dots,
  0\right)^\top$ for edge $e = (i, j)$, where $-1$ and $1$ occur at
the $i$th and $j$th entries, respectively. Then the $k$th order graph
difference operators are defined recursively
\citep{wangTrendFilteringGraphs2016}:
  \begin{equation}\label{graph-diff-op}
  \Delta^{\left(k + 1\right)} = \begin{cases}
    \Delta^{\left(1\right)} = (\delta(e_1), \dots,
  \delta(e_m))^\top & k = 0,\\
    \Delta^{\left(k + 1\right)} = \left(\Delta^{\left(1\right)}\right)^\top \Delta^{\left(k\right)} & k \text{ is odd},\\
    \Delta^{\left(k + 1\right)} = \Delta^{\left(1\right)} \Delta^{\left(k\right)} & k \text{ is even}.
  \end{cases}
\end{equation}
This family of operators can be alternatively formulated in terms of
the graph Laplacian $L=D - A$, where $A$ is the adjacency matrix and
$D$ is the degree matrix. Specifically, $\Delta^{\left(2\right)} = L$
is the graph Laplacian,
$\Delta^{\left(3\right)} = \Delta^{\left(1\right)}L$ is the gradient
of the Laplacian, $\Delta^{\left(4\right)} = L^2$ is the bi-Laplacian,
etc. This is akin to computing derivatives of increasing order; when used as a penalty term, it 
encourages the estimated signal to follow a piecewise polynomial trend
over the graph.

To interpret these operators, consider the output from applying \eqref{graph-diff-op} to $\pmb \beta = (\beta_1,\ldots, \beta_n)^\top$:
\begin{equation}\label{gdo-apply}
    \pmb \omega = \Delta^{(k+1)}\pmb \beta.
\end{equation}
Note that the dimension changes with $k$:
$\Delta^{\left(k + 1\right)}$ is $m \times n$ when $k$ is even or zero
and $n \times n$ when $k$ is odd. Thus, $k$ determines whether the
output $\pmb \omega$ is indexed by edges ($k$ even or zero) or
vertices ($k$ odd). For completeness, we define $\Delta^{(0)} = I_n$
for $k=-1$, which sets $\pmb \omega = \pmb \beta$ to reproduce the
independent case. In general, we assume that $G$ is connected but provide exceptions for isolated nodes below.

When $k=0$, the first-order difference operator induces
$\omega_{e} = \beta_i - \beta_{j}$ for each edge $e=(i,j) \in E$. By
regularizing $\Vert \pmb \omega \Vert$ toward zero, we obtain a trend
$\pmb \beta$ that is locally constant over the graph $G$, so that
connected observations share the same value. In frequentist
estimation, $\ell_1$-penalization of $\pmb \omega$ is known as
\emph{total variation denoising} or the \emph{fused lasso}
\citep{padillaDFSFusedLasso2018,JMLR:v25:23-1061,fanApproximate$ell_0$penalizedEstimation2018}. Among
Bayesian methods, shrinkage via a Gaussian prior on $\pmb \omega$ is
often called a \emph{random walk prior} (or local level model)
\citep{rueGaussianMarkovRandom2005} for time series data or an
\emph{intrinsic conditional autoregressive} (ICAR) model for areal
(spatial) data \citep{besagConditionalIntrinsicAutoregressions1995}.

When $k=1$, the operator \eqref{graph-diff-op} yields  $\omega_i = n_{i} \beta_i - \sum_{j \sim i} \beta_{j}$, where
$n_{i}$ is the number of neighbors for the
$i$th node. Now, regularizing $\vert \omega_i \vert$ toward zero pulls
$\beta_i$ toward an average of its neighbors, which produces a locally
linear trend over the graph.

For illustration, consider a chain graph for time series data. When
$k=0$, $\omega_e = \beta_i - \beta_{i-1}$, so $\omega_e = 0$ implies
that the signal inherits the value from the previous time point,
$\beta_i = \beta_{i-1}$, and the trend is locally constant. When
$k=1$, $\omega_i = (\beta_{i+1} - \beta_i) - (\beta_i - \beta_{i-1})$,
so $\omega_i = 0 $ has two equivalent implications: 1) the
\emph{change} in trend is locally constant, i.e., the trend is locally
linear, and 2) the trend $\beta_i = (\beta_{i+1} - \beta_{i-1})/2$
averages the immediate time points before and after. The pattern
continues and applies for general graphs $G$, so that increasing $k$
leads to higher order polynomials akin to roughness penalties on the
$(k+1)$th derivative.

Based on the graph difference operators in \eqref{graph-diff-op},
\cite{wangTrendFilteringGraphs2016} proposed to estimate the trend by
solving
\begin{equation}\label{gtf-est}
    \pmb{\hat \beta}_\lambda = \arg\min_{\pmb \beta \in \mathbb{R}^n} \frac{1}{2} \Vert \pmb y - \pmb \beta \Vert_2^2 + \lambda \Vert \Delta^{(k+1)}\pmb \beta\Vert_1
\end{equation}
where $\lambda > 0$ is a tuning parameter. By design, the
$\ell_1$-penalty encourages sparsity in
$\pmb \omega = \Delta^{(k+1)}\pmb \beta$, which induces this
generalized notion of (polynomial) smoothness over the graph for the
estimated signal $\pmb{\hat\beta}_\lambda$. Equally important are the
non-sparse entries of $\pmb \omega$, which admit local adaptivity in
the trend estimate. \cite{wangTrendFilteringGraphs2016} designed
algorithms to solve \eqref{gtf-est} depending on the choice of $k$,
which are implemented in the \texttt{R} package \texttt{genlasso} and
used subsequently for comparisons (Section~\ref{sec:sim-data}). The
tuning parameter $\lambda$ is typically selected using
cross-validation or based on a pre-specified degrees of freedom.

For our purposes, we use a subtle modification to the graph
difference operator initially defined in
\cite{wangTrendFilteringGraphs2016}. Consider isolated vertices
$i^* \in V$ with $n_{i^*} = 0$. Such components do not factor into
\eqref{graph-diff-op}, so the trend estimator in \eqref{gtf-est}
satisfies $\hat \beta_{\lambda, i^*} = y_{i^*}$. Yet for Bayesian
inference, explicitly modeling each $\beta_{i^*}$ remains important:
we may still wish to apply a shrinkage prior for $\beta_{i*}$ and to
infer the trend $\pmb \beta$ jointly. Explicitly,
 \eqref{gdo-apply} does not involve \emph{any}
isolated vertices, so building a model or prior from this output would
neglect all such $\beta_{i^*}$. Instead, we may simply modify the
graph $G$ to include self-loops for isolated vertices. Equivalently,
each graph difference operator in \eqref{graph-diff-op} expands to be
block diagonal with an identity matrix,
$\mbox{bdiag}( \Delta^{\left(k + 1\right)}, I_{n^*})$, with the $n^*$
isolated vertices ordered last for convenience. Now, the output
\eqref{gdo-apply} satisfies $\omega_{i^*} = \beta_{i^*}$ for any
isolated vertex. This ensures that our dependent data models
appropriately generalize the independent version in which  all vertices
are isolated.

\subsection{Graph-dependent shrinkage priors}\label{sec:graph_shrink}

To translate the graph trend filtering problem into a Bayesian
framework, the natural strategy is to adopt the likelihood implied by
\eqref{eq:likelihood} and replace the $\ell_1$-penalty in
\eqref{gtf-est} with a prior on
$\pmb \omega=\Delta^{(k+1)}\pmb \beta$. The structure of this prior
requires great care. First, the prior must encourage (near-) sparsity
of $\pmb \omega$ in order to produce a smooth trend. At the same time, the prior must not overshrink the nonzero
entries of $\pmb \omega$, which would impede the model's adaptivity to
abrupt changes in the signal. Equivalently, the prior must have
sufficiently heavy tails. Many common priors do not satisfy these criteria, such as Normal-Inverse-Gamma priors and the Bayesian lasso \citep{parkBayesianLasso2008}. 
Second, the computational burden is substantial: besides the sample size $n$, computational performance
will be determined by the density of the graph and $k$, and
specifically by the number and location of nonzero elements of
$\Delta^{(k+1)}$ (see Section~\ref{sec:mcmc}). Priors that incur a
sizable computing cost for independent data, such as spike-and-slab
priors, are unappealing here. Third, 
$\pmb \omega$ may be indexed by edges ($k$ even
or zero) or vertices ($k$ odd). Either way, any local shrinkage
parameters, like the trend itself, are oriented on the graph. Thus, we
\emph{again} leverage the graph in our model, now to borrow
information among neighboring shrinkage parameters. Finally, graph differences supply relative and thus incomplete information, which may result in improper posteriors. We address each of these challenges below. 

Consider the general family of conditionally Gaussian, global-local shrinkage
priors:
\begin{align}\label{eq:prior}
    \omega_j \mid \tau, \lambda_j &\overset{\mathrm{ind}}{\sim} \mathcal{N}\left(0, \tau^2\lambda_j^2\right) \\
    \label{eq:prior-lam}
    \pmb \lambda &\sim  \pi (\pmb \lambda)%
\end{align}
where $j$ indexes edges ($k$ even or zero) or vertices ($k$ odd), $\pmb \lambda = (\lambda_1, \lambda_2, \ldots)^\top$ are the local shrinkage parameters with a dimension that depends on $k$, and $\tau > 0$ is a global shrinkage parameter. The conditionally Gaussian distribution in
\eqref{eq:prior} enables scalable computing for the
trend $\pmb \beta$ (see Section~\ref{sec:mcmc}). The choice of
$\pi(\pmb\lambda)$ determines the shrinkage behavior of the
hierarchical prior \eqref{eq:prior}--\eqref{eq:prior-lam}. Such priors
are typically iid: the horseshoe prior is
$\lambda_j \sim^\mathrm{iid} C^+(0,1)$ 
and the Bayesian
lasso is
$\lambda_j^2  \sim^\mathrm{iid}\mbox{Exp}(1/2)$. 
The
Normal-Inverse-Gamma prior fixes $\lambda_j = \lambda$ and uses an
Inverse-Gamma prior on $\lambda^{2}$. This is also a typical
choice for ICAR models (usually with $k=1$). These priors are all
considered as competitors in Section~\ref{sec:sim-data}.

The main limitation of these shrinkage priors is that they ignore the
graph information. However, building a graph-informed model for
$\pmb \lambda$ has several new challenges beyond the traditional graph
trend filtering problem (e.g., \eqref{gtf-est}). First, the local
shrinkage parameters are positive, $\lambda_j > 0$, unlike the trend
$\beta_i \in \mathbb{R}$. Thus, incorporating graph dependencies
cannot be done exactly as in \eqref{gdo-apply}.  Second,
$\pi(\pmb \lambda)$ should preserve the appealing shrinkage properties
of (conditionally) iid shrinkage priors such as the horseshoe \citep{carvalhoHorseshoeEstimatorSparse2010}. Without
those, the graph-informed smoothness and local adaptivity for the
trend from \eqref{eq:likelihood} and \eqref{gdo-apply} would be
lost. Finally, even if we can design such a prior that features
both graph dependencies \emph{and} desirable shrinkage properties, we
must ensure that model-fitting with \eqref{eq:likelihood},
\eqref{gdo-apply}, and \eqref{eq:prior}--\eqref{eq:prior-lam} remains theoretically valid and 
computationally tractable.

We introduce a \textit{graph-dependent shrinkage prior} (GDSP) to achieve these goals: 
\begin{align}
    \label{eq:prior-2}
    h_j &= \log(\tau^2\lambda_j^2)\\
    \label{eq:prior-3}
    \Delta_k^{k_h + 1}\pmb{h} &= \pmb{\eta}, \quad \eta_i \overset{iid}{\sim}Z(a, b, 0, 1), \quad a,b > 0 \\
    \label{eq:h-priors}
    h_0 =\log(\tau^2) &\sim \Normal\bigl(\muG,\; s_0^2\bigr), \quad \mu_0 \in \mathbb{R}, s_0 > 0
\end{align}
where $\muG = \mu_0 + s_k s_0^2/2$.
The GDSP hyperparameters are $a,b,\mu_0,s_0$, while $s_{k}$ (and thus $\mu_k$) depends only on whether $k$ is even or odd (see below). 
 This prior has several features that require explanation: the log transformations, the shrinkage graph difference operator $\Delta_k^{k_h + 1}$ together with the global scale prior \eqref{eq:h-priors}, and the (standard) $Z$-distribution in \eqref{eq:prior-3} with density $p_Z(z) = \{\mathrm{B}(a, b)\}^{-1}
{\{\exp{(z)}\}}^{a} {\{1 +
    \exp{(z)}\}}^{-(a + b)}$ \citep{barndorff-nielsenNormalVarianceMeanMixtures1982}.

  First, the log transformations in \eqref{eq:prior-2} and \eqref{eq:h-priors} ensure
  positivity of the global ($\tau$) and local ($\lambda_j$) shrinkage parameters. The
  log-variance parameters $h_j \in \mathbb{R}$ are more amenable to
  the differences and local averaging induced by the graph difference
  operator in \eqref{eq:prior-3}. These remain interpretable on the
  $\tau\lambda_j$-scale as multiplicative (rather than additive) effects.

  Second, the shrinkage graph difference operator $\Delta_k^{k_h + 1}$
  has its own order $k_h$, but also depends on $k$ from
  \eqref{gdo-apply}.
  When $k$ is odd,  $\pmb \lambda$ is indexed by the
  vertices of $G$ just like the trend, so we define
  $\Delta_k^{k_h + 1} = \Delta^{k_h + 1}$ exactly as in
  \eqref{graph-diff-op}.
  When $k$ is even or zero, $\pmb \lambda$ is
  indexed by edges so the same approach cannot apply. Then we
  construct the \emph{line graph} of $G$, which connects any edges
  that share a common vertex in $G$. For instance, the line graph of
  the chain graph connects the edge $(i-1, i)$ with both its
  predecessor $(i-2, i-1)$ and its successor $(i, i+1)$. With this new
  graph, we appropriately define the shrinkage graph difference
  operator via the initialization and recursions in
  \eqref{graph-diff-op}, again up to order $k_h$.
  In both cases, isolated vertices are treated exactly as
  for the trend (see the discussion at the end of
  Section~\ref{sec-gtf}).

Third, a more subtle aspect of \eqref{eq:prior-3} is that it only supplies a prior on the differences of the log-variances, not their level (Lemma~\ref{lem:level-free}). Crucially, a prior of the form
\eqref{eq:prior-2}--\eqref{eq:prior-3} is not guaranteed to be proper or ensure a proper posterior. The critical step is to parameterize and assign a prior to the level of the log-variances, which we represent as the global scale  $h_0 = \log(\tau^2)$. We first decompose the log-variances into a global (mean) component and an orthogonal component,
\begin{equation}\label{eq:h-decomp}
    h_j = h_0 + h_j^\perp, \quad \sum_{j=1}^m h_j^\perp = 0
\end{equation}
so that $h_0 = m^{-1}\sum_{j=1}^m h_j$ by design. Noticing that $\Delta_k^{k_h + 1}\pmb{h} = \Delta_k^{k_h + 1}\pmb{h}^\perp$ and $\lambda_j = \exp(h_j^\perp/2)$, it follows that \eqref{eq:prior-3} supplies a prior \emph{only} for the local scales $\pmb \lambda$. The global scale aims to learn the global smoothness over the graph. However, the data inform $s_k$ fewer directions than $\pmb\omega$ has coordinates, where $s_k = 1$ for $k$ odd and $s_k=m-n+1$ for $k$ even (see the discussion after Lemma~\ref{lem:global-shrinkage-marginal}). 
In the latter case, $s_k$ may be large and the global scale is pulled toward zero on account of these free parameters---thereby overshrinking across the whole graph. 
The proposed prior \eqref{eq:h-priors} counteracts this effect: the implied prior mean is $\mathbb{E}(\tau^2) = \exp\{\mu_0 + (s_k + 1)s_0^2/2\}$, where the offset $s_k s_0^2/2$ corrects the location shift due to the free coordinates. Thus, the prior that acts on $h_0$  is effectively $\Normal(\mu_0, s_0^2)$. The implied joint prior on $\bfh$ can be characterized explicitly (Lemma~\ref{lem:level-free}). 

Finally, the $Z$-distribution endows $\pi(\pmb\lambda)$ with  shrinkage behavior. In the special case without graph dependence ($k_h = -1$), the innovations in \eqref{eq:prior-3} become  $\lambda_j^2 = \exp(\eta_j)$. The
$Z$-prior for $\eta_j$ then induces an inverted-Beta prior for $\lambda_j^2$, which includes several important special cases outlined in Table~\ref{tab:inv-beta} \citep{kowalDynamicShrinkageProcesses2019}. Namely, when $k_h = -1$ and $a = b = 1/2$, the prior \eqref{eq:prior-2}--\eqref{eq:prior-3} exactly reproduces the horseshoe prior (similarly for isolated vertices with any $k_h$). More importantly, when $k_h \ge 0$, the same innovations act on the differences of the log-variances, and the GDSP generalizes the horseshoe and other shrinkage priors for dependent data settings. For instance, consider $k_h = 0$: for an edge $e = (i,j)$, the local shrinkage parameter $\lambda_i^2 = \lambda_j^2\exp(\eta_{ij})$ effectively takes the familiar (iid) shrinkage prior determined by $\exp(\eta_{ij})$ (Table~\ref{tab:inv-beta}) and scales it based on neighboring shrinkage parameters.

\begin{table}[h]
  \centering
  \caption[Special cases of the inverted beta prior]{The (standard) $Z$-distribution on the log-scale corresponds to several well-known shrinkage priors.}
  \begin{tabular}[b]{lll}
    Parameters & Prior & Reference \\
    \hline
    $a = b = \frac{1}{2}$ & Horseshoe & \cite{carvalhoHorseshoeEstimatorSparse2010}\\
    $a = \frac{1}{2}$, $b = 1$ & Strawderman-Berger & \cite{strawdermanProperBayesMinimax1971} \\
    $a = 1$, $b = 2 - c$, $c > 0$ & Normal-Exponential-Gamma & \cite{griffinAlternativePriorDistributions2005}\\
    $a = b \to 0$ & (Improper) normal-Jeffreys & \cite{figueiredoAdaptiveSparsenessSupervised2003} \\
  \end{tabular}
  \label{tab:inv-beta}
\end{table}

When $G$ is a chain graph, the GDSP
\eqref{eq:prior-2}--\eqref{eq:h-priors} is related to the dynamic
shrinkage prior \citep{kowalDynamicShrinkageProcesses2019}, which
instead used a stationary autoregressive model in place of
\eqref{eq:prior-3}. The primary differences here are the use of
$\Delta_k^{k_h + 1}$ for general dependent data described by a
graph, including the necessary accommodations for even vs.\! odd $k$
in \eqref{gdo-apply}, the global scale prior \eqref{eq:h-priors}, and isolated vertices, which were not
addressed in previous work. This
generalized prior also requires modifications to the MCMC sampling
algorithm (see Section~\ref{sec:mcmc}).

\subsection{Theory for GDSP}\label{sec-theory}
To further justify the proposed GDSP, we detail the ways in which it
combines graph information with desirable shrinkage
behavior. Specifically, we study  prior concentration,  posterior adaptivity, and  posterior propriety. 
Following previous work on shrinkage priors, we focus on the
\emph{shrinkage coefficients} $\kappa_i = 1/(1+\tau^2\lambda_i^2)$. To
interpret this term, consider the likelihood \eqref{eq:likelihood} and
the prior \eqref{eq:prior}. For simplicity, fix $\sigma_\epsilon = \tau = 1$
and take $k = -1$ in \eqref{gdo-apply} to neglect any graph smoothing
for the trend. Then the estimated trend is
$\mathbb{E}(\beta_i\mid \pmb y) = \{1 -\mathbb{E}(\kappa_i\mid \pmb y
)\}y_i$. The shrinkage coefficients satisfy $\kappa_i \in (0,1)$ almost surely, but the endpoints are meaningful:
$\kappa_i \approx 0$ implies no shrinkage while $\kappa_i \approx 1$
implies aggressive shrinkage to zero.
For reference, the horseshoe prior uses
$\kappa_i \sim^\mathrm{iid} \mathrm{Beta}(1/2, 1/2)$ to place
substantial prior mass near both endpoints.

For the GDSP, we are similarly interested in the behavior of
$\kappa_i$, but also in how its distribution is informed by its
neighbors $\left\{\kappa_{j}\right\}_{j \sim i}$. We focus on non-isolated vertices of the shrinkage graph: $G$ itself when $k$ is odd, and the line graph of $G$ when $k$ is even or zero. We write $n_i$ to denote the number of neighbors of vertex $i$,  either a node ($k$ odd) or an edge ($k$ even or zero), and require that $n_i \geq 1$. For
isolated vertices, the prior implies
$\kappa_{i} \sim \mathrm{Beta}(b, a)$ as for iid shrinkage priors. We specify $k_h = 0$ here, which illustrates the graph-informed
 shrinkage behavior and is the default we use in the
applications. Generalizations for higher-order differences ($k_h=1$) are in Appendix~\ref{sec:kh1}. Initially, we fix $\tau = 1$, as is standard.

We first derive the conditional prior density of the GDSP, using $a=b$ for convenience and in anticipation of the horseshoe GDSP with $a=b=1/2$. All proofs are in Appendix~\ref{sec:proofs}.
\begin{theorem}[Conditional Prior Density]\label{thm:cond-prior}
For the GDSP with $k_h = 0$ and $a = b$, the conditional prior
density of $\kappa_i$ given the neighboring values on the shrinkage graph
$\{\kappa_j\}_{j \sim i}$
is
\begin{equation}\label{eq:cond-prior}
  \pi(\kappa_i \mid \{\kappa_j\}_{j \sim i}) \;\propto\;
  \frac{\{\kappa_i(1-\kappa_i)\}^{a n_i - 1}}
       {\prod_{j \sim i}\bigl\{\kappa_i(1-\kappa_j) + (1-\kappa_i)\kappa_j\bigr\}^{2a}}.
\end{equation}
\end{theorem}
The conditional prior~\eqref{eq:cond-prior} induces local coherence
in shrinkage behavior through the neighboring factors. Each term in the denominator,
$\kappa_i(1-\kappa_j) + (1-\kappa_i)\kappa_j$ is near zero when $\kappa_i$
and $\kappa_j$ lie near the same extreme, and near one when they lie near
opposite extremes.  When the neighbors unanimously do not apply shrinkage ($\kappa_j \approx 0$ for all $j \sim i$), the conditional prior  concentrates near zero;
when they all apply shrinkage ($\kappa_j \approx 1$ for all $j \sim i$), it concentrates near
one.
The GDSP propagates shrinkage decisions across the graph, encouraging spatially adaptive patterns to distinguish signal and noise. 

More concretely, we show how the shrinkage behavior under the prior reacts to the shrinkage effects among the neighbors, and specifically for the horseshoe GDSP with $a=b=1/2$.
First, we consider the case when the neighbors are unanimous in
non-shrinkage, i.e., adapting for signals.
\begin{theorem}[Prior Concentration: Neighbors Unshrunk]\label{thm:prior_lims}
    For the horseshoe GDSP with $k_h = 0$, $a
    = b = 1 / 2$, the conditional prior
    distribution  satisfies
    \[
    \mathbb{P}(\kappa_i < \varepsilon \mid \{\kappa_j\}_{j
      \sim i}) \to 1 \quad \mbox{as }  \kappa^* \to 0
      \]
      for any $\varepsilon > 0$, where $\kappa^* =
    \max_{j \sim i} \{\kappa_j\}$.
\end{theorem}
Similar behavior occurs when the neighbors are unanimous in aggressive
shrinkage, i.e., smoothness.
\begin{theorem}[Prior Concentration: Neighbors Shrunk]\label{thm:prior-concentration-one}
For the horseshoe GDSP with $k_h = 0$, $a
    = b = 1 / 2$, the conditional prior
    distribution  satisfies
    \[
    \mathbb{P}(\kappa_i > 1 - \varepsilon \mid \{\kappa_j\}_{j \sim i}) \to 1 \quad \mbox{as } \kappa_* \to 1
    \]
    for any $\varepsilon > 0$, where $\kappa_* = \min_{j \sim i} \kappa_j$.
\end{theorem}
Whether the neighbors unanimously favor non-shrinkage or shrinkage,  the conditional prior  for $\kappa_i$ under the horseshoe GDSP  collects mass in the corresponding region.

Next, we investigate the properties of the conditional \emph{posterior} distribution of $\kappa_i$. For simplicity, we omit graph dependence in the trend ($k=-1$) so that $\omega_i = \beta_i$ and fix $\sigma_\epsilon = 1$.
\begin{theorem}[Posterior Shrinkage Adaptation]\label{thm:posterior-shrinkage}
Consider the horseshoe GDSP with  $a = b = 1/2$, $k_h = 0$, $k = -1$, and $\sigma_\epsilon=1$. The conditional posterior of $\kappa_i$ satisfies:
\begin{enumerate}
\item[\textbf{(i)}] \textbf{Neighborhood-driven shrinkage:} for any $\varepsilon > 0$ and any $M > 0$,
    \[
    \mathbb{P}(\kappa_i > 1 - \varepsilon \mid \bfy, \{\kappa_j\}_{j \sim i}) \to 1 \quad \mbox{as } \kappa_* \to 1
    \]
    uniformly for $\bfy \in \mathcal{Y}_M =\{\bfy \in \R^n : |y_i| \le M\}$, where $\kappa_* = \min_{j \sim i} \kappa_j$.

\item[\textbf{(ii)}] \textbf{Signal preservation:} If $\kappa_j \in (0, 1)$ for all $j \sim i$, then for any $\varepsilon > 0$,
\[
\mathbb{P}(\kappa_i < \varepsilon \mid \bfy, \{\kappa_j\}_{j \sim i}) \to 1 \quad \mbox{as } |y_i| \to \infty.
\]
\end{enumerate}
\end{theorem}
Theorem~\ref{thm:posterior-shrinkage} shows the bidirectional
adaptivity of the horseshoe GDSP. First, (i) shows that if the neighbors unanimously recommend aggressive shrinkage, then the conditional
posterior for $\kappa_i$ follows suit. Thus, the trend $\pmb \beta$
will exhibit local smoothness. Alternatively, in (ii), we find that as
long as the neighboring shrinkage behavior is moderated, the horseshoe GDSP will remove shrinkage for  strong signals. Thus, the trend $\pmb \beta$ will exhibit local adaptivity.

By comparison, higher-order difference operator ($k_h = 1$) link each local shrinkage parameter to its neighbors via $\prod_{j \sim i}(1-\kappa_j)/\kappa_j = \prod_{j \sim i} \lambda_j^2$ rather than $\kappa_*$ or $\kappa^*$ (see Appendix~\ref{sec:kh1}). In general, larger $k_h$ encourages more smoothing over the shrinkage graph, but the neighbors similarly inform prior concentration and posterior adaptation as for $k_h=0$. 

Finally, we establish posterior properiety under the GDSP.  This result is nontrivial: shrinkage priors combine substantial mass near zero with heavy tails, while graph-difference operators only inform relative (not absolute) behavior. Thus, simply combining them is not sufficient. Perhaps surprisingly, the delicate term is the global scale $\tau$: without specifying its prior carefully, the induced joint posterior may not be proper or the estimated trend may oversmooth. 

Anticipating the central role of $\tau$, we first provide additional clarity on the joint prior for 
$\bfh$ under the graph-differencing in \eqref{eq:prior-3} and the global scale in \eqref{eq:h-priors}. 
\begin{lemma}[Joint prior on the log-variance]\label{lem:level-free}
Suppose $G$ is connected with $n \geq 2$ vertices, $k \geq 0$, and
$k_h \geq 0$. Let
$g(\bfh) = \prod_i p_Z\bigl\{(\Delta_k^{k_h+1}\bfh)_i\bigr\}$, the
product running over the rows of $\Delta_k^{k_h+1}$, denote the prior
density kernel induced by \eqref{eq:prior-3}. Then (i) 
$g(\bfh + c\,\bfone_m) = g(\bfh)$ for every $c \in \R$; (ii) the joint prior is $\pi(\bfh) \propto g(\bfh) \phi\bigl(\hzero; \muG,\, s_0^2\bigr)$, where $\phi(\cdot\,; \mu, s^2)$ is the $\Normal(\mu, s^2)$ density and $\hzero = m^{-1}\bfone_m^\top\bfh$; and (iii) $\pi(\bfh)$ is proper. 
\end{lemma}
Lemma~\ref{lem:level-free} first shows that the graph-differencing in \eqref{eq:prior-3} is uninformative about the level of $\bfh$. It then describes how the global scale prior \eqref{eq:h-priors} and the decomposition \eqref{eq:h-decomp} combine to induce a proper prior for $\bfh$ using independent global and (graph-informed) local components.

We give the posterior propriety result using an Inverse-Gamma prior on the error variance:
\begin{equation}\label{eq:sigma-prior}
  \sigma_\epsilon^2 \sim \mathrm{Inverse\text{-}Gamma}(a_\sigma, b_\sigma),
  \qquad a_\sigma, b_\sigma > 0.
\end{equation}
No conditions are needed beyond connectedness of the graph.
\begin{theorem}[Posterior Propriety]\label{thm:posterior-proper}
Suppose $G = (V, E)$ is connected with $n \geq 2$ vertices and $\bfy$ is
fully observed. Under the likelihood~\eqref{eq:likelihood}, the
conditionally Gaussian prior~\eqref{eq:prior}--\eqref{eq:prior-lam}, the
GDSP prior~\eqref{eq:prior-2}--\eqref{eq:h-priors}, and the
prior~\eqref{eq:sigma-prior}, the joint posterior distribution of
$(\pmb\beta, \pmb\lambda, \tau^2, \sigma_\epsilon^2)$ is proper for every
$k \geq 0$ and every $k_h \geq 0$.
\end{theorem}
The result rests on properties of both global and local components. The global scale prior \eqref{eq:h-priors} has all negative
moments finite, which controls the growth of the marginal likelihood near
total shrinkage for every $k$ and $k_h$
(Lemma~\ref{lem:global-shrinkage-marginal}). For the local
shrinkage, the $Z$-innovations have exponential tails, so the prior on  $\pmb h^\perp$ is proper for every $k_h \geq 0$
(Lemma~\ref{lem:hperp-proper}). Together these results
bound the posterior integral for any connected graph, any $k \geq 0$, and
any $k_h \geq 0$.
The excluded values $k = -1$ and $k_h = -1$ are proper trivially,
since no graph is involved and every prior is proper.

\section{Posterior sampling via Markov Chain Monte Carlo}\label{sec:mcmc}

We design a Gibbs sampling algorithm that exploits several critical features of the model. First, we use sparse and banded matrix operations---reflecting the graph's inherent structure---to jointly sample the trend parameters $\bfbeta$ (all without a loop). Second, we leverage two data augmentations and a Sherman-Morrison update to apply sparse and banded matrix operations for sampling the log-variances $\bfh$. All global parameters are sampled using standard steps. Finally, prediction and imputation sampling is straightforward using (conditionally) independent Gaussian draws. Perhaps surprising, the resulting algorithm achieves competitive or superior time performance relative to frequentist alternatives (Section~\ref{sec:sim-data}), while providing full posterior and predictive inference and principled imputation of missing data. 

We make a few simplifying assumptions for clarity. Initially, we assume no missing data; the case of missing data adds an imputation step (Section~\ref{mcmc-imp}) within the Gibbs sampler. We focus on the horseshoe GDSP ($a = b =1/2$), but extensions to the other choices in Table~\ref{tab:inv-beta} are available. Finally, we present all displays for the edge-indexed case, in which $\pmb\omega$, $\pmb\lambda$, and $\pmb h$ have common dimension $m = |E|$; for odd $k$, these vectors are vertex-indexed and the dimension is $n = |V|$ instead.

Given an initialization (see Appendix~\ref{sec:extra-mcmc}), the Gibbs sampler cycles repeatedly through full conditionals for the trend $\pmb\beta$ (Section~\ref{mcmc-trend}), the log-variance $\pmb h$ (Section~\ref{mcmc-var}), the observation error variance $[\sigma_\epsilon^{2} \mid \pmb y, -] \sim
\mathrm{Inverse\text{-}Gamma}\bigl(a_\sigma + n/2,\;
b_\sigma + \sum_{i=1}^n(y_i-\beta_i)^2/2\bigr)$, and, if missing data are present, imputation  (Section~\ref{mcmc-imp}). 

\subsection{Sampling the trend}\label{mcmc-trend}
Under the likelihood \eqref{eq:likelihood} and the (conditional) prior implied by \eqref{gdo-apply} and \eqref{eq:prior}, the full conditional posterior distribution for the trend  is $[\pmb{\beta} \mid \pmb y, -] \sim \mathcal{N}_n(Q^{-1}_{\beta}\pmb{\ell}_{\beta},
   Q^{-1}_{\beta})$, where
\begin{equation}
  \label{eq:awol}
    Q_{\beta} = \sigma_\epsilon^{-2} I_n +
\tau^{-2}(\Delta^{\left(k +
  1\right)})^\top\Lambda^{-1}\Delta^{\left(k + 1\right)},
\end{equation}
$\Lambda  = \mathrm{diag}(\lambda_1^2,\ldots,\lambda_m^2)$, and $\pmb{\ell}_{\beta} = \sigma_\epsilon^{-2}\pmb{y}$.
This full conditional distribution is general, regardless of the continuous shrinkage induced by $\pi(\pmb \lambda)$ in \eqref{eq:prior-lam}.  Notably, $Q_\beta$ is the critical term for computational costs in sampling $\pmb \beta$. There are two key features in our approach.

 First, the sparsity of $Q_\beta$ is determined by the graph structure. The term $\sigma_{\epsilon}^{-2}I_n$ only impacts the diagonal and  $\Lambda^{-1}$ is available without any numerical matrix inversion. Therefore, the sparsity pattern of $Q_{\beta}$ is governed exclusively by the product $(\Delta^{\left(k+1\right)})^{\top}\Lambda^{-1} \Delta^{\left(k+1\right)}$, which encodes the conditional dependencies among the nodes. For instance, consider a chain graph with $k=1$, akin to second-order differencing. Then $Q_{\beta}$ is a banded  matrix  with non-zero elements appearing only on the main diagonal and the first two off-diagonals above and below it. This structural sparsity extends naturally to more complex topologies. Continuing with $k=1$,  $\Delta^{\left(2\right)} = D - A$ is the graph Laplacian, where $D$ is the (diagonal) degree matrix and $A$ is the (sparse) adjacency matrix. For a lattice, each node has at most four neighbors, so $\Delta^{\left(2\right)}$ contains at most five non-zero entries per row. Consequently, the key term  $(\Delta^{\left(k+1\right)})^{\top}\Lambda^{-1} \Delta^{\left(k+1\right)}$
 only has non-zero entries for node pairs that are either directly adjacent or share a common neighbor. This is a small fraction: out of the $n^2$ total entries in $Q_{\beta}$, only about $13n$ are non-zero.
Thus, $Q_\beta$ is highly sparse with density $\mathcal{O}\left(1/n\right)$.

Second, we do not ever explicitly form $Q_\beta^{-1}$, despite its apparently central role in $[\pmb{\beta} \mid \pmb y, -]$. Instead, we sample $\pmb{\beta}$ jointly, all without a loop, using efficient matrix decompositions and forward-/back-solve operations. These exploit the sparse or banded structure of $Q_\beta$. The first step is to compute the Cholesky decomposition $Q_{\beta} = L_{\beta}L_{\beta}^\top$, where
$L_{\beta}$ is a lower-triangular matrix. Although this decomposition must be re-computed at each MCMC iteration---$Q_\beta$ depends on parameters ($\sigma_\epsilon$, $\tau^2$, and $\pmb\lambda$) that will be updated---its sparsity pattern will remain the same, which can be exploited for  time- and memory-efficient computing. Next, we solve $L_{\beta}\pmb{a}_{\beta} = \pmb{\ell}_{\beta}$ for
$\pmb{a}_{\beta}$ followed by 
$L_{\beta}^\top\pmb{\beta} = \pmb{a}_{\beta} + \pmb{e}_{\beta}$ for $\pmb \beta$, where 
$\pmb{e}_{\beta} \sim \mathcal{N}_n\left(\pmb{0}, I_n\right)$. Similar approaches have appeared previously for specific graphs in spatial and time series contexts \citep{rueFastSamplingGaussian2001,kastnerAncillaritySufficiencyInterweavingStrategy2014,kowalDynamicShrinkageProcesses2019}.  

\subsection{Sampling the log-variances}\label{mcmc-var}
Because the shrinkage parameters also share graph  dependencies, we seek to adapt the strategy from Section~\ref{mcmc-trend}: a joint draw, all without a loop, that leverages the sparsity of the graph. However, there are three obstacles for the GDSP shrinkage parameters that do not occur for the trend. First, the scales are positive. This is easily addressed by sampling on the log scale $h_j = \log(\tau^2\lambda_j^2)$, which removes the constraint. Second, neither the conditional model for $\bfh$ implied by \eqref{eq:prior} nor
the innovations of \eqref{eq:prior-3} are Gaussian. We identify data  augmentations customized to each case to induce conditionally Gaussian distributions. And finally, the GDSP not only consists of the graph differences in \eqref{eq:prior-3}, but also the global scale \eqref{eq:h-priors} with the decomposition from \eqref{eq:h-decomp}, neither of which is present for the trend $\pmb\beta$. 

We first consider the contributions from \eqref{eq:prior}. From \eqref{eq:prior} and \eqref{eq:prior-2}, $\omega_j \mid h_j \sim \Normal(0, \exp(h_j))$ independently. Then conditional on $\pmb \beta$,  $\pmb \omega = \Delta^{(k+1)}\pmb \beta$ is given and we may express $\bfh$ on the additive scale with pseudo-data 
\begin{equation}\label{eq:pseudo-data}
  \tilde\omega_j \;=\; \log(\omega_j^2 + c_0) \;=\; h_j + e_j,
  \qquad e_j \overset{\mathrm{iid}}{\sim} \log\chi_1^2,
\end{equation}
where $c_0 \geq 0$ is a numerical guard, taken to be zero unless some
$\omega_j^2$ underflows.
To handle the non-Gaussian errors, we incorporate a popular strategy from Bayesian stochastic volatility models that approximates the $\log\chi_1^2$ distribution as a discrete mixture of Gaussians. Specifically, we use the ten-component Gaussian mixture from \citet{omoriStochasticVolatilityLeverage2007}, with fixed and known weights
$q_{1:10}$, means $m_{1:10}$ and variances $v_{1:10}^2$. Let $z_j \in \{1,\ldots,10\}$ be mixture indicators so that
$e_j \mid z_j \sim \Normal(m_{z_j}, v_{z_j}^2)$. This approximation requires sampling the mixture indicators from their full conditionals, which follow the discrete distribution
\begin{equation}\label{eq:z-update}
  \Pr(z_j = \ell \mid \tilde\omega_j, h_j) \;\propto\; q_\ell\,\phi(\tilde\omega_j - h_j;\, m_\ell, v_\ell^2),
  \qquad \ell = 1,\ldots,10.
\end{equation}
Then, conditional on
$z_j$, the contribution from \eqref{eq:prior}  is a Gaussian location model,
$\tilde\omega_j \mid h_j, z_j \sim \Normal(h_j + m_{z_j},\, v_{z_j}^2)$.

Second, consider the GDSP innovations in \eqref{eq:prior-3}. We augment each $\eta_i$ with an auxiliary variate $\xi_i$  with joint density proportional to $\exp\{(a-b)\eta_i/2 - \xi_i\eta_i^2/2\}\,p_{a+b}(\xi_i)$, where $p_{a+b}$ is the P\'olya-Gamma $\mathrm{PG}(a+b,\,0)$ density. This joint recovers the $Z(a,b,0,1)$ marginal for $\eta_i$ and yields the Gaussian conditional $\eta_i \mid \xi_i \sim \Normal\!\bigl(\xi_i^{-1}(a-b)/2,\;\xi_i^{-1}\bigr)$, with mean zero in the horseshoe case $a = b = 1/2$ used here (Theorem~4 of \citealp{kowalDynamicShrinkageProcesses2019}, as corrected in \citealp{kowalCorrectionDynamicShrinkage2025}). This parameter expansion requires a draw from the full conditional of $\xi_i$, which is conditionally conjugate, $\xi_i \mid \eta_i \sim \mathrm{PG}(a+b,\,\eta_i)$,
and readily sampled using
\citet{polsonBayesianInferenceLogistic2013} implemented in
\texttt{BayesLogit::rpg}. Now, conditional on $\pmb \xi$, the GDSP innovations in \eqref{eq:prior-3} are Gaussian.

With these augmentations, we have $\tilde{\pmb\omega} \mid \bfh, \pmb z \sim \Normal_m(\bfh + \pmb\mu_z,\, \Sigma_v)$ for $\tilde{\pmb\omega} = (\tilde\omega_1,\ldots,\tilde\omega_m)^\top$,
$\pmb{\mu}_z = (m_{z_1},\ldots,m_{z_m})^\top$,
$\Sigma_v = \mathrm{diag}(v_{z_j}^2)$ and
$\Sigma_\xi = \mathrm{diag}(\xi_i^{-1})$, while \eqref{eq:prior-3} now contributes $\exp\{-\tfrac{1}{2}\,\bfh^\top(\Delta_k^{k_h+1})^\top
\Sigma_\xi^{-1}\Delta_k^{k_h+1}\bfh\}$. 
It remains to consider the global scale $h_0$ in \eqref{eq:h-priors} and the decomposition \eqref{eq:h-decomp}.  To facilitate efficient computing, consider first the full conditional precision without incorporating $h_0$: $Q_h = \Sigma_v^{-1} + \bigl(\Delta_k^{k_h+1}\bigr)^\top \Sigma_\xi^{-1}\,\Delta_k^{k_h+1}$, which shares a similar sparsity structure as for the trend precision in \eqref{eq:awol}. The full conditional distribution for $\bfh$, now accounting for the contributions from $h_0$ in \eqref{eq:h-priors}, is $\Normal_m\bigl((Q_h^{\star})^{-1}\bfell_h^{\star}, (Q_h^{\star})^{-1}\bigr)$,
with $Q_h^{\star} = Q_h + (s_0^2m^2)^{-1}\pmb{1}_m\pmb{1}_m^\top$ and $\bfell_h^{\star} =
  \Sigma_v^{-1}\bigl(\tilde{\pmb\omega} - \pmb{\mu}_z\bigr)
  + \muG/(s_0^2 m)\pmb{1}_m$ (Lemma~\ref{lem:h-fcd}). 
  
Finally, we sample $\bfh$ in two steps. This is essential because the global scale disrupts the sparsity of the full conditional precision: $Q_h$ inherits sparsity from $\Delta_k^{k_h+1}$, but $Q_h^\star$ does not. First, we apply a Cholesky decomposition to the sparse matrix $Q_h$, followed by forward-/back-solve operations with a sampling step like for the trend (Section~\ref{mcmc-trend}). Then, to correct for $Q_h^\star$, we apply a Sherman-Morrison update that recycles the Cholesky decomposition of $Q_h$ to provide the necessary rank-one update.  The explicit steps and verification of the sampling algorithm are in Appendix~\ref{sec:extra-mcmc}. 
Upon jointly sampling $\bfh$, we compute the global and local scales directly:  $\tau^2 = \exp(\hzero)$ with $\hzero = m^{-1}\bfone_m^\top\bfh$  and $\lambda_j^2 = \exp(h_j - \hzero) = \exp(h_j^\perp)$.

\subsection{Imputation and prediction}\label{mcmc-imp}
Suppose we do not observe data for vertices $\mathcal{I} \subset \{1,\ldots,n\}$. Imputation or prediction is straightforward and fast: the unobserved data have full conditional distribution
\begin{equation}\label{pred}
[\tilde y_i \mid \pmb y, - ] \stackrel{indep}{\sim} \mathcal{N}(\beta_i, \sigma_\epsilon^2), \quad i \in \mathcal{I}. 
\end{equation}
Sampling \eqref{pred} is easy, efficient, and parallelizable across all unobserved data. This imputation step is appended to the previous Gibbs sampling steps. The samplers for $\pmb \beta$, $\pmb h$, and $\sigma_\epsilon^2$ proceed exactly as in the complete data case, now using the imputed data  $\pmb y = (\pmb y_{-\mathcal{I}}, \pmb{\Tilde{y}}_{\mathcal{I}})$  updated at each MCMC iteration.

\section{Simulations}
\label{sec:sim-data}
We conduct a series of simulation studies to compare among frequentist and Bayesian approaches for graph trend filtering. We evaluate 1) trend estimation, which is the primary goal for frequentist methods;  2) trend uncertainty quantification, especially to compare different shrinkage priors; and 3) computational efficiency. We consider both complete datasets and datasets with missing values. 

\subsection{Simulation design: lattice data}
\label{sec:simulation-design}
We generate dependent (image) data on a $d_1 \times d_2$ lattice subject to measurement error and missingness. The true trends $\pmb \beta$ are displayed in Figures~\ref{fig:sim_lattice}~and~\ref{fig:true-mean-funcs}; see Appendix~\ref{sec:extra-comp} (Table~\ref{tab:kernel-choices}) for the explicit forms. These trends include: \texttt{bg} (short for bivariate Gaussian kernel), which is smooth yet nonlinear; \texttt{blocks} and \texttt{blocks+}, which are piecewise constant and discontinuous but with different patterns; and \texttt{lin}, which is  linear.
The lattices have $d_1 = d_2 = 50$, so $n=2500$ vertices and $m=4900$ edges. Given a trend $\bfbeta$, each simulated dataset is generated from \eqref{eq:likelihood}. We set $\sigma_\epsilon = \mbox{sd}(\pmb \beta)/\mbox{RSNR}$, where $\mbox{sd}(\pmb \beta)$ is the sample standard deviation of $\pmb \beta$ and RSNR~$=3$ is a pre-specified root-signal-to-noise ratio. We repeat this process for each trend to create 100 synthetic datasets. Here, we report results under substantial missingness with $50\%$ of vertices uniformly at random removed (see Appendix~\ref{sec:extra-comp} for the complete data case). 
Subsequently, we extend this simulation design for a time series of images (Section~\ref{sec:st_sims}). We performed all simulations on two Nvidia DGX clusters running Ubuntu Linux in Singularity containers.

\begin{figure}[h]
  \centering
  \includegraphics[width=\linewidth]{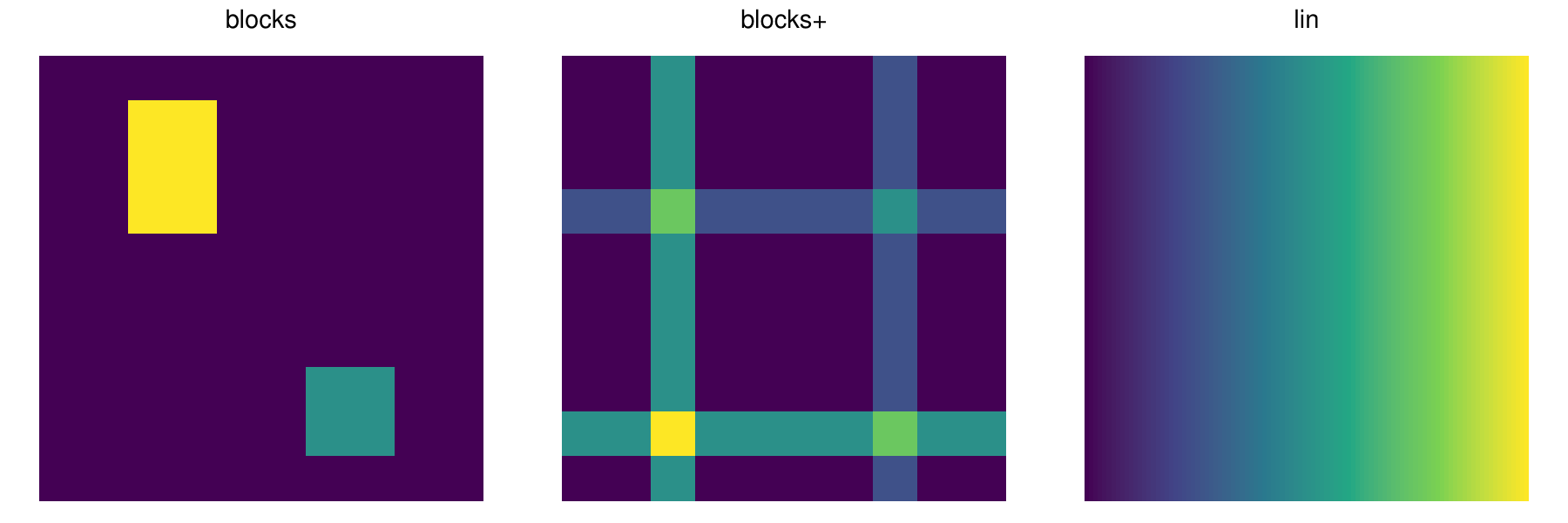}
  \caption{True $\pmb\beta$ on a $50 \times 50$ lattice for  \texttt{blocks}, \texttt{blocks+} and \texttt{lin}; see Figure~\ref{fig:sim_lattice} for \texttt{bg}.}
  \label{fig:true-mean-funcs}
\end{figure}

\subsection{Competing methods}
\label{sec:competing-methods}
We compare the proposed approach against Bayesian and frequentist alternatives for graph trend filtering. All Bayesian methods use the likelihood \eqref{eq:likelihood}, the same graph difference operator \eqref{graph-diff-op} with $k=1$ (see Appendix~\ref{sec:extra-comp} for $k \in\{ 0, 1,2\}$ comparisons),
the graph-informed prior on the trend given by \eqref{gdo-apply} and \eqref{eq:prior}--\eqref{eq:prior-lam}, and 
the error-variance prior \eqref{eq:sigma-prior} with $a_\sigma = b_\sigma = 0.01$. Differences only occur in the prior on the shrinkage parameters, $\pi(\pmb\lambda)$. ``NIG" adopts a Normal-Inverse-Gamma prior with  $ \lambda_j^2 = \lambda^2 \sim \mathrm{Inverse\text{-}Gamma}(0.01, 0.01)$ and $\tau=1$. NIG resembles (second order) ICAR models and incorporates graph information for modeling the trend, but without aggressive shrinkage or sparsity.  ``BL" is the Bayesian lasso \citep{parkBayesianLasso2008} and ``HS" is the horseshoe prior \citep{carvalhoHorseshoeEstimatorSparse2010}. HS does not use any graph information for $\pi(\pmb \lambda)$ and thus is a special case of the GDSP with $k_h = -1$.  ``GDSP" refers to the proposed horseshoe GDSP  ($a = b = 1/2$). We use  $k_h=0$ for first-order differences for the local (log) shrinkage and $\mu_0 = 0, s_0 =1$ so $\hzero \sim \Normal\bigl(1/2,\; 1\bigr)$ for the global (log) shrinkage.   The MCMC implementations of these Bayesian models share many of the same Gibbs sampling steps, including for the trend (Section~\ref{mcmc-trend}), the error variance,  and prediction/imputation  (Section~\ref{mcmc-imp}). We run each MCMC algorithm for  $N + N_{\mathrm{burn}} = 15000$ iterations,
where the initial $N_{\mathrm{burn}} = 7500$ simulations are discarded as a burn-in.

Finally, we include the frequentist graph trend filtering estimator \eqref{gtf-est},  and specifically the fused lasso variant  ($k=0$) recommended by \citet{wangTrendFilteringGraphs2016}. As a favor to this method, we select $\lambda$ via an oracle: namely, we use $\lambda$ for which the estimated trend $\pmb{\hat\beta}_\lambda$ has the smallest root mean squared error (RMSE) for estimating the true trend $\pmb\beta$. This strategy is not feasible in practice, but represents a best case scenario for frequentist graph trend filtering. It also circumvents any computationally intensive methods to select $\lambda$, such as cross-validation. 
We fit this model using the \texttt{genlasso} package in \texttt{R} and thus denote it by ``GL". Since GL does not automatically handle missingness, we impute missing data using the sample mean among observed vertices.
Appendix~\ref{sec:extra-comp}  includes a simulation study without missingness, so the GL performance there does not depend on the imputation method. 

\subsection{Evaluating model accuracy and uncertainty quantification}
\label{sec:eval-model-accur}
We assess the performance of each method by considering both point estimates and
intervals for the true trend ${\beta}_{i}$, evaluated on a set
$S \subseteq V$: the held-out vertices when data are missing, and all
vertices in the complete case. For point
estimation, we compute the RMSE 
\begin{equation*}
  \mathrm{RMSE} = \sqrt{\frac{1}{|S|}\sum_{i \in S}\big(\hat{\beta}_{i} - \beta_{i}\big)^2}
\end{equation*}
where $\hat \beta_i$ is the posterior mean $\mathbb{E}(\beta_i \mid \pmb y)$ for Bayesian methods or \eqref{gtf-est} for GL. Given an interval estimate $(\hat \ell_i, \hat u_i)$ for the trend, we compute both the mean credible interval widths and the empirical coverages:
\[
  \mathrm{MCIW} = \frac{1}{|S|}\sum_{i \in S} (\hat u_i - \hat \ell_i), \quad  \mathrm{ECP} = \frac{1}{|S|}\sum_{i \in S} \mathbf{1}\{ \hat \ell_i \leq \beta_{i} \leq \hat u_i\}
\]
We use  95\% posterior (equal-tailed) credible intervals. Intervals that are calibrated ($\mbox{ECP} \ge 0.95$) and
precise (small MCIW) are preferred.  GL does not produce interval estimates and thus is excluded from this metric. 

\begin{figure}[h]
  \centering
  \includegraphics[width=0.49\linewidth,trim=2 2 2 2,clip]{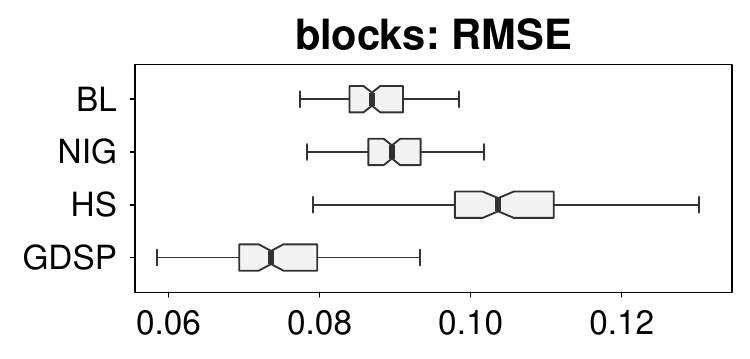}\hfill
  \includegraphics[width=0.49\linewidth,trim=2 2 2 2,clip]{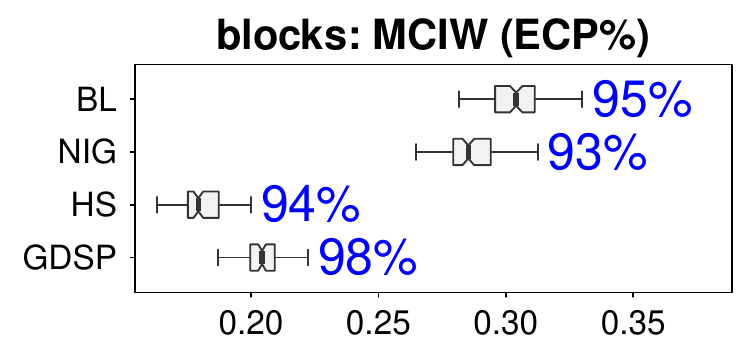}

  \includegraphics[width=0.49\linewidth,trim=2 2 2 2,clip]{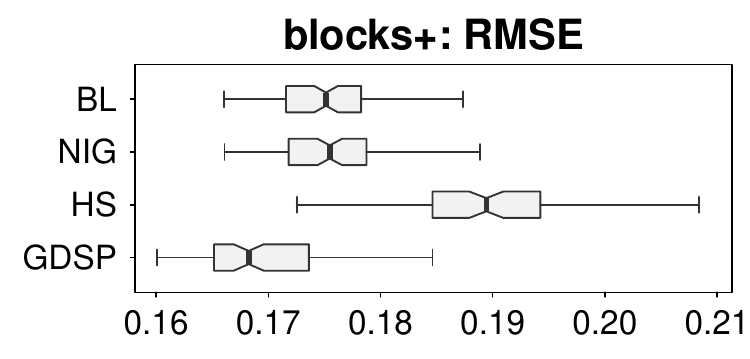}\hfill
  \includegraphics[width=0.49\linewidth,trim=2 2 2 2,clip]{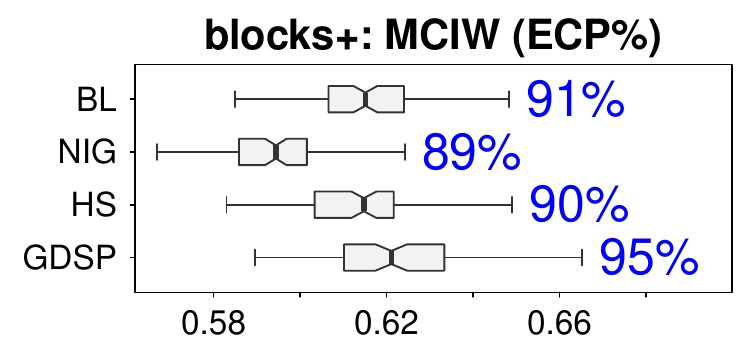}

    \includegraphics[width=0.49\linewidth,trim=2 2 2 2,clip]{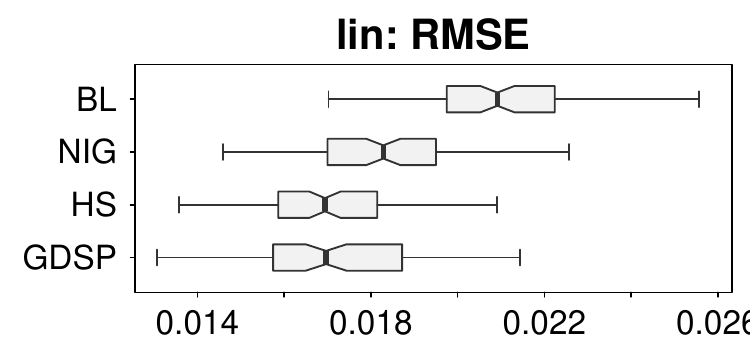}\hfill
  \includegraphics[width=0.49\linewidth,trim=2 2 2 2,clip]{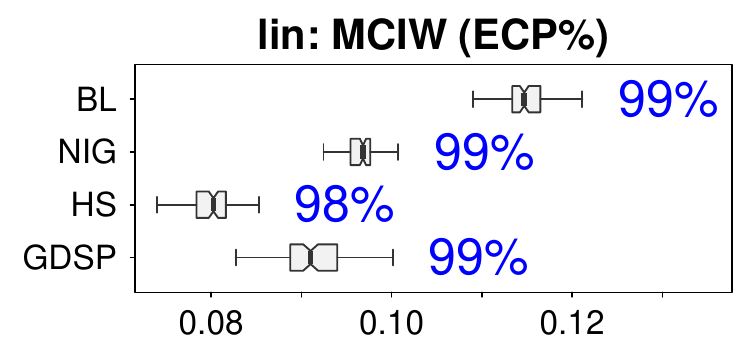}

  \includegraphics[width=0.49\linewidth,trim=2 2 2 2,clip]{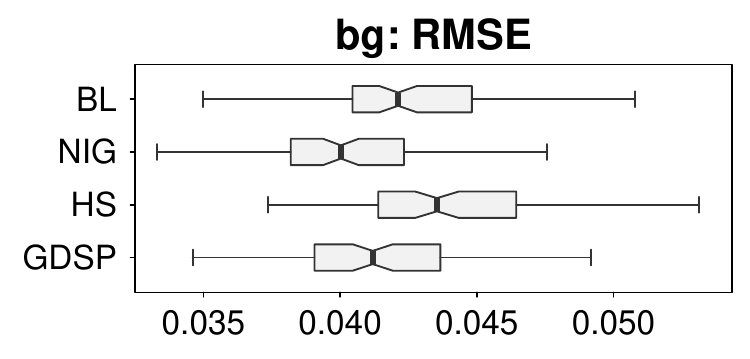}\hfill
  \includegraphics[width=0.49\linewidth,trim=2 2 2 2,clip]{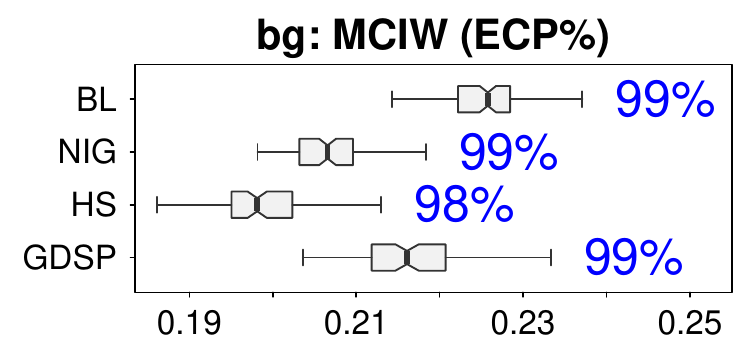}
  \caption{RMSEs and MCIWs with ECPs (blue annotations) for the four trends with 50\% missingness, computed over the held-out vertices; each panel title reports the trend and the metric. GL performs poorly and is omitted (median RMSEs 0.189 for \texttt{blocks}, 0.350 for \texttt{blocks+}, 0.172 for \texttt{lin}, and 0.316 for \texttt{bg}). The proposed horseshoe GDSP is the most accurate for the non-smooth trends and the only one whose intervals stay near the nominal level at \texttt{blocks+}; for the smooth trends the Bayesian priors perform comparably.}
  \label{fig:sims}
\end{figure}

We report the RMSEs, MCIWs, and ECPs in Figure~\ref{fig:sims}
for the case with $50\%$ missingness. First, the Bayesian  estimators significantly outperform the frequentist estimator, despite the latter's oracle selection of its tuning parameter. This is partly due to the inability of GL to handle missing data. Second, the Bayesian models are mostly well-calibrated, with coverages that typically achieve the nominal level. Thus, Bayesian models for graph trend filtering can offer reliable inference for the trend, whereas the frequentist GL does not provide interval estimates. Third, the choice of shrinkage prior is highly impactful for both estimation and inference, and the differences are sharpest for the most challenging, non-smooth trends. For \texttt{blocks} and \texttt{blocks+}, the horseshoe GDSP is the most accurate by a wide margin (e.g., median RMSE $0.074$ versus $0.087$--$0.104$ for the Bayesian competitors on \texttt{blocks}), and it is the only prior whose intervals stay near the nominal level at \texttt{blocks+} (ECP $94.8\%$), while the competitors undercover (ECP $ \le 91\%$). The narrower intervals from HS or NIG in that case simply reflect this undercoverage. For the smooth trends (\texttt{lin} and \texttt{bg}), the Bayesian estimates offer similar accuracy
while HS has slightly narrower calibrated intervals. Perhaps surprisingly, the aggressive shrinkage of the horseshoe alone does not deliver accuracy on the non-smooth trends: HS has the \emph{worst} point estimation among the Bayesian priors on \texttt{bg}, \texttt{blocks}, and \texttt{blocks+}. Thus, both the aggressive shrinkage \emph{and} the graph-dependent shrinkage of the horseshoe GDSP are essential where trend filtering is hardest.

For completeness, we also compute WAIC \citep{watanabeAsymptoticEquivalenceBayes2010} to compare  $k \in\{ 0, 1,2\}$ (Appendix~\ref{sec:extra-comp}). The trends for $k=2$ are much smoother but less adaptive than for $k=0$; we find that $k=1$ provides a good balance.

\subsection{Evaluating MCMC efficiency and computational scalability}
\label{sec:eval-mcmc-effic}
A natural concern is whether these Bayesian graph trend filtering models are computationally viable. In particular, the  GDSP introduces another layer of dependencies and our MCMC algorithm requires multiple data augmentations that are not present for iid shrinkage priors. Here, we evaluate the computational performance of each method. As a benchmark, we include the frequentist estimator GL from \cite{wangTrendFilteringGraphs2016}, but this computing time does \emph{not} include any data-driven selection method for the tuning parameter $\lambda$, such as cross-validation. Instead, we record the time necessary to compute the solution path from the fully regularized solution ($\lambda$ large) to fully dense estimator ($\lambda \approx 0$), 
noting that the complexity of the path grows with $n,m$, and the topology of the graph.

For the MCMC algorithms, we consider both computing time (i.e., wall-clock time) and MCMC efficiency. The latter is measured by the effective sample size $N_{\mathrm{eff}}(\beta_i)$ from $N$ MCMC draws of each node $i \in V$. We compute \emph{average relative efficiency}, 
\begin{equation*}
  \frac{\bar{N}_{\mathrm{eff}}}{N} = \frac{1}{n}\sum_{i \in V}\frac{N_{\mathrm{eff}}(\beta_{i})}{N}.
\end{equation*}
Values that approach one achieve (iid) Monte Carlo efficiency, while values near zero indicate substantial autocorrelations in the MCMC chain. Next, we estimate the (wall-clock) \emph{time to 1000 effective samples}, which jointly considers computing time and MCMC efficiency:
\begin{equation*}
  \bar{s}_{1000} = \frac{1}{n}\sum_{i \in V}\left\{s_{N_{\mathrm{burn}}} + s_{N}\frac{1000}{N_{\mathrm{eff}}(\beta_{i})}\right\},
\end{equation*}
where  $s_{N_{\mathrm{burn}}}$ and $s_{N}$ are the (wall-clock) time needed to complete the burn-in and MCMC draws, respectively. By design, $\bar{s}_{1000}$ enables comparisons among MCMC algorithms that are fast yet MCMC-inefficient with those that may be slow  but achieve greater MCMC efficiency. The total of 1000 effective samples is a reasonable target for algorithm completion, and thus comparable to the frequentist GL estimator.

We report the computational comparisons in Appendix~\ref{app-compute}  for the simulations and in Section~\ref{sec:us_unemployment} for the real data. Briefly, we note that 1)  Bayesian graph trend filtering models are often \emph{more efficient} than frequentist GL estimation, 2) the more aggressive shrinkage priors (HS, GDSP) sacrifice some MCMC efficiency relative to simpler (NIG) alternatives, and 3)  the horseshoe GDSP, which includes more (shrinkage) parameters and strong dependencies among them, typically sacrifices some MCMC efficiency relative to its Bayesian competitors, but remains computationally viable.

\subsection{Simulation design: time series of images}\label{sec:st_sims}

A primary advantage of (Bayesian or frequentist) graph trend filtering models is that they
are readily applicable for a broad variety of dependent data
scenarios. Once the graph is defined, our competing (GL, BL,
NIG, HS) and proposed (GDSP) models and algorithms apply immediately
without any further customization. Here, we showcase this feature while
investigating how each method performs under more complex
dependencies.

We extend the simulation design from
Section~\ref{sec:simulation-design} to consider a time series of
images, $\bfbeta_t \in \mathbb{R}^n$, $t=1,\ldots,T$. The baseline
trend $\mathbf{b}_0 \in \mathbb{R}^n$ is one of the four static
lattice trends \texttt{bg}, \texttt{blocks}, \texttt{blocks+},
\texttt{lin} from Section~\ref{sec:simulation-design}; each choice
yields a separate simulation study. The time series is initialized at
the baseline ($\bfbeta_1 = \mathbf{b}_0$) and then evolves as a
combination of autoregressive and spatial-diffusion dynamics:
\begin{equation}
  \label{eq:heat-equation}
      \bfbeta_t = \rho\,\bfbeta_{t-1} - \kappa\, \Delta^{(2)} \bfbeta_{t-1} + (1 - \rho)\,\mathbf{b}_t + \boldsymbol{\nu}_t, \quad t=2,\ldots,T,
  \end{equation}
  where $\rho \in [0,1]$ is the temporal autocorrelation coefficient,
  $\kappa > 0$ controls the strength of spatial diffusion,
  $\Delta^{(2)}$ is the graph Laplacian of the lattice,
  $\boldsymbol\nu_t \sim^\mathrm{iid} \mathcal{N}(\mathbf{0},
  \sigma_{\mathrm{AR}}^2 \mathbf{I})$ is the process noise, and $\mathbf{b}_t = \mathbf{b}_0\{1 + \alpha (t-1)/(T-1)\}$ 
  is a time-varying target trend that drifts linearly away from the
  baseline, from $\mathbf{b}_1 = \mathbf{b}_0$ at $t=1$ to
  $\mathbf{b}_T = (1+\alpha)\mathbf{b}_0$ at $t=T$.
  The AR(1) term $\rho\bfbeta_{t-1}$ introduces
  temporal persistence: each image is a damped version of the previous
  one. The diffusion term $-\kappa\Delta^{(2)}\bfbeta_{t-1}$
  introduces spatial smoothing: the $i$th entry of
  $\Delta^{(2)}\bfbeta_{t-1}$ is
  $n_{i}\beta_{i,t-1} - \sum_{j\sim i}\beta_{j,t-1}$, with $D_{ii}$ the
  spatial degree of vertex $i$, so subtracting it
  pulls $\beta_{i,t}$ toward its local (spatial) average. This is the
  discrete analogue of the continuous heat equation
  $\partial_t u = \kappa\nabla^2 u$, which smooths a surface by
  redistributing values toward their spatial surroundings. Observed
  data $\bfy_t$ are then generated from \eqref{eq:likelihood} as in Section~\ref{sec:simulation-design}.  
  Both the
  trend and the observed data are visualized in
  Figure~\ref{fig:spatiotemporal_evolution}.
 
\begin{figure}[ht]
  \centering
\includegraphics[width=1\linewidth]{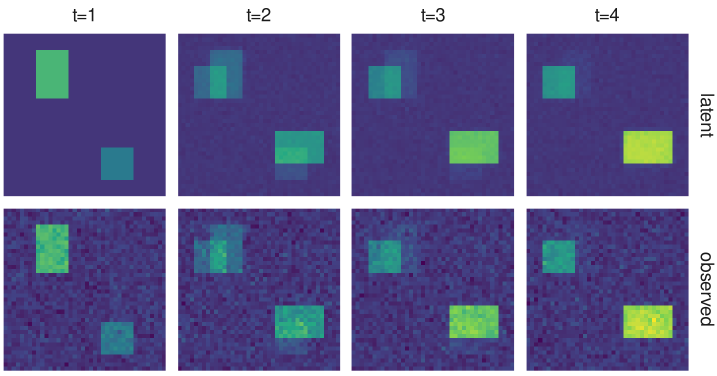}
  \caption[AR(1) blocks example.]{Spatio-temporal evolution of a square
    lattice ($d_1 = d_2 = 40$) under \eqref{eq:heat-equation} with the \texttt{blocks} initial trend. Each node evolves via AR(1) temporal dynamics and heat-diffusion spatial smoothing for $T=4$.}
\label{fig:spatiotemporal_evolution}
\end{figure}

We consider square lattices with $d_1 = d_2 = 40$ over $T = 4$
timesteps. The spatio-temporal graph is constructed by connecting
spatial neighbors on the lattice at each time and linking each (lattice) node at time $t$ to its counterpart at time $t+1$. This yields
$d_1 d_2  T = 6400$ nodes and $17280$ edges in total. It fully
determines the graph difference operator \eqref{graph-diff-op} shared
among all competing methods. We set $\rho = 0.4$, $\kappa = 0.03$, $\alpha = 0.6$,
$\sigma_{\mathrm{AR}} = 0.03$,  and
$\mathrm{RSNR} = 3$, matching the lattice design of Section~\ref{sec:simulation-design}. We hold out the entire final time slice ($t = T$) during training, so that the trend at $t = T$ is fully forecast using only the preceding time points. All Bayesian methods use $N_{\mathrm{burn}} = 7500$, $N = 10000$, and otherwise the same specifications as in Section~\ref{sec:competing-methods}.  For GL, we impute the entries of the held-out final time slice with the grand mean of the observed entries and fit the fused lasso (\texttt{genlasso::fusedlasso}, $k=0$) on the full spatio-temporal graph. We again select $\lambda$ by oracle: the value of $\lambda$ that minimizes RMSE against the (held-out) true trend at $t=T$. The fitted values at $t=T$ are then taken as the GL forecast.

Figure~\ref{fig:st-summary} summarizes the forecast RMSEs and MCIW/ECP at $t = T$ across 100 simulated datasets for \texttt{blocks} and \texttt{lin} trends (see Appendix~\ref{sec:extra-comp} for additional trends). The results  extend the patterns observed in the (time-invariant) lattice case (Section~\ref{sec:eval-model-accur}): the proposed horseshoe GDSP delivers accurate point forecasts and precise yet well-calibrated uncertainty quantification for the trend. For \texttt{blocks}, the horseshoe GDSP  attains the lowest forecast RMSE by a clear margin and is competitive with NIG and BL for the smoothing \texttt{lin} trend. In both cases, the horseshoe GDSP intervals are much narrower yet still achieve the nominal coverage. Notably, the gains over the horseshoe prior remain substantial for both estimation and inference, which again emphasizes the benefits of graph-dependent shrinkage.

\begin{figure}[h]
\centering
\includegraphics[width=\linewidth]{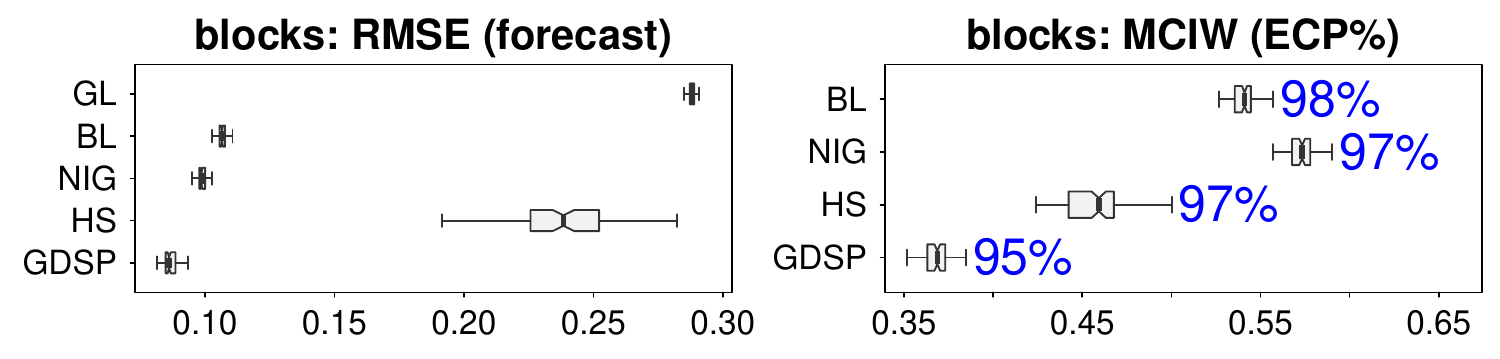}
\vspace{2pt}
\includegraphics[width=\linewidth]{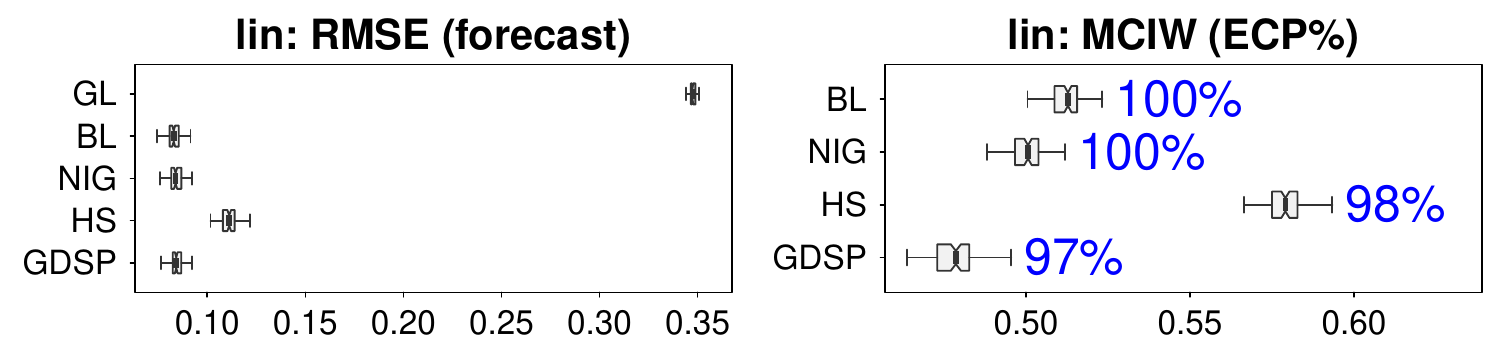}
\caption{Forecast RMSEs (left) and MCIWs (right) with ECPs (blue annotations) at $t = T$ for the \texttt{blocks} (top) and \texttt{lin} (bottom) spatio-temporal trends ($d_1=d_2 = 40$, $T = 4$, RSNR $= 3$). In both cases, the horseshoe GDSP delivers highly competitive point forecasts and the narrowest (calibrated) forecast intervals.}
\label{fig:st-summary}
\end{figure}

\FloatBarrier
\section{Modeling and Forecasting U.S.\ County-Level Unemployment}\label{sec:us_unemployment}

The U.S.\ Bureau of Labor Statistics publishes local area (county-level) unemployment data each month. Unemployment is a key macroeconomic indicator: together with inflation, it forms the basis of the Phillips curve, a central relationship in monetary policy. County unemployment data, rather than aggregated national or state data, preserves the spatial heterogeneity driven by local labor markets, policies, and demographics. We focus on county-level unemployment during April through July 2020, which spans the complete arc of the sharp COVID-19 unemployment shock: the April surge, the May-June trough and turnaround, and the July continuation of the recovery. The unemployment trends during this period exhibit substantial spatial and temporal heterogeneity. Our goal is to apply  several graph trend filtering models and evaluate their (spatial) imputation, (temporal) forecasting, and computational efficiency. 

We define our spatio-temporal graph on the county- and time-indexed nodes as follows. First, the spatial graph operates on the 3108 counties in the continental U.S.\ (excluding Alaska and Hawaii) and connects counties that share a physical border. The average degree is 5.94 (standard deviation 1.31). The most connected county is
San Juan County, UT (degree 14), while several counties, including
island counties such as Nantucket County, MA, have degree~1. The
spatial graph is sparse: the edge density is 0.00191, compared to 0.00125 for
a square lattice with the same number of nodes.
The spatio-temporal
graph is constructed by stacking the spatial graphs for April through July 2020 with temporal edges that connect
each county in month $t$ to the same county in month $t+1$.  The
resulting spatio-temporal graph has $n = 12432$ nodes.

Using this spatio-temporal graph, we apply several graph trend filtering models for the  log-unemployment rates.  We construct two out-of-sample tasks.
First, an \emph{imputation} task: we hold out $25\%$ of the $9324$ county-months from 
April to June 2020,
selected uniformly at random and identically for every method. Each model must reconstruct the held-out
values from the remaining observations. These targets typically have observed
neighbors in space and in time, so this task isolates each method's ability
to borrow information across the graph. Second, we \emph{forecast} the
entire July 2020 surface. Here, we observe the complete spatio-temporal data of the first COVID-19 wave (April to July) and must extrapolate the recovery. Both tasks are extremely challenging due to the spatial heterogeneity  (see Figure~\ref{fig:us_forecast}), the rapid temporal changes, and the short time series for training ($T=3$).

We use the same competing methods from Section~\ref{sec:competing-methods}.
Bayesian imputations and forecasts are generated as in Section~\ref{mcmc-imp}.  We run the MCMC for $N = 50000$ iterations after discarding the initial $N_{\mathrm{burn}} = 20000$ iterations as a burn-in, storing every $10$th draw ($5000$ saved samples).
Trace plots are reported in Appendix~\ref{app-data}.
For GL, which provides no native imputation, we first fill the
held-out training county-months with each county's mean over its observed
months, fit the fused
lasso on the spatio-temporal graph spanning April-June, and select
$\lambda$ along the solution path by oracle. For the imputation task, the oracle is the value of $\lambda$ that minimizes the RMSE on the held-out set of county-months. For the prediction task, the oracle minimizes the RMSE on the held-out July data. No data-driven choice of $\lambda$ could do
better on either task. The training graph contains no July nodes, so GL also
requires an explicit forecasting rule: with the piecewise-constant trend
implied by $k = 0$, each county's July forecast is its estimated June value.

To visualize these tasks, we present county-level maps in Figure~\ref{fig:us_forecast} for imputation   and forecasting. The GDSP accurately recovers the spatial patterns, including both broad regional gradients
and finer county-level variation. The GDSP forecasts must rely primarily on temporal connections, which produces a smoother map. A complementary
county-profile view, with all $3108$ counties ordered by FIPS code and the
GDSP intervals for every month of the window, is reported in the
appendix (Figure~\ref{fig:us_profile}).

\begin{figure}[t!]
  \centering
    \includegraphics[width=0.92\linewidth]{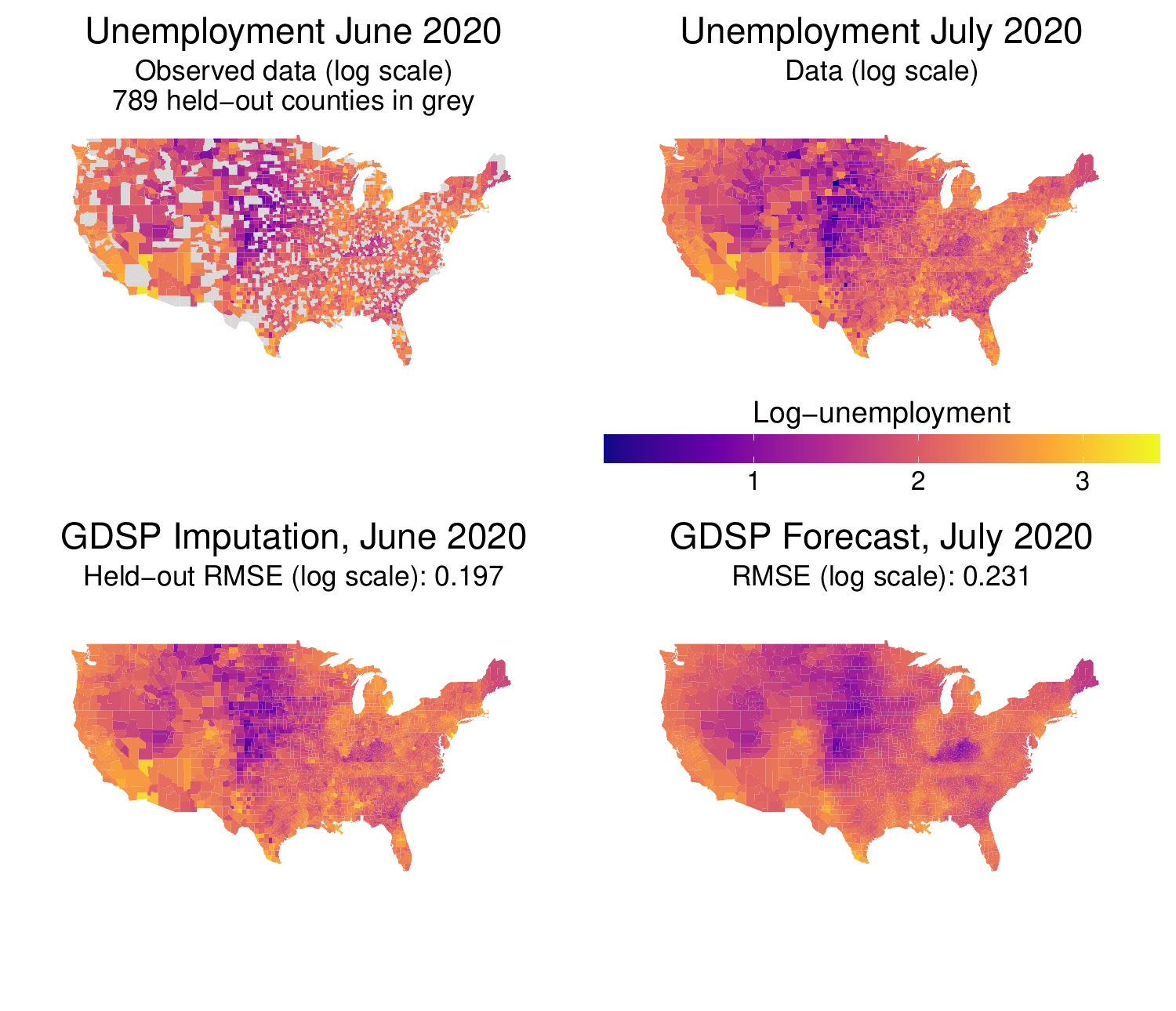}
  \caption{County-level maps of  U.S.\ log-unemployment to illustrate the imputation (left) and forecasting (right) tasks. Observed data (top) is presented for both tasks, along with the point predictions and forecasts of the proposed GDSP (bottom).}
  \label{fig:us_forecast}
\end{figure}

We evaluate each method on the two held-out sets (Table~\ref{tab:us-imputation}). Since these are observable data, we report mean prediction interval widths (MPIW) computed similarly as for the trend in Section~\ref{sec:eval-model-accur}.
On the imputation task, the horseshoe
GDSP attains the lowest imputation RMSE
and improves on the frequentist GL by $20\%$, even though the GL tuning parameter is selected by oracle using the held-out data. The GDSP intervals
are also the narrowest of all methods, with empirical coverage at the
nominal level ($95.5\%$). 
On the forecasting task, several Bayesian competitors (NIG, BL, GDSP) perform similarly for point forecasts and outperform the remaining methods (GL, HS). Here, NIG offers the best interval forecasts, which suggests that aggressive shrinkage is not necessary for temporal neighbors on this graph. This is reflected in Figure~\ref{fig:us_forecast} (bottom right), where the GDSP forecast is spatially smooth. 
Still, each Bayesian method attains the nominal coverage. 
Yet again, the horseshoe GDSP offers massive improvements over the HS for both point prediction and uncertainty quantification, thus emphasizing the importance of \emph{graph-dependent} shrinkage. 

Finally, the Bayesian methods are computationally competitive with GL: the
horseshoe GDSP sampler runs in about $1.4$ times the wall-clock time of GL's
full solution path ($3.2$ against $2.3$ hours),
where the latter does \emph{not} include the time to select the GL tuning parameter. Even so, NIG and HS are faster than the GL path. These Bayesian methods   provide full posterior and predictive inference with no further tuning parameters to select. Appendix~\ref{app-data} the county-level profile of the fits and intervals across all $3108$ counties and the MCMC diagnostics (trace plots and effective sample sizes). 

\begin{table}[t!]
\centering\small
\begin{tabular}{lcccc}
\toprule
 & \multicolumn{2}{c}{Imputation} & \multicolumn{2}{c}{Forecast} \\
\cmidrule(lr){2-3}\cmidrule(lr){4-5}
 & RMSE & MPIW (ECP) & RMSE & MPIW (ECP) \\
\midrule
GDSP & {\bf 0.199} ($\pm$0.004) & {\bf 0.827} (95.5\%) & 0.231 ($\pm$0.005) & 1.251 (98.5\%) \\
BL   & 0.201 ($\pm$0.004) & 0.835 (95.2\%) & 0.226 ($\pm$0.005) & 1.233 (98.0\%) \\
NIG  & 0.203 ($\pm$0.004) & 0.850 (95.4\%) & {\bf 0.223} ($\pm$0.005) & {\bf 1.176} (97.9\%) \\
HS   & 0.231 ($\pm$0.007) & 0.975 (96.1\%) & 0.291 ($\pm$0.008) & 1.316 (96.6\%) \\
GL (oracle) & 0.248 ($\pm$0.005) & --- & 0.242 ($\pm$0.005) & --- \\
\bottomrule
\end{tabular}
\caption{RMSEs, MPIW, and ECP for imputation of $2331$ held-out county-months (left) and RMSEs, MPIW, and ECP for forecasting $3108$ counties in July 2020 (right). RMSE brackets
report $\pm 1$ standard error under resampling of the held-out entries. GL is tuned by oracle for each task and does not produce interval estimates. Bolded entries are the smallest RMSE and MPIW values (subject to ECP $\ge 95\%$).}
\label{tab:us-imputation}
\end{table}

\section{Conclusions}\label{sec:conclusions}

We introduced a Bayesian modeling and computational framework for trend filtering on a graph, thus providing a general approach for smoothing, inference, imputation, and prediction with dependent data. Our approach leveraged the  graph at three junctures: 1) the trend, for smoothing and imputation; 2) the local shrinkage, for greater adaptivity and more precise inference; and 3) the MCMC algorithm, to design more efficient sampling steps. The proposed graph-dependent shrinkage prior (GDSP) extended well-established continuous shrinkage priors to dependent data settings via graph difference operations on the log-shrinkage parameters. We showed how the GDSP maintains the appealing shrinkage properties of its iid counterparts, but critically incorporates the graph to borrow shrinkage information across connected nodes. Like popular (iid) global-local shrinkage priors, the GDSP retains a global scale parameter, which we used both to infer the global smoothness across the graph and to ensure posterior propriety. For both the graph-dependent trend and the graph-dependent shrinkage parameters, we developed MCMC sampling steps that exploited sparse and banded matrix operations to facilitate scalable posterior inference. 

Empirically, the horseshoe GDSP outperformed frequentist and Bayesian alternatives across simulated lattice data, spatio-temporal image sequences, and county-level U.S.\ unemployment imputation and forecasting. The gains were most pronounced for challenging trends, including non-smooth signals and real, local area unemployment data, where the interplay between graph-informed smoothness and adaptive shrinkage is most critical. The horseshoe GDSP offered superior point and interval estimates in nearly all cases, often with substantial improvements over the iid horseshoe, thus emphasizing the importance of \emph{graph-dependent} shrinkage. On the U.S.\ unemployment dataset with $9324$ observed county-month nodes, the fully Bayesian graph trend filtering methods were computationally competitive with the frequentist point estimator, even without including the cost of selecting the frequentist tuning parameter. 

There are several directions for future work. Although 
\eqref{eq:likelihood} assumes Gaussian errors, it is straightforward to extend Bayesian graph trend filtering models for non-Gaussian data, including binomial,  negative-binomial, and  mixtures of Gaussians. Data augmentation sampling algorithms exist for each of these cases, which would preserve the main sampling steps from our approach (Section~\ref{sec:mcmc}). Alternatively, given node-specific covariates $x_i$,  one might modify \eqref{eq:likelihood} for graph-dependent regression, such as $y_i = x_i^\top\beta_i + \epsilon_i$, as a  generalization of spatially-varying coefficient or time-varying parameter regression models. Finally, it is possible to incorporate node- or edge-level covariates into the trend or the shrinkage prior, lending further adaptability to these critical graph-informed quantities. 

\acks{Kowal gratefully acknowledges funding from the National Science Foundation under award number SES-2435310.}

\bibliography{reference}
\clearpage

\clearpage
\appendix

\setcounter{theorem}{0}
\renewcommand{\thetheorem}{\Alph{section}.\arabic{theorem}}
\makeatletter\@addtoreset{theorem}{section}\makeatother

\setcounter{lemma}{0}
\renewcommand{\thelemma}{\Alph{section}.\arabic{lemma}}
\makeatletter\@addtoreset{lemma}{section}\makeatother

\counterwithin{figure}{section}
\counterwithin{table}{section}
\counterwithin{equation}{section}

\section{Mathematical details and proofs of theorems}
\label{sec:proofs}

We first prove Theorems~\ref{thm:cond-prior}--\ref{thm:posterior-shrinkage} on the prior concentration and posterior adaptivity properties, and then prove Theorem~\ref{thm:posterior-proper}
on posterior propriety result. The modifications for 
$k_h = 1$ results, collected in Theorem~\ref{rem:kh1-prior}, are deferred to
Appendix~\ref{sec:kh1} and hold under the same standing conventions with
$k_h = 1$ in place of $k_h = 0$. Throughout, the vector $\bfh$ is edge-indexed for $k$ even or zero,
with length $m = |E|$, and vertex-indexed for $k$ odd, with length $n$.
The proofs are general: we use the notation $m$ for the length of $\bfh$ with the understanding that  $m=n$ should be substituted for when $k$ is odd.

\subsection*{Proofs of Theorems~\ref{thm:cond-prior}--\ref{thm:posterior-shrinkage}}
For $k_h = 0$ the innovations are the first differences across the edges of
the shrinkage graph, so each $h_i$ enters the prior only through comparisons
with its neighbors; as in Section~\ref{sec:graph_shrink}, vertex $i$ is non-isolated ($n_i \geq 1$). Throughout
Theorems~\ref{thm:cond-prior}--\ref{thm:posterior-shrinkage}, all
conditionals are those of the graph-dependent component, the difference
relation \eqref{eq:prior-3} alone, and the global scale
is fixed at $\tau^2 = 1$, as stated in Section~\ref{sec:graph_shrink}, so that
$h_i = \log\lambda_i^2$ and $\kappa_i = 1/(1+e^{h_i}) \in (0,1)$, i.e.\
$e^{h_i} = (1-\kappa_i)/\kappa_i$.

\begin{proof}[Theorem~\ref{thm:cond-prior}]
Recall that \eqref{eq:prior-3} is read as the product density over the
independent edge innovations. Since $a = b$, the $Z(a,a,0,1)$
density $p_Z(z) = \{B(a,a)\}^{-1}e^{az}(1+e^z)^{-2a}$ is symmetric. Every edge at $i$ can be written
as $h_i - h_j$, giving
$\pi(h_i \mid \bfh_{-i}) \propto \prod_{j \sim i} p_Z(h_i - h_j)$.
Substituting
$e^{h_i} = (1-\kappa_i)/\kappa_i$ and $e^{h_j} = (1-\kappa_j)/\kappa_j$,
\[
  e^{a(h_i-h_j)} = \frac{(1-\kappa_i)^a\,\kappa_j^a}{\kappa_i^a\,(1-\kappa_j)^a},
  \qquad
  1 + e^{h_i - h_j}
   = \frac{\kappa_i(1-\kappa_j) + (1-\kappa_i)\kappa_j}{\kappa_i(1-\kappa_j)},
\]
so each factor $p_Z(h_i - h_j)$ is proportional to $\{\kappa_i(1-\kappa_i)\}^a\,\{\kappa_j(1-\kappa_j)\}^a
\bigl\{\kappa_i(1-\kappa_j)+(1-\kappa_i)\kappa_j\bigr\}^{-2a}$, symmetric in
$i$ and $j$ as it must be.
Multiplying over $j \sim i$, discarding factors free of $\kappa_i$, and
applying the Jacobian $|\dd h_i/\dd\kappa_i| = \{\kappa_i(1-\kappa_i)\}^{-1}$
yields \eqref{eq:cond-prior}.
\end{proof}

\begin{lemma}[Drift bound]\label{lem:drift}
Let $f$ be a probability density on $\R$ and $r > 0$. If
$(\log f)' \geq r$ on $(-\infty, t\,]$, then
$\int_{-\infty}^{t-M} f(x)\,\dd x \leq e^{-rM}$ for every $M \geq 0$.
Symmetrically, if $(\log f)' \leq -r$ on $[\,t, \infty)$, then
$\int_{t+M}^{\infty} f(x)\,\dd x \leq e^{-rM}$.
\end{lemma}

\begin{proof}[Lemma~\ref{lem:drift}]
The substitution $x = u - M$ gives
$\int_{-\infty}^{t-M} f(x)\,\dd x = \int_{-\infty}^{t} f(u - M)\,\dd u$;
since the whole segment from $u - M$ to $u$ lies in $(-\infty, t\,]$,
$\log f(u - M) \leq \log f(u) - rM$, so the integral is at most
$e^{-rM}\int_{-\infty}^{t} f(u)\,\dd u \leq e^{-rM}$. The second claim is the
mirror image.
\end{proof}

\begin{proof}[Theorems~\ref{thm:prior_lims}--\ref{thm:prior-concentration-one}]
Both statements follow the same two steps: the conditional log-density has
derivative bounded away from zero on the relevant half-line, and
Lemma~\ref{lem:drift} converts that drift into an exponential tail bound,
which we then read off on the $\kappa$ scale. For $\varepsilon \geq 1$ both
statements are trivial, since $\kappa_i \in (0,1)$; we therefore take
$\varepsilon < 1$. The function
$\log p_Z(z) = az - (a+b)\log(1+e^z) - \log B(a,b)$ is strictly concave with
derivative $a - (a+b)\sigma(z)$, $\sigma(z) = e^z/(1+e^z)$, so
$\pi(h_i \mid \cdot)$ is log-concave with
\[
  \frac{\dd}{\dd h_i}\log\pi(h_i \mid \cdot)
   \;=\; \sum_{j\sim i}\bigl\{a - (a+b)\,\sigma(h_i - h_j)\bigr\}.
\]
Let $h_{\min} = \min_{j\sim i} h_j$ and fix $c_0$ with
$\sigma(-c_0) = a/\{2(a+b)\}$ (possible since $a, b > 0$). For
$h_i \le h_{\min} - c_0$ every summand is at least $a/2$, so the log-density
increases at rate at least $r = n_i a/2$ on $(-\infty,\, h_{\min} - c_0]$.
Lemma~\ref{lem:drift} then gives, for $h_{\min} \geq c_0 + t$ (so that the
lemma applies with $M \geq 0$; the bound is vacuous otherwise),
$\mathbb{P}(h_i \le h_{\min} - c_0 - M \mid \cdot) \le e^{-n_i a M/2}$.
Now $\kappa^* = \max_{j\sim i}\kappa_j \to 0$ means $h_{\min} \to \infty$,
and for fixed $\varepsilon$, with $t_\varepsilon = \log\{(1-\varepsilon)/\varepsilon\}$,
\[
  \mathbb{P}(\kappa_i \ge \varepsilon \mid \cdot)
   = \mathbb{P}(h_i \le t_\varepsilon \mid \cdot)
   \le \exp\bigl\{-n_i a\,(h_{\min} - c_0 - t_\varepsilon)/2\bigr\} \to 0.
\]

The second statement is symmetric: choosing $c_1$ with
$\sigma(c_1) = (2a+b)/\{2(a+b)\}$, the log-density decreases at rate at least
$n_i b/2$ on $[h_{\max} + c_1, \infty)$ with
$h_{\max} = \max_{j\sim i} h_j$. Then $\kappa_* \to 1$ means
$h_{\max} \to -\infty$, so by the second claim of
Lemma~\ref{lem:drift},
$\mathbb{P}(\kappa_i \le 1 - \varepsilon \mid \cdot) \to 0$.
\end{proof}

\subsection*{Proof of Theorem~\ref{thm:posterior-shrinkage}}
The two parts of Theorem~\ref{thm:posterior-shrinkage} share an unnormalized
conditional posterior whose tail near $\kappa = 1$ is controlled via
Lemma~\ref{lem:diverge-one-kh0}.
Throughout, we take $\varepsilon \leq 1/2$ without loss: the events are
monotone in $\varepsilon$, so larger $\varepsilon$ follow at once.

\begin{lemma}[Divergence near $\kappa = 1$, product form]\label{lem:diverge-one-kh0}
Write $\kappa_* = \min_{j\sim i}\kappa_j$. For every
$\varepsilon \in (0, 1/2]$,
\[
  \int_{1-\varepsilon}^{1}
    \frac{\{\kappa(1-\kappa)\}^{n_i/2-1}}
         {\prod_{j\sim i}\bigl\{\kappa(1-\kappa_j) + (1-\kappa)\kappa_j\bigr\}}\,\dd\kappa
   \;\geq\; F_{\varepsilon, n_i}(1 - \kappa_*),
   \]
where $F_{\varepsilon, n_i}(\delta) \to +\infty$ as $\delta \downarrow 0$,
so the integral diverges whenever $\kappa_* \to 1$.
\end{lemma}

\begin{proof}[Lemma~\ref{lem:diverge-one-kh0}]
Substitute $t = 1 - \kappa$ and write $\delta = 1 - \kappa_*$. On
$t \in (0, \varepsilon]$ we have $\kappa \geq 1/2$, so
$\kappa^{n_i/2-1} \geq c_{n_i} = \min\{1, 2^{-(n_i/2-1)}\} > 0$, while each
factor of the denominator satisfies
$\kappa(1-\kappa_j) + t\,\kappa_j \leq (1-\kappa_j) + t \leq \delta + t$.
Hence the integral is bounded below by
$F_{\varepsilon, n_i}(\delta) = c_{n_i}\int_0^\varepsilon t^{n_i/2-1}/(\delta+t)^{n_i}\,\dd t$,
whose integrand increases monotonically as $\delta \downarrow 0$ to
$c_{n_i}\,t^{-n_i/2-1}$, non-integrable at $0$ (since $n_i/2 + 1 > 1$ for
$n_i \geq 1$); by Monotone Convergence $F_{\varepsilon, n_i}(\delta) \to +\infty$.
\end{proof}

\begin{proof}[Theorem~\ref{thm:posterior-shrinkage}]
Both parts concern the conditional posterior of $\kappa_i$ given $\bfy$ and
the neighboring values $\{\kappa_j\}_{j \sim i}$. With $k = -1$ we have
$\omega_i = \beta_i$, so, for $\sigma_\epsilon = \tau = 1$, integrating the
trend out of the likelihood~\eqref{eq:likelihood} against the
prior~\eqref{eq:prior} gives the marginal likelihood
$y_i \mid \lambda_i \sim \Normal(0, 1 + \lambda_i^2)$. Since
$1 + \lambda_i^2 = \kappa_i^{-1}$, this contributes the factor
$\kappa_i^{1/2}\exp(-\kappa_i y_i^2/2)$, which depends on $\bfy$ only through
$y_i$. Multiplying by the conditional prior~\eqref{eq:cond-prior} at
$a = b = 1/2$, the conditional posterior of $\kappa_i$ is proportional to
\[
  g_y(\kappa_i)
   \;=\; \frac{\exp(-\kappa_i y_i^2/2)\,
               \kappa_i^{(n_i-1)/2}\,(1-\kappa_i)^{n_i/2-1}}
              {\prod_{j\sim i}\bigl\{\kappa_i(1-\kappa_j) + (1-\kappa_i)\kappa_j\bigr\}}.
\]

\textbf{Proof of (i).}
We bound the ratio
$\int_0^{1-\varepsilon} g_y\,\dd\kappa_i \big/ \int_{1-\varepsilon}^1 g_y\,\dd\kappa_i$
uniformly for $\bfy \in \mathcal{Y}_M$ (recall
$\mathcal{Y}_M = \{\bfy \in \R^n : |y_i| \le M\}$ from
Theorem~\ref{thm:posterior-shrinkage}); since $\kappa_* \to 1$, we may
assume $\kappa_* \geq 1/2$. On
$(0, 1-\varepsilon]$, $\exp(-\kappa_i y_i^2/2) \leq 1$ and each denominator
factor satisfies
$\kappa_i(1-\kappa_j) + (1-\kappa_i)\kappa_j \geq (1-\kappa_i)\kappa_j
\geq \varepsilon\,\kappa_* \geq \varepsilon/2$, so
\[
  \int_0^{1-\varepsilon} g_y\,\dd\kappa_i
  \;\leq\; \Bigl(\frac{2}{\varepsilon}\Bigr)^{n_i}
     \int_0^1 \kappa^{(n_i-1)/2}(1-\kappa)^{n_i/2-1}\,\dd\kappa
  \;=:\; C_1 \;<\; \infty,
\]
free of $y_i$ and of $\{\kappa_j\}$. On $[1-\varepsilon, 1]$,
$\exp(-\kappa_i y_i^2/2) \geq e^{-M^2/2}$ and
$\kappa_i^{(n_i-1)/2} = \kappa_i^{1/2}\,\kappa_i^{n_i/2-1}
\geq 2^{-1/2}\kappa_i^{n_i/2-1}$, so
$\int_{1-\varepsilon}^1 g_y\,\dd\kappa_i$ is bounded below by
$2^{-1/2}e^{-M^2/2}\,F_{\varepsilon, n_i}(1-\kappa_*)$ by
Lemma~\ref{lem:diverge-one-kh0}. The
ratio is therefore at most
$2^{1/2}e^{M^2/2}\,C_1/F_{\varepsilon, n_i}(1-\kappa_*) \to 0$ as $\kappa_* \to 1$,
uniformly for $\bfy \in \mathcal{Y}_M$  and over all neighbor
configurations sharing $\kappa_*$.

\textbf{Proof of (ii).}
Fix the neighboring values $\{\kappa_j\}_{j\sim i}$ and let $|y_i| \to \infty$. On
$[\varepsilon, 1]$, $\exp(-\kappa_i y_i^2/2) \leq e^{-\varepsilon y_i^2/2}$
and each denominator factor satisfies
$\kappa_i(1-\kappa_j) + (1-\kappa_i)\kappa_j \geq \varepsilon(1-\kappa_j) > 0$,
so $\int_\varepsilon^1 g_y\,\dd\kappa_i \leq C_2\,e^{-\varepsilon y_i^2/2}$
with $C_2 = C_2(\{\kappa_j\}, \varepsilon) < \infty$. On $(0, \varepsilon]$,
each denominator factor is at most
$\kappa_i + \kappa_j \leq 2$ and
$(1-\kappa_i)^{n_i/2-1} \geq \min\{1, (1-\varepsilon)^{n_i/2-1}\} > 0$, so,
substituting $u = \kappa_i y_i^2$, for $|y_i| \geq 1$
\[
  \int_0^{\varepsilon} g_y\,\dd\kappa_i
  \;\geq\; c\,|y_i|^{-(n_i+1)}
     \int_0^{\varepsilon y_i^2} e^{-u/2}\,u^{(n_i-1)/2}\,\dd u
  \;\geq\; c'\,|y_i|^{-(n_i+1)}
\]
for constants $c, c' > 0$ free of $y_i$. The ratio
$\int_\varepsilon^1 g_y / \int_0^\varepsilon g_y
\leq C_3\,|y_i|^{n_i+1} e^{-\varepsilon y_i^2/2} \to 0$, which is (ii): the
mass on $[\varepsilon, 1]$ decays exponentially in $y_i^2$ while the mass on
$(0, \varepsilon]$ decays only polynomially.
\end{proof}

\subsection*{Proof of Theorem~\ref{thm:posterior-proper}}

The proof rests on four lemmas, each with a single, separable role.
Lemma~\ref{lem:kernel-preserved} fixes the null space of the graph-difference
operators. Lemma~\ref{lem:trend-marginal} integrates $\bfbeta$ out of the
likelihood-times-prior product and produces an upper bound on the marginal that is
uniform in $\sigma_\epsilon^2$ and in the local components.
Lemma~\ref{lem:global-shrinkage-marginal} integrates that bound over the global
scale, using only that the log-normal prior on $\tau^2$ has finite negative
moments. Lemma~\ref{lem:hperp-proper} establishes that the local prior on
$\bfh^\perp$ is proper.

Throughout this proof we write $m = |E|$ for the number of edges. The
argument is presented for $k$ even or zero, where the shrinkage vector $\bfh$ has length
$m$; when $k$ is odd, $\bfh$ is indexed instead by the $n$ vertices, and one substitutes
$n$ for $m$ throughout.

\begin{lemma}[Null space of graph-difference operators]\label{lem:kernel-preserved}
If $G = (V, E)$ is connected, then
$\ker(\Delta^{(k+1)}) = \mathrm{span}\{\bfone_n\}$ for all
$k \geq 0$.
\end{lemma}
\begin{proof}[Lemma~\ref{lem:kernel-preserved}]
Induction on $k$; the two parities are handled separately
because $\Delta^{(k+1)}$ factors as $\Delta^{(1)}\Delta^{(k)}$ for $k$ even
and as $(\Delta^{(1)})^\top\Delta^{(k)}$ for $k$ odd. For $k = 0$, $\Delta^{(1)} = (\delta(e_1),\ldots,\delta(e_{m}))^\top$
is the oriented incidence matrix and $\kernel(\Delta^{(1)}) = \mathrm{span}\{\bfone_n\}$
because $G$ is connected. Suppose $\kernel(\Delta^{(k)}) = \mathrm{span}\{\bfone_n\}$
for some $k \geq 1$.

\emph{$k$ even.} Then $\Delta^{(k+1)} = \Delta^{(1)}\Delta^{(k)}$ with
$\Delta^{(k)} = L^{k/2}$ a power of the graph Laplacian
$L = (\Delta^{(1)})^\top \Delta^{(1)}$. Since $L$ is symmetric and
$L\bfone_n = \mathbf{0}$, also $L^{k/2}\bfone_n = \mathbf{0}$, hence
$\bfone_n \in \kernel(L^{k/2}) = \mathrm{im}(L^{k/2})^{\perp}$, i.e.\
$\mathrm{im}(\Delta^{(k)})\perp\bfone_n$. Trivially
$\mathrm{span}\{\bfone_n\}\subseteq\kernel(\Delta^{(k+1)})$. Conversely,
$\Delta^{(1)}(\Delta^{(k)}\mathbf{v}) = \mathbf{0}$ implies
$\Delta^{(k)}\mathbf{v}\in\kernel(\Delta^{(1)}) = \mathrm{span}\{\bfone_n\}$,
say $\Delta^{(k)}\mathbf{v} = t\bfone_n$. Taking inner product with
$\bfone_n$ and using $\mathrm{im}(\Delta^{(k)})\perp\bfone_n$ yields
$tn = 0$, hence $\Delta^{(k)}\mathbf{v} = \mathbf{0}$, and the inductive
hypothesis gives $\mathbf{v}\in\mathrm{span}\{\bfone_n\}$.

\emph{$k$ odd.} Then $\Delta^{(k+1)} = (\Delta^{(1)})^\top\Delta^{(k)}$.
Trivially $\bfone_n\in\kernel(\Delta^{(k+1)})$, since $\Delta^{(k)}\bfone_n
= \mathbf{0}$ by the inductive hypothesis. Conversely, let
$\mathbf{v} \in \kernel(\Delta^{(k+1)})$ and set
$\mathbf{w} = \Delta^{(k)}\mathbf{v}$. Then
$(\Delta^{(1)})^\top \mathbf{w} = \Delta^{(k+1)}\mathbf{v} = \mathbf{0}$, so
$\mathbf{w}\in\kernel((\Delta^{(1)})^\top)$. Since $k$ is odd,
$\Delta^{(k)} = \Delta^{(1)} L^{(k-1)/2}$ factors through $\Delta^{(1)}$,
so $\mathbf{w}\in\mathrm{im}(\Delta^{(1)})$. By the fundamental theorem of
linear algebra these two subspaces are orthogonal, so $\mathbf{w} = \mathbf{0}$,
i.e.\ $\Delta^{(k)}\mathbf{v} = \mathbf{0}$, and the inductive hypothesis
gives $\mathbf{v}\in\mathrm{span}\{\bfone_n\}$.
\end{proof}

\begin{lemma}[Trend marginal: upper bound]\label{lem:trend-marginal}
Let $\Delta^{(k+1)}$ with $k\geq0$ and $\rank(\Delta^{(k+1)}) = n - 1$, let
$\Lambda = \mathrm{diag}(\lambda_1^2,\ldots,\lambda_m^2)$ with all
$\lambda_j > 0$, $\tau^2 > 0$, and let $\sigma_\epsilon^2 > 0$. Then
\[
  M(\Lambda, \tau^2, \sigma_\epsilon^2)
    \;=\; \int_{\R^n} p(\bfy\mid\bfbeta,\sigma_\epsilon^2)\,
                       p(\Delta^{(k+1)}\bfbeta\mid\Lambda, \tau^2)\,\dd\bfbeta
\]
admits the closed form
\begin{equation}\label{eq:M-closed}
  M(\Lambda, \tau^2,\sigma_\epsilon^2) = C\,(\sigma_\epsilon^2)^{-(n-1)/2}\,
     |Q_\alpha|^{-1/2}\,\tau^{-m}\,\prod_{j=1}^m\lambda_j^{-1}\,e^{R(\bfy,\sigma_\epsilon^2,\tau^2,\Lambda)},
\end{equation}
where $C = (2\pi)^{-m/2}$ is a positive constant;  $Q_\alpha = \sigma_\epsilon^{-2}I_{n-1} + \tau^{-2} \,
V^\top {\Delta^{(k+1)}}^{\top} \Lambda^{-1} \Delta^{(k+1)} V$ is a positive definite matrix
for any $V\in\R^{n\times(n-1)}$ whose columns form an orthonormal basis
of $(\kernel \left(\Delta^{(k+1)}\right))^\perp$; and $R(\bfy,\sigma_\epsilon^2,\tau^2,\Lambda) \le 0$ is the exponent that
remains after integrating out $\bfbeta$, given explicitly in the proof.

In particular, if $\prod_{j=1}^m \lambda_j = 1$, then for every $\tau^2 > 0$ and $\sigma_{\epsilon}^2 > 0$,
\begin{equation}\label{eq:M-bound}
  M(\Lambda, \tau^2, \sigma_\epsilon^2) \;\le\; C \tau^{-m}.
\end{equation}
\end{lemma}
The bound~\eqref{eq:M-bound} is what Step~1 of the
Theorem~\ref{thm:posterior-proper} proof needs. Note that $\prod_{j=1}^m \lambda_j = 1$ holds by construction under \eqref{eq:h-decomp}.%

\begin{proof}[Lemma~\ref{lem:trend-marginal}]
We split $\bfbeta$ along $\kernel(\Delta^{(k+1)})$ and its
orthogonal complement, integrate the two directions separately, and bound
the resulting quadratic form via $Q_\alpha \succcurlyeq \sigma_\epsilon^{-2}I_{n-1}$.
Since $\rank(\Delta^{(k+1)}) = n - 1$ and $\bfone_n \in \kernel(\Delta^{(k+1)})$ always, $\kernel\left(\Delta^{(k+1)}\right) = \mathrm{span}\{\bfone_n\}$; for connected $G$ this is Lemma~\ref{lem:kernel-preserved}.
Let $U = n^{-1/2}\bfone_n$ and $V$ as above, so $[U, V]$ is orthogonal.
Decompose $\bfbeta = U\gamma + V\pmb\alpha$. Then $\Delta^{(k+1)}\bfbeta = \Delta^{(k+1)}V\pmb\alpha$ and
$\|\bfy - \bfbeta\|_2^2 = (U^\top\bfy - \gamma)^2 + \|V^\top\bfy - \pmb\alpha\|_2^2$. The prior is $p(\Delta^{(k+1)}\bfbeta\mid\Lambda, \tau^2)\propto \tau^{-m} \prod_j\lambda_j^{-1}\exp(-\pmb\alpha^\top P_\alpha\pmb\alpha/2)$, where the matrix $P_\alpha = \tau^{-2}\, V^\top {\Delta^{(k+1)}}^\top  \Lambda^{-1}\Delta^{(k+1)}V$ is positive definite. To see why, note that $\Delta^{(k+1)}V$ must have full column rank: the columns of $V$ span $\kernel(\Delta^{(k+1)})^\perp$ by definition, and $\Delta^{\left(k+1\right)}\bfbeta$ depends only on $\pmb\alpha$.

The $\gamma$-integral contributes $(2\pi\sigma_\epsilon^2)^{1/2}$; completing
the square in $\pmb\alpha$ and integrating yields~\eqref{eq:M-closed} with
\[
  R(\bfy,\sigma_\epsilon^2,\tau^2,\Lambda)
   \;=\; \frac{(V^\top\bfy)^\top Q_\alpha^{-1}(V^\top\bfy)}{2\sigma_\epsilon^4}
       - \frac{\|V^\top\bfy\|_2^2}{2\sigma_\epsilon^2}.
\]
Since $Q_\alpha\succcurlyeq\sigma_\epsilon^{-2}I_{n-1}$, we have
$Q_\alpha^{-1}\preccurlyeq\sigma_\epsilon^2 I_{n-1}$ (so $R\le 0$) and
$|Q_\alpha|^{-1/2}\le(\sigma_\epsilon^2)^{(n-1)/2}$. Substituting
into~\eqref{eq:M-closed} gives~\eqref{eq:M-bound}.
\end{proof}

\begin{lemma}[Global shrinkage marginal]
\label{lem:global-shrinkage-marginal}
Let $\tau^2 = e^{\hzero}$ with $\hzero \sim \Normal(\muG, s_0^2)$ for some
$\muG \in \R$ and $s_0^2 > 0$, as in \eqref{eq:h-priors}, and let
$M(\Lambda, \tau^2, \sigma_\epsilon^2)$ be the trend marginal of
Lemma~\ref{lem:trend-marginal} with the geometric mean of $\{\lambda_j\}$ equal
to $1$. Then, for every $p > 0$,
\begin{equation}\label{eq:lognormal-moment}
  \mathbb{E}\bigl(\tau^{-p}\bigr)
  \;=\; \mathbb{E}\bigl\{e^{-(p/2)\hzero}\bigr\}
  \;=\; \exp\Bigl(-\tfrac{p}{2}\,\muG + \tfrac{p^2}{8}\,s_0^2\Bigr)
  \;<\; +\infty,
\end{equation}
and consequently the global shrinkage marginal is finite,
\begin{equation*}
  \int_0^\infty M\!\left(\Lambda, \tau^2, \sigma_\epsilon^2\right)
  p(\tau^2)\,\dd\tau^2 \;\le\; C\,\mathbb{E}\bigl(\tau^{-m}\bigr) \;<\; +\infty,
\end{equation*}
for every $k \geq 0$ and $k_h \geq 0$, with no condition on the graph.
\end{lemma}

\begin{proof}[Lemma~\ref{lem:global-shrinkage-marginal}]
The identity \eqref{eq:lognormal-moment} is the moment generating function of
$\hzero \sim \Normal(\muG, s_0^2)$ evaluated at $-p/2$, finite for every $p$. The
bound follows by integrating \eqref{eq:M-bound} of Lemma~\ref{lem:trend-marginal}
against the prior density $p(\tau^2)$.
\end{proof}

The exact growth rate of $M$ as $\tau^2 \to 0^+$ explains the mean offset in
\eqref{eq:h-priors}, although the proof of Theorem~\ref{thm:posterior-proper} does
not require it. As $\tau^2 \to 0^+$,
\begin{equation*}
    Q_\alpha
    = \sigma_\epsilon^{-2}I_{n-1}
      + \tau^{-2}V^\top{\Delta^{(k+1)}}^\top\Lambda^{-1}\Delta^{(k+1)}V
    \;\sim\; \tau^{-2} V^\top{\Delta^{(k+1)}}^\top\Lambda^{-1}\Delta^{(k+1)}V \succ 0,
\end{equation*}
so $|Q_\alpha|^{-1/2} \sim \tau^{n-1}|V^\top{\Delta^{(k+1)}}^\top\Lambda^{-1}\Delta^{(k+1)}V|^{-1/2}$,
while $e^{R(\bfy,\sigma_\epsilon^2,\tau^2,\Lambda)} \to
e^{-\|V^\top\bfy\|^2/(2\sigma_\epsilon^2)} > 0$ as $Q_\alpha^{-1}\to 0$.
Substituting into~\eqref{eq:M-closed} and absorbing all factors constant in
$\tau^2$ into $C' > 0$ gives, for fixed
$(\bfy, \Lambda, \sigma_\epsilon^2)$,
\begin{equation}\label{eq:M-asymp}
    M\!\left(\Lambda,\tau^2,\sigma_\epsilon^2\right)
    \;=\; C'\,\tau^{-(m-n+1)}
           \bigl(1 + o(1)\bigr),
    \qquad \tau^2\to 0^+.
\end{equation}
The exponent $m - n + 1 = m - (n - 1)$ counts the shrinkage coordinates left
uninformed by the trend: $\Delta^{(k+1)}$ has rank $n - 1$, so for $k$ even or zero
it equals the dimension of the cycle space of $G$ (zero on trees), and for $k$ odd
it equals one after substituting $n$ for $m$. On the log scale, \eqref{eq:M-asymp} is the linear tilt
$\exp\{-(m - n + 1)\hzero/2\}$, and the mean offset in \eqref{eq:h-priors} cancels
it exactly:
\begin{equation}\label{eq:offset-cancel}
  e^{-(m-n+1)\hzero/2}\,
  e^{-\left\{\hzero - \mu_0 - (m-n+1)s_0^2/2\right\}^2/(2 s_0^2)}
  \;\propto\;
  e^{-(\hzero - \mu_0)^2/(2 s_0^2)},
\end{equation}
up to a constant free of $\hzero$. Hence near total shrinkage the level prior acts
like the centered Gaussian $\Normal(\mu_0, s_0^2)$: the growth of the marginal
likelihood creates no spurious mode at $\tau^2 = 0$.

\begin{lemma}[Properness of the local prior]\label{lem:hperp-proper}
Let $\Delta_k^{k_h+1} \in \R^{r \times m}$ with $k_h \geq 0$ be the shrinkage
operator, whose null space is $\mathrm{span}\{\bfone_m\}$: this follows from
Lemma~\ref{lem:kernel-preserved} applied to the shrinkage graph, which is
connected whenever $G$ is: for $k$ odd it is $G$ itself, and for $k$ even or
zero it is the line graph of $G$, connected for any connected $G$ with at least
one edge. If $m = 1$ the hyperplane $\bfone_m^\perp = \{0\}$ is a single point
and the conclusion is trivial; assume $m \geq 2$, so that
$\Delta_k^{k_h+1}$ has a nonzero singular value. Let
$\sigma_{\min}^{+} > 0$ denote its smallest nonzero singular value. Then, for any
$a, b > 0$, the local prior
\[
  \pi(\bfh^\perp) \;\propto\;
  \prod_{i=1}^{r} p_Z\bigl\{(\Delta_k^{k_h+1}\bfh^\perp)_i\bigr\},
  \qquad \bfh^\perp \in \bfone_m^\perp,
\]
with $p_Z(z) = \{\mathrm{B}(a,b)\}^{-1} e^{az}(1 + e^z)^{-(a+b)}$ the
$Z(a, b, 0, 1)$ density, is proper on the hyperplane $\bfone_m^\perp$.
\end{lemma}

\begin{proof}[Lemma~\ref{lem:hperp-proper}]
The $Z$-density has exponential tails: for $z \geq 0$, using
$(1 + e^z)^{-(a+b)} \leq e^{-(a+b)z}$ gives
$p_Z(z) \leq \{\mathrm{B}(a,b)\}^{-1} e^{-bz}$, while for $z < 0$, using
$(1 + e^z)^{-(a+b)} \leq 1$ gives $p_Z(z) \leq \{\mathrm{B}(a,b)\}^{-1} e^{az}$;
hence $p_Z(z) \leq \{\mathrm{B}(a,b)\}^{-1} e^{-\min(a,b)|z|}$ for all $z$.
Writing $\bfeta = \Delta_k^{k_h+1}\bfh^\perp$,
  \begin{align*}
  \prod_{i=1}^{r} p_Z(\eta_i) & \leq \{\mathrm{B}(a,b)\}^{-r} e^{-\min(a,b)\,\|\bfeta\|_1} \\
  & \leq \{\mathrm{B}(a,b)\}^{-r} e^{-\min(a,b)\,\|\bfeta\|_2} \\
  & \leq \{\mathrm{B}(a,b)\}^{-r}
     e^{-\min(a,b)\,\sigma_{\min}^{+}\|\bfh^\perp\|_2},
\end{align*}
using $\|\cdot\|_1 \geq \|\cdot\|_2$ and, on $\bfone_m^\perp$,
$\|\Delta_k^{k_h+1}\bfh^\perp\|_2 \geq \sigma_{\min}^{+}\|\bfh^\perp\|_2$. The
final bound is integrable over the $(m-1)$-dimensional hyperplane
$\bfone_m^\perp$.
\end{proof}

The argument covers every $k_h \geq 0$ uniformly, in particular the
case $k_h = 0$, where $r$ exceeds $m - 1$ whenever the shrinkage
graph contains cycles and no bijective change of variables to the innovations
exists.

\begin{proof}[Theorem~\ref{thm:posterior-proper}]
We evaluate the properness of the joint posterior by checking the integrability of:
\begin{equation}\label{eq:posterior-integral}
  \int p(\bfy\mid\bfbeta,\sigma_\epsilon^2)\,p(\Delta^{(k+1)}\bfbeta\mid\Lambda, \tau^2)\,
       \pi(\bfh^\perp)\,\pi(\tau^2)\,\pi(\sigma_\epsilon^2)\,\dd\bfbeta\,\dd\tau^2\,\dd\bfh^\perp\,\dd\sigma_\epsilon^2,
\end{equation}
where $\Lambda = \mathrm{diag}(e^{h_j^\perp})$ represents the local
components, with $\bfh^\perp$ ranging over the zero-sum subspace
$\bfone_m^\top\bfh^\perp = 0$ of the change of coordinates
\eqref{eq:h-decomp} (so the geometric mean of the $e^{h_j^\perp}$ is
$1$ by construction): $h_j = \hzero + h_j^\perp$ with
$\tau^2 = e^{\hzero}$, so $\lambda_j^2 = e^{h_j^\perp}$ and this $\Lambda$
is exactly the $\Lambda = \mathrm{diag}(\lambda_1^2,\ldots,\lambda_m^2)$ of
Lemma~\ref{lem:trend-marginal}, with the zero-sum coordinate supplying the
geometric-mean-one normalization that \eqref{eq:M-bound} requires.
Here $\pi(\tau^2)$ is the log-normal density induced by \eqref{eq:h-priors}, and
$\pi(\sigma_\epsilon^2)$ is the Inverse-Gamma prior. We proceed sequentially in
three steps; the integrand is nonnegative, so the sequential
order is justified (Tonelli).

\textbf{Step~1 ($\bfbeta$ and $\tau^2$).} We integrate out $\bfbeta$ via
Lemma~\ref{lem:trend-marginal} to yield the trend marginal
$M(\Lambda, \tau^2, \sigma_\epsilon^2)$. Since the geometric mean of the local
components is fixed equal to one on $\bfone_m^\perp$, the uniform
bound~\eqref{eq:M-bound} applies: $M \leq C\,\tau^{-m}$, uniformly in
$\sigma_\epsilon^2 > 0$ and $\bfh^\perp \in \bfone_m^\perp$. Integrating over the
global scale with Lemma~\ref{lem:global-shrinkage-marginal},
\[
  \int_0^\infty M(\Lambda, \tau^2, \sigma_\epsilon^2)\,\pi(\tau^2)\,\dd\tau^2
  \;\leq\; C\,\mathbb{E}\bigl(\tau^{-m}\bigr) \;=\; C_1 \;<\; +\infty,
\]
with $C_1$ free of both $\sigma_\epsilon^2$ and $\bfh^\perp$.

\textbf{Step~2 ($\bfh^\perp$).} The bound of Step~1 is constant in $\bfh^\perp$,
and the local prior $\pi(\bfh^\perp)$ is proper on $\bfone_m^\perp$ for every
$k_h \geq 0$ by Lemma~\ref{lem:hperp-proper}. Hence the integral over
$\bfh^\perp$ remains bounded by $C_1$.

\textbf{Step~3 ($\sigma_\epsilon^2$).} The Inverse-Gamma
prior~\eqref{eq:sigma-prior} is proper for any $a_\sigma, b_\sigma > 0$, so integrating the
constant bound of Step~2 leaves \eqref{eq:posterior-integral} bounded by $C_1$.
The joint marginal is finite, confirming that the posterior is proper.
\end{proof}

\section{Higher-order difference operators for shrinkage ($k_h = 1$)}
\label{sec:kh1}
We now extend the GDSP prior and posterior analysis (Section~\ref{sec-theory}) from $k_h = 0$ to $k_h = 1$. When $k_h = 0$,  the full conditionals only require comparisons with immediate neighbors
(Theorem~\ref{thm:cond-prior}). For $k_h \geq 1$, the shrinkage graph differences couple each $h_i$ to vertices beyond its immediate neighbors.
The following result collects the $k_h = 1$ analogues of
Theorems~\ref{thm:cond-prior}--\ref{thm:posterior-shrinkage}.

\begin{theorem}[Higher-order shrinkage operators]\label{rem:kh1-prior}
Let $k_h = 1$ and write
$\zeta_i = \prod_{j \sim i}(1-\kappa_j)/\kappa_j = \prod_{j \sim i} \lambda_j^2$.
\begin{enumerate}
\item[\textbf{(i)}] For any $a, b > 0$,
\begin{equation}\label{eq:cond-prior-kh1}
  \pi_i(\kappa_i \mid \{\kappa_j\}_{j \sim i})
  = \frac{n_i}{B(a,b)}
    \frac{\zeta_i^{b}\,\kappa_i^{bn_i - 1}\,(1-\kappa_i)^{an_i - 1}}
         {\bigl\{\zeta_i \kappa_i^{n_i} + (1-\kappa_i)^{n_i}\bigr\}^{a+b}}.
\end{equation}
\item[\textbf{(ii)}] For the horseshoe GDSP ($a = b = 1/2$) and any
$\varepsilon > 0$,
$\mathbb{P}_i(\kappa_i < \varepsilon \mid \{\kappa_j\}_{j\sim i})
\to 1$ as $\zeta_i \to \infty$, and
$\mathbb{P}_i(\kappa_i > 1 - \varepsilon \mid \{\kappa_j\}_{j\sim i})
\to 1$ as $\zeta_i \to 0$.
\item[\textbf{(iii)}] Theorem~\ref{thm:posterior-shrinkage} holds: for part (i) replace $\kappa_* \to 1$ with $\zeta_i \to 0$, and for part (ii) replace $\kappa_j \in (0, 1)$ for all $j \sim i$ with $\zeta_i \in (0, \infty)$.
\end{enumerate}
\end{theorem}

The neighbors now inform $\kappa_i$ only through the single aggregate
$\zeta_i$. Unanimity among the neighbors is therefore sufficient but not
necessary: a single extreme neighbor can dominate the product $\zeta_i$
when the remaining neighbors are bounded away from the opposite
extreme. In this sense, $k_h = 1$ propagates shrinkage through an
aggregate of the neighborhood, whereas $k_h = 0$ weighs each neighbor
separately.

The three parts are proved separately below: the density
\eqref{eq:cond-prior-kh1} first; the concentration limits (ii) after the
auxiliary Lemma~\ref{lem:g-bounds}; and the posterior adaptation (iii)
after Lemma~\ref{lem:diverge-one}.

\begin{proof}[Theorem~\ref{rem:kh1-prior}(i)]
Because $k_h = 1$, the shrinkage operator in \eqref{eq:prior-3}
is $\Delta_k^{2} = L$, the Laplacian of the shrinkage graph, so the prior on
$\bfh = (\log\lambda_1^2,\ldots,\log\lambda_m^2)^\top$
($m = |E|$ for $k$ even or zero; substitute $n$ for $m$ when
$k$ is odd) is specified through the
graph Laplacian, $L\bfh = \bfeta$ with
$\eta_i \sim^{iid} Z(a,b,0,1)$. For a non-isolated vertex $i$,
this is $n_i h_i - \sum_{j\sim i} h_j = \eta_i$, equivalently
$h_i = \mu_i + \eta_i/n_i$ with $\mu_i = n_i^{-1}\sum_{j\sim i} h_j$. The
$Z$-distribution is a location-scale family, so
\[
  h_i \mid \{h_j\}_{j \sim i} \;\sim\; Z(a, b, \mu_i, 1/n_i),
\]
with density $f(h_i) = (n_i/B(a,b))\cdot e^{a w}/(1+e^w)^{a+b}$, where
$w = n_i(h_i - \mu_i) = n_i h_i - \sum_{j\sim i}h_j$.

Reparametrize via $\kappa_i = 1/(1+e^{h_i})$, so that
$h_i = \log\{(1-\kappa_i)/\kappa_i\}$ and $|\dd h_i/\dd\kappa_i| =
\{\kappa_i(1-\kappa_i)\}^{-1}$. Substituting,
\[
  w \;=\; n_i\log\frac{1-\kappa_i}{\kappa_i} - \sum_{j\sim i}\log\frac{1-\kappa_j}{\kappa_j}
       \;=\; \log\frac{(1-\kappa_i)^{n_i}}{\zeta_i\,\kappa_i^{n_i}},
  \qquad
  \zeta_i = \prod_{j\sim i}\frac{1-\kappa_j}{\kappa_j}.
\]
The change of variables and the algebraic identity
$e^{aw}/(1+e^w)^{a+b}
 = (1-\kappa_i)^{a n_i}(\zeta_i\kappa_i^{n_i})^b / \{\zeta_i\kappa_i^{n_i}+(1-\kappa_i)^{n_i}\}^{a+b}$
then yield
\[
  \pi_i(\kappa_i\mid\{\kappa_j\}_{j\sim i})
   \;=\; \frac{n_i}{B(a,b)}\cdot
     \frac{\zeta_i^{b}\,\kappa_i^{bn_i-1}\,(1-\kappa_i)^{an_i-1}}
          {\{\zeta_i\kappa_i^{n_i}+(1-\kappa_i)^{n_i}\}^{a+b}},
\]
which is~\eqref{eq:cond-prior-kh1}.
\end{proof}

The two concentration limits of Theorem~\ref{rem:kh1-prior}(ii) share a
common analytical core, isolated
in Lemma~\ref{lem:g-bounds}. In what follows, $g(\kappa;\zeta)$
denotes the kernel of the horseshoe ($a=b=1/2$) conditional
prior~\eqref{eq:cond-prior-kh1}:
\begin{equation}\label{eq:g-def}
  g(\kappa;\zeta) \;=\;
    \frac{\{\kappa(1-\kappa)\}^{n_i/2-1}}{\zeta\,\kappa^{n_i}+(1-\kappa)^{n_i}},
  \qquad \kappa\in(0,1),
\end{equation}
so that $\pi_i(\kappa_i\mid\{\kappa_j\}_{j\sim i}) =
(n_i/\pi)\sqrt{\zeta_i}\,g(\kappa_i;\zeta_i)$.

\begin{lemma}[Tail bounds for $g$]\label{lem:g-bounds}
Let $g$ be as in \eqref{eq:g-def}, the horseshoe
($a = b = 1/2$) kernel.
For every $\zeta>0$ and $\varepsilon\in(0,1)$,
\begin{enumerate}
\item[\textbf{(a)}] $\displaystyle
  \sqrt{\zeta}\int_{\varepsilon}^{1} g(\kappa;\zeta)\,\dd\kappa
  \;\le\; \frac{C_g}{\sqrt{\zeta}\,\varepsilon^{n_i}}$,
\item[\textbf{(b)}] $\displaystyle
  \sqrt{\zeta}\int_{0}^{1-\varepsilon} g(\kappa;\zeta)\,\dd\kappa
  \;\le\; \frac{C_g\,\sqrt{\zeta}}{\varepsilon^{n_i}}$,
\end{enumerate}
where $C_g = \mathrm{B}(n_i/2, n_i/2) < \infty$.
\end{lemma}

The two bounds are stated in the form their consumers need: each left-hand side equals
$(\pi/n_i)\,\mathbb{P}_i(\kappa_i\in I)$ for the relevant tail interval $I$, so
the concentration claims of Theorem~\ref{rem:kh1-prior}(ii)
reduce to taking limits in $\zeta$.

\begin{proof}[Lemma~\ref{lem:g-bounds}]
On $\kappa\in[\varepsilon,1)$, the denominator of $g$ satisfies
$\zeta\kappa^{n_i}+(1-\kappa)^{n_i}\ge\zeta\varepsilon^{n_i}$, so
\[
  \int_{\varepsilon}^{1} g(\kappa;\zeta)\,\dd\kappa
   \;\le\; \frac{1}{\zeta\varepsilon^{n_i}}\int_0^1\{\kappa(1-\kappa)\}^{n_i/2-1}\,\dd\kappa
   \;=\; \frac{C_g}{\zeta\varepsilon^{n_i}}.
\]
Multiplying by $\sqrt\zeta$ gives (a). Part (b) is analogous: on
$\kappa\in(0,1-\varepsilon]$,
$\zeta\kappa^{n_i}+(1-\kappa)^{n_i}\ge(1-\kappa)^{n_i}\ge\varepsilon^{n_i}$.
\end{proof}

\begin{proof}[Theorem~\ref{rem:kh1-prior}(ii), $\zeta_i \to \infty$]
The claim is trivial for $\varepsilon \geq 1$, so take $\varepsilon < 1$.
As $\zeta_i \to \infty$, Lemma~\ref{lem:g-bounds}(a) gives
\[
  \mathbb{P}_i(\kappa_i\ge\varepsilon\mid\{\kappa_j\}_{j\sim i})
   \;=\; \frac{n_i}{\pi}\sqrt{\zeta_i}\int_\varepsilon^1 g(\kappa;\zeta_i)\,\dd\kappa
   \;\le\; \frac{n_i C_g}{\pi\sqrt{\zeta_i}\,\varepsilon^{n_i}}
   \;\to\; 0.
\]
\end{proof}

\begin{proof}[Theorem~\ref{rem:kh1-prior}(ii), $\zeta_i \to 0$]
The claim is trivial for $\varepsilon \geq 1$, so take $\varepsilon < 1$.
As $\zeta_i \to 0$, Lemma~\ref{lem:g-bounds}(b) gives
\[
  \mathbb{P}_i(\kappa_i\le 1-\varepsilon\mid\{\kappa_j\}_{j\sim i})
   \;=\; \frac{n_i}{\pi}\sqrt{\zeta_i}\int_0^{1-\varepsilon} g(\kappa;\zeta_i)\,\dd\kappa
   \;\le\; \frac{n_i C_g\sqrt{\zeta_i}}{\pi\,\varepsilon^{n_i}}
   \;\to\; 0.
\]
\end{proof}

We now turn to the posterior adaptation claimed in
Theorem~\ref{rem:kh1-prior}(iii), whose tail near $\kappa = 1$ is controlled by the
following analogue of Lemma~\ref{lem:diverge-one-kh0}. As there, we take
$\varepsilon \leq 1/2$ without loss, by monotonicity of the events in
$\varepsilon$.

\begin{lemma}[Divergence near $\kappa=1$, aggregate form]\label{lem:diverge-one}
For every $\varepsilon\in(0,1/2]$,
\[
  \int_{1-\varepsilon}^{1}
    \frac{\kappa^{(n_i-1)/2}\,(1-\kappa)^{n_i/2-1}}
         {\zeta\,\kappa^{n_i}+(1-\kappa)^{n_i}}\,\dd\kappa
   \;\to\; +\infty
  \qquad\text{as } \zeta\to 0^+.
\]
\end{lemma}

\begin{proof}[Lemma~\ref{lem:diverge-one}]
Substituting $t=1-\kappa$ and using $(1-t)^{(n_i-1)/2}\ge 2^{-(n_i-1)/2}$ on
$t\in(0,\varepsilon]$, the integrand is bounded below by
$2^{-(n_i-1)/2}\,t^{n_i/2-1}/\{\zeta(1-t)^{n_i}+t^{n_i}\}$, which increases
monotonically as $\zeta\downarrow 0$ to $2^{-(n_i-1)/2}\,t^{-n_i/2-1}$. The
limit is non-integrable at $0$ (since $n_i/2+1>1$ when $n_i\ge 1$); the
conclusion follows by Monotone Convergence.
\end{proof}

\begin{proof}[Theorem~\ref{rem:kh1-prior}(iii)]
With $k=-1$ and $\sigma_\epsilon=1$, the marginal likelihood is
$y_i\mid\lambda_i\sim\Normal(0,1+\lambda_i^2)$. Since $1+\lambda_i^2=\kappa_i^{-1}$,
the likelihood contributes $\kappa_i^{1/2}\exp(-\kappa_i y_i^2/2)$ to
the tilted density
$\pi_i(\kappa_i \mid y_i, \{\kappa_j\}_{j\sim i}) \propto
\kappa_i^{1/2}\exp(-\kappa_i y_i^2/2)\,\pi_i(\kappa_i \mid \{\kappa_j\}_{j\sim i})$,
the $k_h = 1$ analogue of the conditional posterior used in
Theorem~\ref{thm:posterior-shrinkage}, which depends on $\bfy$ only through
$y_i$ and is proportional to (we reuse the name $g_y$ from the
$k_h = 0$ proof; here the denominator is the aggregate form)
\[
  g_y(\kappa_i)
   \;=\; \kappa_i^{1/2}\,\exp(-\kappa_i y_i^2/2)\,g(\kappa_i;\zeta_i)
   \;=\;
   \frac{\exp(-\kappa_i y_i^2/2)\,\kappa_i^{(n_i-1)/2}\,(1-\kappa_i)^{n_i/2-1}}
        {\zeta_i\,\kappa_i^{n_i}+(1-\kappa_i)^{n_i}}.
\]

\textbf{Proof of (i).}
As $\zeta_i \to 0$, we bound the ratio
$\int_0^{1-\varepsilon}g_y\,\dd\kappa_i / \int_{1-\varepsilon}^1 g_y\,\dd\kappa_i$
uniformly for $\bfy\in\mathcal{Y}_M$. On $(0,1-\varepsilon]$,
$\exp(-\kappa_i y_i^2/2)\le 1$ and the denominator of $g(\kappa_i;\zeta_i)$ is
at least $\varepsilon^{n_i}$, so, using
$\kappa_i^{(n_i-1)/2} \le \kappa_i^{n_i/2-1}$,
$\int_0^{1-\varepsilon}g_y\,\dd\kappa_i \le C_g/\varepsilon^{n_i}$, a finite
constant independent of $\zeta_i$ and $y_i$, with
$C_g = \mathrm{B}(n_i/2, n_i/2)$ as in Lemma~\ref{lem:g-bounds}. On $[1-\varepsilon,1]$,
$\exp(-\kappa_i y_i^2/2)\ge e^{-M^2/2}=:c_M>0$, so
$\int_{1-\varepsilon}^1 g_y\,\dd\kappa_i$ is bounded below by $c_M$ times
the integral in Lemma~\ref{lem:diverge-one}, which diverges to $+\infty$ as
$\zeta_i\to 0$. The ratio thus tends to $0$ uniformly for $\bfy\in\mathcal{Y}_M$.

\textbf{Proof of (ii).}
Fix $\zeta_i\in(0,\infty)$ and let $|y_i|\to\infty$. On $[\varepsilon,1]$,
$\exp(-\kappa_i y_i^2/2)\le e^{-\varepsilon y_i^2/2}$ and the
denominator of $g$ is at least $\zeta_i\varepsilon^{n_i}$, so, using
$\kappa_i^{(n_i-1)/2} \le \kappa_i^{n_i/2-1}$,
$\int_\varepsilon^1 g_y\,\dd\kappa_i \le
\{C_g/(\zeta_i\varepsilon^{n_i})\}\,e^{-\varepsilon y_i^2/2}$ with
$C_g = \mathrm{B}(n_i/2, n_i/2)$.
On $(0,\delta]$ with $\delta=\min(\varepsilon,1/2)$, the substitution
$u=\kappa_i y_i^2/2$, together with $(1-\kappa_i)^{n_i/2-1}\ge c_{n_i} = \min\{1,\, 2^{-(n_i/2-1)}\}$
and $\zeta_i\kappa_i^{n_i}+(1-\kappa_i)^{n_i}\le\zeta_i+1$, gives
\[
  \int_0^\delta g_y\,\dd\kappa_i
   \;\ge\; \frac{c_{n_i}}{\zeta_i+1}\left(\frac{2}{y_i^2}\right)^{(n_i+1)/2}
   \int_0^{\delta y_i^2/2} u^{(n_i-1)/2}e^{-u}\,\dd u,
\]
whose right-hand side is bounded below by
$c'\,|y_i|^{-(n_i+1)}$ for $|y_i| \geq 1$, with $c' > 0$ free of $y_i$. The mass on
$[\varepsilon,1]$ decays exponentially while the mass on $(0,\varepsilon]$
decays only polynomially, so the ratio tends to $0$.
\end{proof}

\section{MCMC sampling algorithm and computational details}
\label{sec:extra-mcmc}
We now show the steps of the Gibbs sampler for the joint
posterior density of $(\bfbeta, \bfh, \sigma_\epsilon^2)$, which also implicitly determines $\pmb \lambda$ and $\tau$. The posterior is
proportional to
\begin{equation}\label{eq:joint-kernel}
  \prod_{i=1}^{n} \phi\bigl(y_i;\, \beta_i,\, \sigma_\epsilon^2\bigr)
  \;\times\;
  \prod_{j} \phi\bigl(\omega_j;\, 0,\, e^{h_j}\bigr)
  \;\times\; \pi(\bfh) \;\times\; \pi(\sigma_\epsilon^2),
\end{equation}
where $\pmb\omega = \Delta^{(k+1)}\bfbeta$, $\pi(\bfh)$ is in Lemma~\ref{lem:level-free},  and $\pi(\sigma_\epsilon^2)$
is the Inverse-Gamma prior \eqref{eq:sigma-prior}. With missing data,
$\bfy$ collects the observed and currently imputed values. Throughout this section we fix $a = b = 1/2$ for the horseshoe GDSP, but modifications are available for general $a,b > 0$.

We initialize the Gibbs sampler as follows:
$\pmb{\beta}^{\left(0\right)}$ is the graph-weighted average of $\pmb y$ (plus Gaussian noise to randomize the initialization), $\sigma_{\epsilon}^{\left(0\right)}$ is the sample standard deviation, and $h_i^{(0)} = \log\{(\omega_i^{(0)})^2 + c\}$, i.e.\ the initial prior variance $\exp(h_i^{(0)})$ is $(\omega_i^{(0)})^2 + c$, where  $\omega^{\left(0\right)} =
\Delta^{\left(k+1\right)}\pmb{\beta}^{\left(0\right)}$ and
$c = 10^{-4}$ is a small offset we add for numerical
stability. 

Each iteration of the Gibbs sampler then cycles through the
following steps.
\begin{enumerate}
\item Draw $[\bfbeta \mid -] \sim \Normal_n(Q^{-1}_{\beta}\bfell_{\beta},\,
  Q_{\beta}^{-1})$, with $Q_\beta$ and $\bfell_\beta$ in \eqref{eq:awol}:
  compute the sparse Cholesky decomposition
  $Q_{\beta} = L_{\beta}L_{\beta}^\top$, solve
  $L_{\beta}\pmb{a}_{\beta} = \bfell_{\beta}$ for $\pmb{a}_{\beta}$, and
  solve $L_{\beta}^\top\bfbeta = \pmb{a}_{\beta} + \pmb{e}_{\beta}$ for
  $\bfbeta$, with $\pmb{e}_{\beta} \sim \Normal_n\left(\pmb 0, I_n\right)$
  (Section~\ref{mcmc-trend}). The sparsity pattern of $Q_\beta$ is fixed
  across iterations; when $Q_\beta$ is banded with bandwidth $b_Q$, the
  Cholesky decomposition requires $\mathcal{O}(n b_Q^2)$ FLOPS and
  $\mathcal{O}(n b_Q)$ memory.
\item Draw $[\bfh \mid -]$ in four sub-steps:
  \begin{enumerate}
  \item Set $\pmb{\omega} = \Delta^{\left(k+1\right)}\bfbeta$ and
    $\tilde\omega_j = \log(\omega_j^2 + c_0)$ as in \eqref{eq:pseudo-data}.
  \item Sample the augmentations for \eqref{eq:prior} and \eqref{eq:prior-3}: draw the mixture indicators $z_j$
    from \eqref{eq:z-update} and 
    draw $[\xi_i \mid -] \sim \mathrm{PG}(1,\, \eta_i)$ independently 
    via \texttt{BayesLogit::rpg}, where $\pmb{\eta} = \Delta_k^{k_h + 1}\bfh$ uses the current value of $\bfh$. 
  \item Draw
    $[\bfh \mid -] \sim \Normal_m\bigl((Q_h^{\star})^{-1}\bfell_h^{\star},\,
    (Q_h^{\star})^{-1}\bigr)$ as defined in Lemma~\ref{lem:h-fcd}. First, compute the sparse
    Cholesky decomposition $Q_{h} = L_{h}L_{h}^\top$ and draw
    $\pmb{x}_0 \sim \Normal_m(Q_h^{-1}\bfell_h^{\star},\, Q_h^{-1})$ using two triangular solves as in step 1. Next, apply the Sherman-Morrison
    correction: (i) solve $L_h\pmb{c} = \pmb{u}$ and
    $L_h^\top\pmb{g} = \pmb{c}$, so that $\pmb{g} = Q_h^{-1}\pmb{u}$ with
    $\pmb{u} = (s_0\, m)^{-1}\bfone_m$, (ii) draw $r_0 \sim \Normal(0, 1)$, and (iii)
    set
    \begin{equation}\label{eq:sherman-morrison}
      \bfh \;=\; \pmb{x}_0 \;-\; \pmb{g}\,
                 \frac{\pmb{u}^\top\pmb{x}_0 + r_0}{1 + \pmb{u}^\top\pmb{g}}.
    \end{equation}
    By design, this algorithm leverages the sparsity of $Q_h$ (inherited from $\Delta_k^{k_h+1}$) for both the initial draw of $\pmb{x}_0$ and subsequent correction, which reuses $L_h$. The actual precision
    $Q_h^{\star} = Q_h + \pmb{u}\pmb{u}^\top$ is dense and never formed.
  \item Set $\hzero = m^{-1}\bfone_m^\top\bfh$, $\tau^2 = \exp(\hzero)$,
    and $\lambda_j^2 = \exp(h_j - \hzero)$ for $j = 1, \ldots, m$ (Lemma~\ref{lem:level-free}). 
  \end{enumerate}
\item Draw $[\sigma_\epsilon^{2} \mid \pmb y, -] \sim
  \mathrm{Inverse\text{-}Gamma}\bigl(a_\sigma + n/2,\;
  b_\sigma + \sum_{i=1}^n(y_i - \beta_i)^2/2\bigr)$.
\item With missing data, draw
  $[\tilde y_i \mid -] \sim \Normal(\beta_i,\, \sigma_\epsilon^2)$
  independently for $i \in \mathcal{I}$ as in \eqref{pred}, and set
  $\pmb y = (\pmb y_{-\mathcal{I}}, \pmb{\Tilde{y}}_{\mathcal{I}})$.
\end{enumerate}

To complete the discussion of the sampler, we derive the full conditional distribution of $\bfh$, justify the sampling step 2(c), and then present the proof of Lemma~\ref{lem:level-free}, which 
clarifies how $\bfh$ encodes both the global and local shrinkage parameters. 

\begin{lemma}[Full conditional of $\bfh$]\label{lem:h-fcd}
Let $Q_h = \Sigma_v^{-1} + \bigl(\Delta_k^{k_h+1}\bigr)^\top
\Sigma_\xi^{-1}\,\Delta_k^{k_h+1}$,
$Q_h^{\star} = Q_h + (s_0^2m^2)^{-1}\pmb{1}_m\pmb{1}_m^\top$, and
$\bfell_h^{\star} = \Sigma_v^{-1}\bigl(\tilde{\pmb\omega} - \pmb{\mu}_z\bigr)
+ \muG/(s_0^2 m)\pmb{1}_m$. The full conditional distribution of the log-variance is $[\bfh \mid -]\sim\Normal_m\bigl((Q_h^{\star})^{-1}\bfell_h^{\star},\,
(Q_h^{\star})^{-1}\bigr)$.
\end{lemma}

\begin{proof}[Lemma~\ref{lem:h-fcd}]
Conditional on $(\pmb\omega, \pmb z, \pmb\xi)$, the full conditional
density of $\bfh$ is proportional to the product of three Gaussian
factors: the pseudo-data model,
$\exp\{-\tfrac12(\tilde{\pmb\omega} - \pmb\mu_z -
\bfh)^\top\Sigma_v^{-1}(\tilde{\pmb\omega} - \pmb\mu_z - \bfh)\}$; the
conditional difference prior,
$\exp\{-\tfrac12\,\bfh^\top(\Delta_k^{k_h+1})^\top\Sigma_\xi^{-1}
\Delta_k^{k_h+1}\bfh\}$; and the level prior,
$\exp\{-(\hzero - \muG)^2/(2s_0^2)\}$ with
$\hzero = m^{-1}\bfone_m^\top\bfh$. Expanding the square in the level
factor,
\begin{equation*}
  -\frac{(\hzero - \muG)^2}{2 s_0^2}
  \;=\; -\frac{1}{2}\,\bfh^\top
    \Bigl(\frac{1}{s_0^2 m^2}\,\bfone_m\bfone_m^\top\Bigr)\bfh
    + \frac{\muG}{s_0^2 m}\,\bfone_m^\top\bfh + \mbox{const},
\end{equation*}
so the three quadratic forms and the two linear terms collect into
$-\tfrac12\,\bfh^\top Q_h^{\star}\bfh + \bfh^\top\bfell_h^{\star}$. The full
conditional is
$\Normal_m\bigl((Q_h^{\star})^{-1}\bfell_h^{\star},
(Q_h^{\star})^{-1}\bigr)$.
\end{proof}

\begin{proof}[Sampling step 2(c)]
We verify that the draw~\eqref{eq:sherman-morrison} of step 2(c) is exact.
First, we draw $\pmb{x}_0 \sim \Normal_m\bigl(Q_h^{-1}\bfell_h^{\star},\, Q_h^{-1}\bigr)$. Next, form $\pmb{g} = Q_h^{-1}\pmb{u}$, with $\pmb{u} = (s_0m)^{-1}\bfone_m$ and $c = 1 + \pmb{u}^\top\pmb{g}$. Notice that $Q_h^{\star} = Q_h + \pmb{u}\pmb{u}^\top$. By the Sherman-Morrison identity, $(Q_h^{\star})^{-1} = Q_h^{-1} - \pmb{g}\pmb{g}^\top/c$. Applying this to $\pmb{u}$ together with $\pmb{u}^\top\pmb{g} = c - 1$, we have $(Q_h^{\star})^{-1}\pmb{u} = \pmb{g}/c$.
Next, independently draw $r_0 \sim \Normal(0,1)$ and return~\eqref{eq:sherman-morrison}, which we rearrange into
\begin{equation*}
      \bfh = A\pmb{x}_0 - (r_0/c)\pmb{g},\quad A = I_m - \pmb{g}\pmb{u}^\top / c.
\end{equation*}
Thus, $\bfh$ is an affine function of independent Gaussians, so we only need to match its first two moments to those of
the full conditional. For the mean, $\mathbb{E}(r_0) = 0$ and $AQ_h^{-1} = Q_h^{-1} - \pmb{g}\pmb{g}^\top/c = (Q_h^{\star})^{-1}$, so
\begin{equation*}
    \mathbb{E}\left(\bfh \mid -\right) = AQ_h^{-1}\bfell_h^{\star} = (Q_h^{\star})^{-1}\bfell_h^{\star}.
\end{equation*}
For the covariance, $r_0$ and $\pmb{x}_0$ contribute independently. Noting that $A Q_h^{-1} A^\top = (Q_h^{\star})^{-1} A^\top = (Q_h^{\star})^{-1} - \pmb{g}\pmb{g}^\top/c^2$, where the second equality uses $(Q_h^{\star})^{-1}\pmb{u} = \pmb{g}/c$,
\begin{equation*}
    \mathrm{Cov}(\bfh \mid -) = AQ_h^{-1}A^\top + \pmb{g}\pmb{g}^\top/c^2 = (Q_h^{\star})^{-1}.
\end{equation*}
Both moments agree with the full conditional, so~\eqref{eq:sherman-morrison} is an exact draw from it.
\end{proof}

\begin{proof}[Lemma~\ref{lem:level-free}]
We establish the statements in turn. First, the shrinkage operator annihilates constant vectors,
$\Delta_k^{k_h+1}\bfone_m = \pmb 0$ (Lemma~\ref{lem:kernel-preserved}
applied to the shrinkage graph). Therefore, 
$\Delta_k^{k_h+1}(\bfh + c\bfone_m) = \Delta_k^{k_h+1}\bfh$, which implies 
$g(\bfh + c\bfone_m) = g(\bfh)$, establishing (i). 

Taking $c = -\hzero$ shows that
$g(\bfh)$ depends only on the deviations
$\bfh^\perp = \bfh - \hzero\bfone_m$. The change of variables
$\bfh \leftrightarrow (\hzero, \bfh^\perp)$ is linear with constant
Jacobian, so $g$ factors as a constant function of $\hzero$ times a
function of $\bfh^\perp$. Its integral over $\hzero \in \R$ diverges,
while its integral over the zero-sum hyperplane is finite by
Lemma~\ref{lem:hperp-proper}. Because $g$ alone is improper in exactly the
level direction and defines no distribution for $\hzero$, specifying a valid probabilistic model requires assigning a proper marginal distribution to the level. Multiplying the kernel $g(\bfh)$ by a proper Gaussian prior on $\hzero$ defines the joint prior $\pi(\bfh) \propto g(\bfh) \phi\bigl(\hzero; \muG,\, s_0^2\bigr)$, establishing (ii).

Finally, under the same change of variables, this joint prior factors as
$\pi(\hzero, \bfh^\perp) \propto \phi(\hzero;\, \muG,\, s_0^2)\,
g(\bfh^\perp)$. This is a product of proper densities, ensuring that $\pi(\bfh)$
is proper, establishing (iii). The factorization moreover demonstrates that the level and the
deviations are independent, with
$\hzero \sim \Normal(\muG,\, s_0^2)$ and $\bfh^\perp$ having a density
proportional to $g$ on the zero-sum hyperplane.
\end{proof}

\section{Additional simulations and results}
\label{sec:extra-comp}
\subsection{2D Lattice}

\subsubsection{Details on the true trends}

In Table~\ref{tab:kernel-choices}, we provide the details of the true trends to accompany the visualizations in Figures~\ref{fig:sim_lattice}~and~\ref{fig:true-mean-funcs}.

\begin{table}[h]
  \centering \small 
  \caption{True trends $\beta(u_1, u_2)$ over a $d_1 \times d_2$ lattice, expressed as normalized coordinates  $u_1 = (r-1)/(d_1-1) \in [0,1]$ (rows) and $u_2 = (c-1)/(d_2-1)\in[0,1]$ (columns). For \texttt{bg}, the bandwidths are $\sigma_1 = 0.3$ and $\sigma_2 = 0.4$.}
    \begin{tabular}{c| l}   
       Name  & True trend $\beta(u_1, u_2)$ \\
       \hline
        \texttt{bg} &  $ 
             \frac{0.75}{\pi \sigma_1 \sigma_2} \exp\left\{-\frac{(u_1 - 0.2)^2}{\sigma_1^2} - \frac{(u_2 - 0.3)^2}{\sigma_2^2}\right\}  + \frac{0.45}{\pi \sigma_1 \sigma_2} \exp\left\{-\frac{(u_1 - 0.7)^2}{\sigma_1^2} - \frac{(u_2 - 0.8)^2}{\sigma_2^2}\right\} $ \\
        \texttt{blocks} & $  \mathbf{1}\{0.2 \leq u_1 \leq 0.4,   0.6 \leq u_2 \leq 0.9\}  + 0.5\mathbf{1}\{0.6 \leq u_1 \leq 0.8, 0.1 \leq u_2 \leq 0.3\} $\\
        \texttt{blocks+} & $\mathbf{1}\{0.2 \leq u_1 \leq 0.3\} + \mathbf{1}\{0.1 \leq u_2 \leq 0.2\} +0.5\mathbf{1}\{0.7 \leq u_1 \leq 0.8\} +   0.5\mathbf{1}\{0.6 \leq u_2 \leq 0.7\}$\\
        \texttt{lin} & $u_1$ \\
    \end{tabular}
    \label{tab:kernel-choices}
\end{table}

\subsubsection{Computational comparisons}\label{app-compute}
The evaluations of computational efficiency are in  
Figure~\ref{fig:sims-time}. Recall that the frequentist graph trend filtering estimator \eqref{gtf-est}, denoted ``GL", refers to the time to compute the solution path using the \texttt{R} package
\texttt{genlasso}. However, GL does \emph{not} include the time needed to compute a data-driven choice of $\lambda$, for instance using cross-validation. The Bayesian methods all use MCMC and are evaluated based on relative efficiency, which reports the number of effective samples relative to the total number of MCMC samples, and time to 1000 effective samples, which jointly considers computing time and MCMC efficiency. 

\begin{figure}[h]
  \centering
  \includegraphics[width=0.49\linewidth,trim=2 2 2 2,clip]{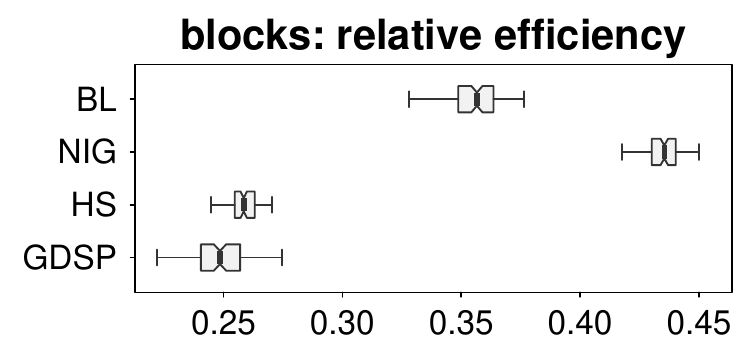}\hfill
  \includegraphics[width=0.49\linewidth,trim=2 2 2 2,clip]{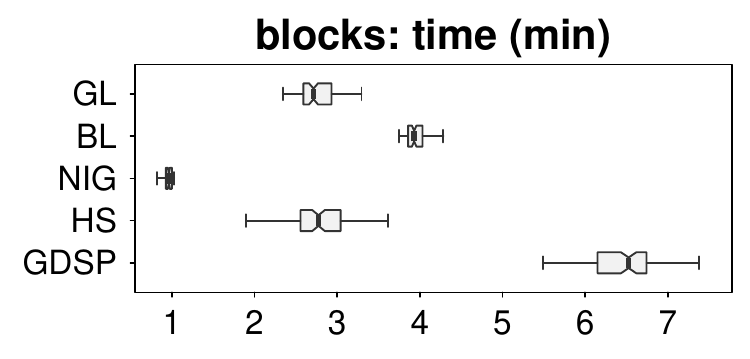}

  \vspace{2pt}
  \includegraphics[width=0.49\linewidth,trim=2 2 2 2,clip]{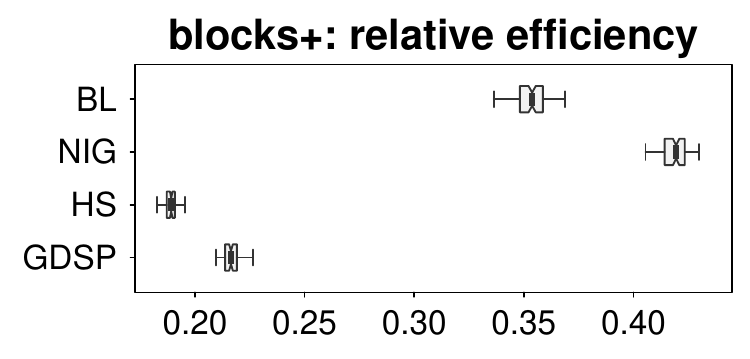}\hfill
  \includegraphics[width=0.49\linewidth,trim=2 2 2 2,clip]{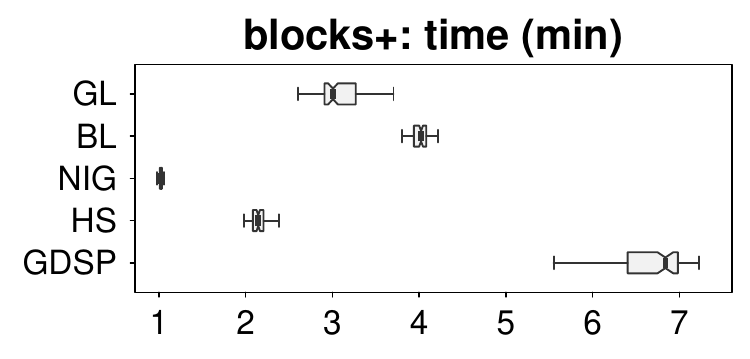}

  \includegraphics[width=0.49\linewidth,trim=2 2 2 2,clip]{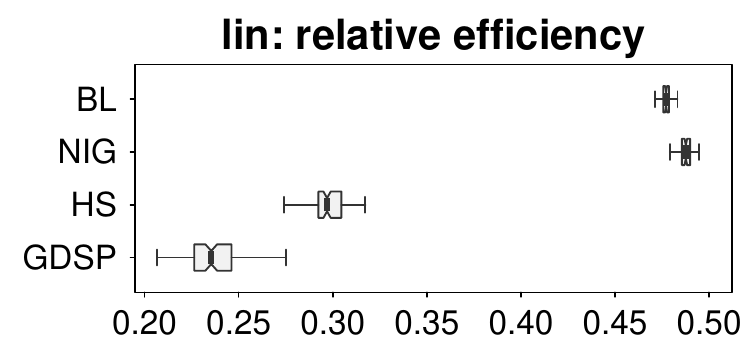}\hfill
  \includegraphics[width=0.49\linewidth,trim=2 2 2 2,clip]{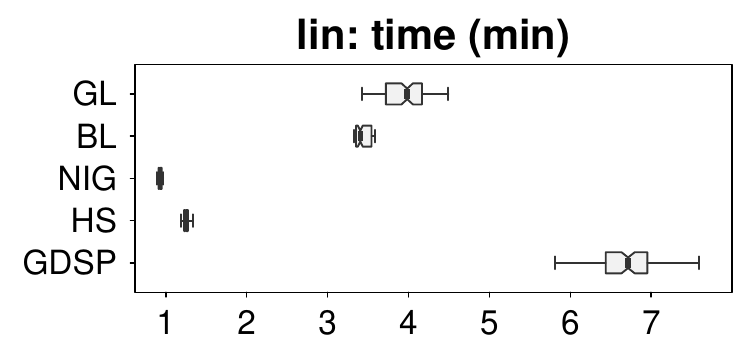}

  \vspace{2pt}
  \includegraphics[width=0.49\linewidth,trim=2 2 2 2,clip]{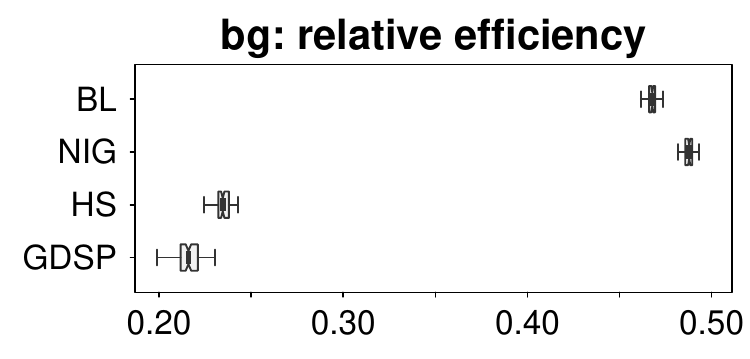}\hfill
  \includegraphics[width=0.49\linewidth,trim=2 2 2 2,clip]{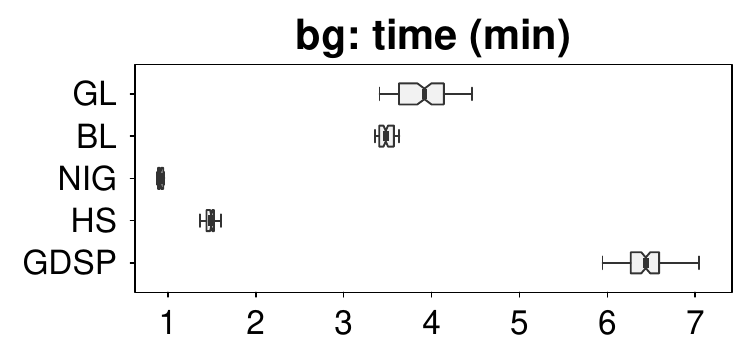}

  \caption{Average relative efficiency $\bar{N}_{\mathrm{eff}}/N$ and wall-clock time to 1000 effective samples $\bar{s}_{1000}$ (GL: full solution path) for the four trends with 50\% missingness.}%
  \label{fig:sims-time}
\end{figure}

As anticipated, the simpler shrinkage priors (NIG, BL) achieve higher MCMC efficiency than the more aggressive shrinkage priors (HS, GDSP). The graph-dependent shrinkage in GDSP incurs a small cost in MCMC efficiency relative to the HS, except in the case of \texttt{blocks+}. Comparing algorithm timing more broadly, we see that the frequentist estimator GL is actually slower than many of the Bayesian graph trend filtering methods. Naturally, the horseshoe GDSP requires the longest time to achieve 1000 effective samples. This is expected: the horseshoe GDSP includes more (shrinkage) parameters than other Bayesian models (NIG, BL) and imposes strong dependencies among those (shrinkage) parameters, which both make the MCMC less efficient and more computational intensive. However, the horseshoe GDSP remains computationally viable, and  only takes about  twice as long as the frequentist estimator GL. We emphasize that GL does \emph{not} include the cost to cross-validate or otherwise select $\lambda$; including this would make the horseshoe GDSP computations far more favorable. 

 \subsubsection{Different settings}
We report additional simulation results for 1) a complete case with no missing data and 2) varying $k = 0$
  and $k = 2$ for each graph trend filtering model. 

\paragraph{Complete case.} We present results under the same design as in Section~\ref{sec:simulation-design}, but without introducing any missingness, in Figure~\ref{fig:sims-complete}. Without any missing data, the frequentist GL of \cite{wangTrendFilteringGraphs2016}, using the oracle value of $\lambda$, now performs best for the non-smooth trends (\texttt{blocks} and \texttt{blocks+}). However, GL remains noncompetitive for the smooth trends (\texttt{lin} and \texttt{bg}) and does not provide interval estimates. The proposed horseshoe GDSP maintains its strong performance among the Bayesian models: it is the most accurate and has the narrowest (calibrated) intervals for the non-smooth trends, while all Bayesian priors perform comparably for the smooth trends, like in the case with missing data.

For completeness, we also include the computational comparisons as in Section~\ref{app-compute}. These are presented in Figure~\ref{fig:sims-time-complete}. As expected, the relative performance among Bayesian methods is similar to that in the missing data case, but the MCMC efficiency improves without the need to impute the missing observations. 

\begin{figure}[h]
  \centering
  \includegraphics[width=0.49\linewidth,trim=2 2 2 2,clip]{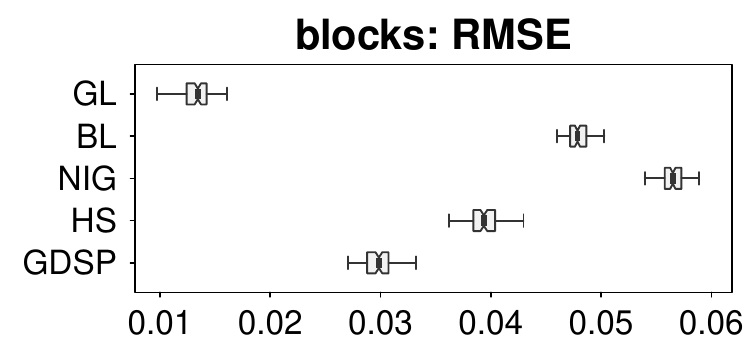}\hfill
  \includegraphics[width=0.49\linewidth,trim=2 2 2 2,clip]{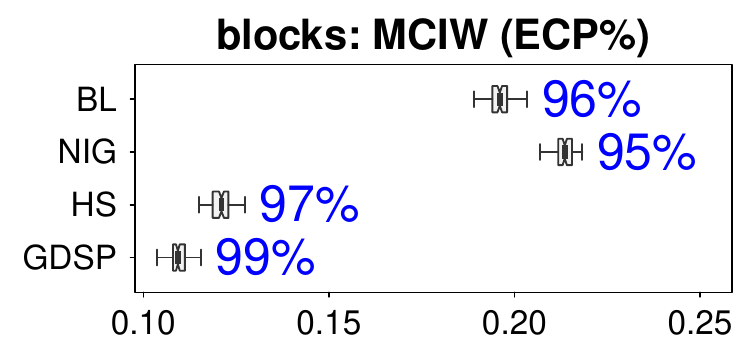}

  \vspace{2pt}
  \includegraphics[width=0.49\linewidth,trim=2 2 2 2,clip]{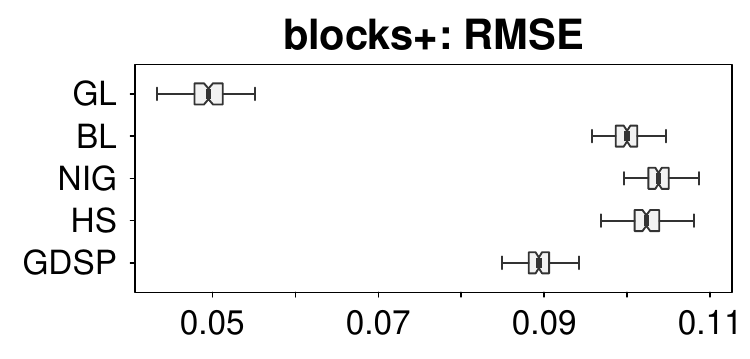}\hfill
  \includegraphics[width=0.49\linewidth,trim=2 2 2 2,clip]{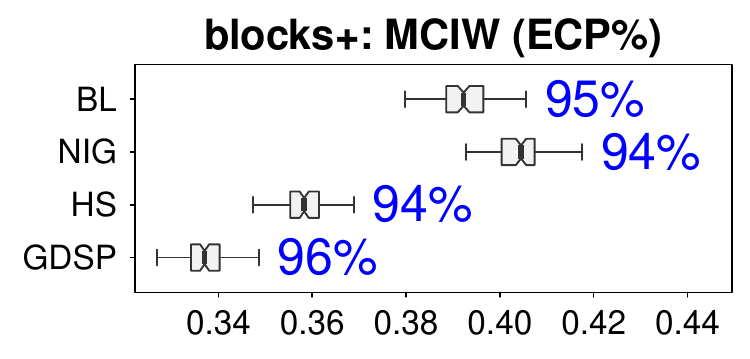}

  \includegraphics[width=0.49\linewidth,trim=2 2 2 2,clip]{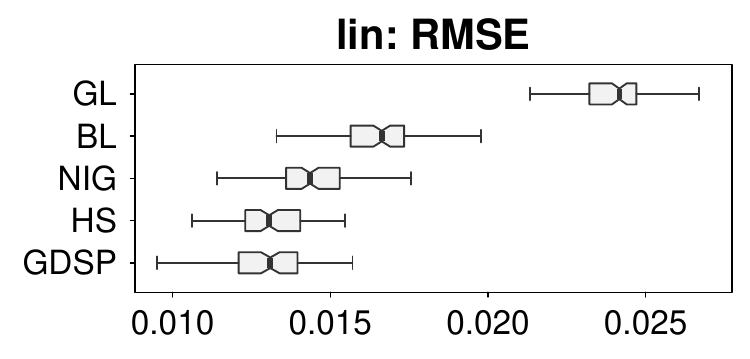}\hfill
  \includegraphics[width=0.49\linewidth,trim=2 2 2 2,clip]{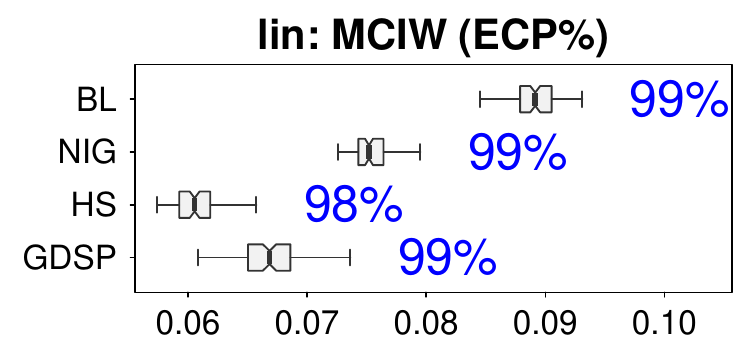}

  \vspace{2pt}
  \includegraphics[width=0.49\linewidth,trim=2 2 2 2,clip]{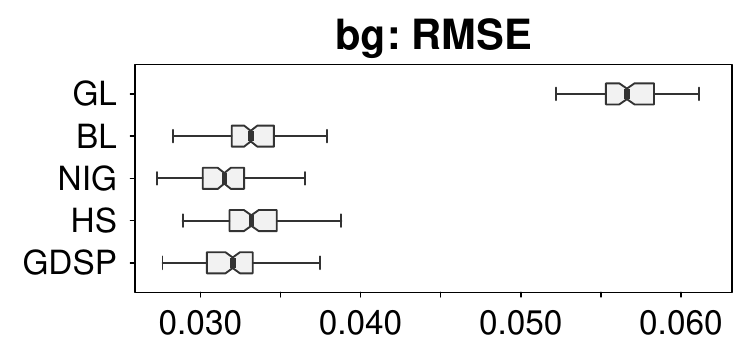}\hfill
  \includegraphics[width=0.49\linewidth,trim=2 2 2 2,clip]{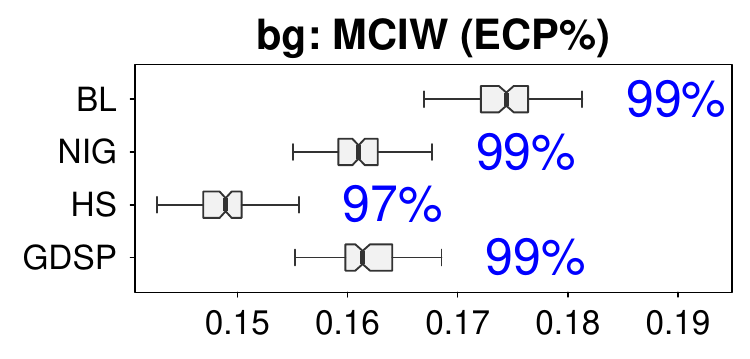}
  
  \caption{RMSEs and MCIWs with ECPs (blue annotations) for the four trends with 0\% missingness; each panel title reports the trend and the metric. The frequentist GL estimator, using the oracle $\lambda$, performs best for \texttt{blocks} and \texttt{blocks+}, but is not competitive for \texttt{lin} and \texttt{bg}. The proposed horseshoe GDSP maintains excellent point estimation and precise yet well-calibrated uncertainty quantification, especially compared to competing Bayesian models.}
    \label{fig:sims-complete}
\end{figure}

\begin{figure}[h]
  \centering
  \includegraphics[width=0.49\linewidth,trim=2 2 2 2,clip]{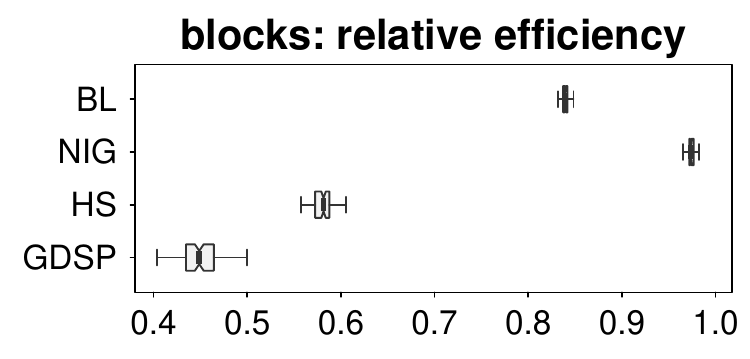}\hfill
  \includegraphics[width=0.49\linewidth,trim=2 2 2 2,clip]{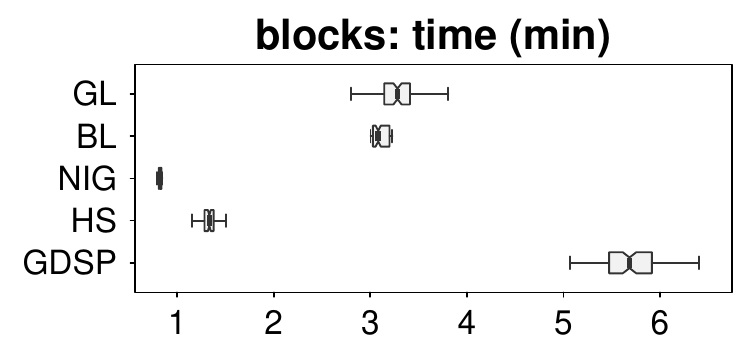}

  \vspace{2pt}
  \includegraphics[width=0.49\linewidth,trim=2 2 2 2,clip]{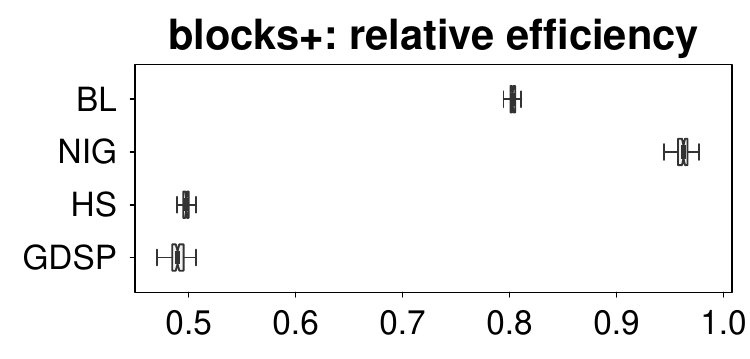}\hfill
  \includegraphics[width=0.49\linewidth,trim=2 2 2 2,clip]{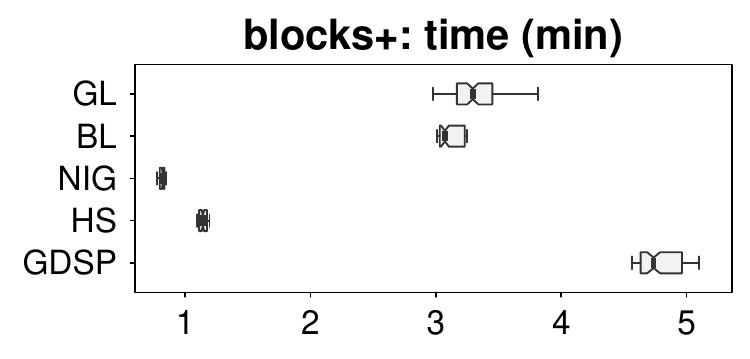}

  \includegraphics[width=0.49\linewidth,trim=2 2 2 2,clip]{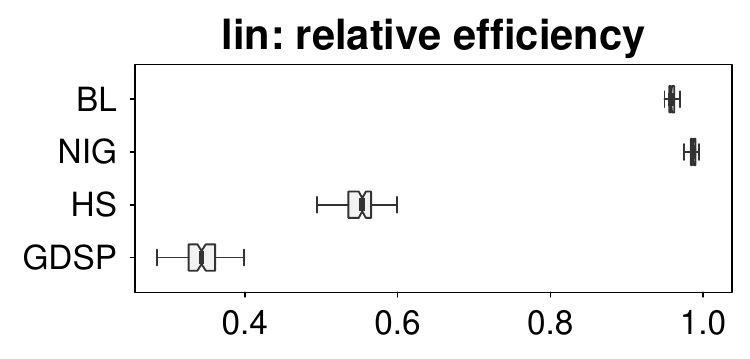}\hfill
  \includegraphics[width=0.49\linewidth,trim=2 2 2 2,clip]{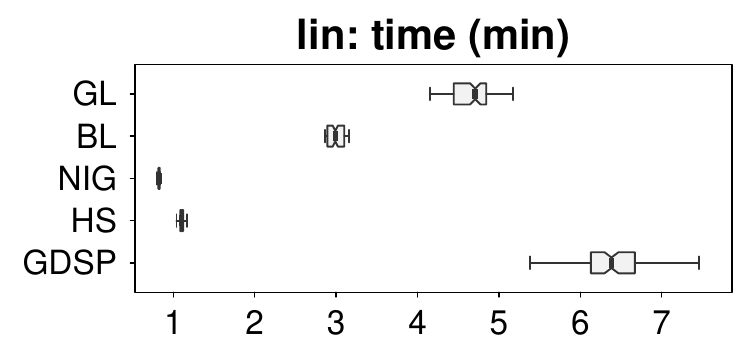}

  \vspace{2pt}
  \includegraphics[width=0.49\linewidth,trim=2 2 2 2,clip]{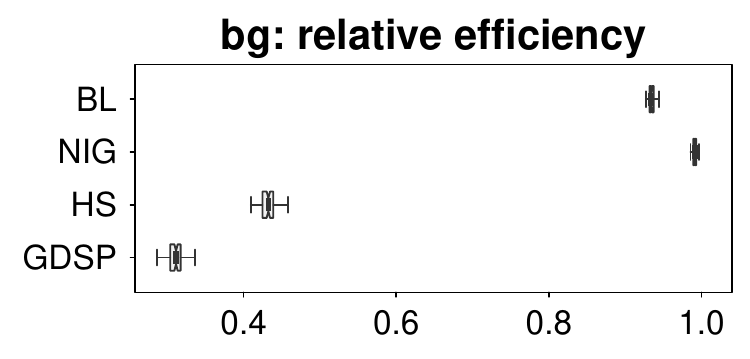}\hfill
  \includegraphics[width=0.49\linewidth,trim=2 2 2 2,clip]{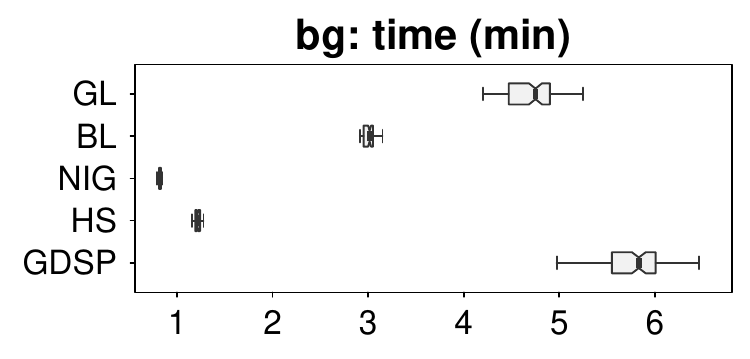}
  
  \caption{Average relative efficiency $\bar{N}_{\mathrm{eff}}/N$ and wall-clock time to 1000 effective samples $\bar{s}_{1000}$ (GL: full solution path) for the four trends with no missingness.}
  \label{fig:sims-time-complete}

\end{figure}

\paragraph{Graph trend filtering with $k = 0$ and $k = 2$.} We compare performance of the horseshoe GDSP across $k\in\{0,1,2\}$ in Figure~\ref{fig:many-k}. We focus on one non-smooth (\texttt{blocks}) trend and one smooth (\texttt{bg}) trend.
In addition to RMSE and MCIW and ECP, we include WAIC \citep{watanabeAsymptoticEquivalenceBayes2010}  in Figure~\ref{fig:many-k-waic}. As expected, increasing $k$ encourages smooth trend estimates, which is most favorable for smooth trends. Larger $k$ also provides narrower intervals. The choice of $k=1$ from Section~\ref{sec:competing-methods} performs well in both cases. However, increasing $k$ also produces a denser graph difference operator $\Delta^{(k+1)}$, which generally increases computing time.

\begin{figure}[h]
    \centering
    \includegraphics[width=\textwidth]{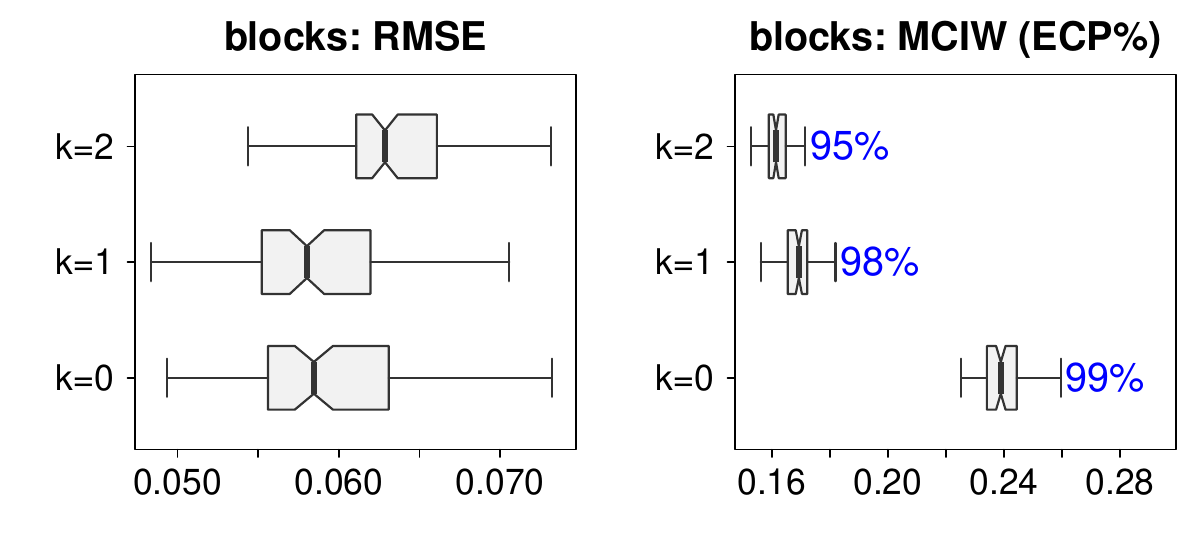}
    \vspace{4pt}
    \includegraphics[width=\textwidth]{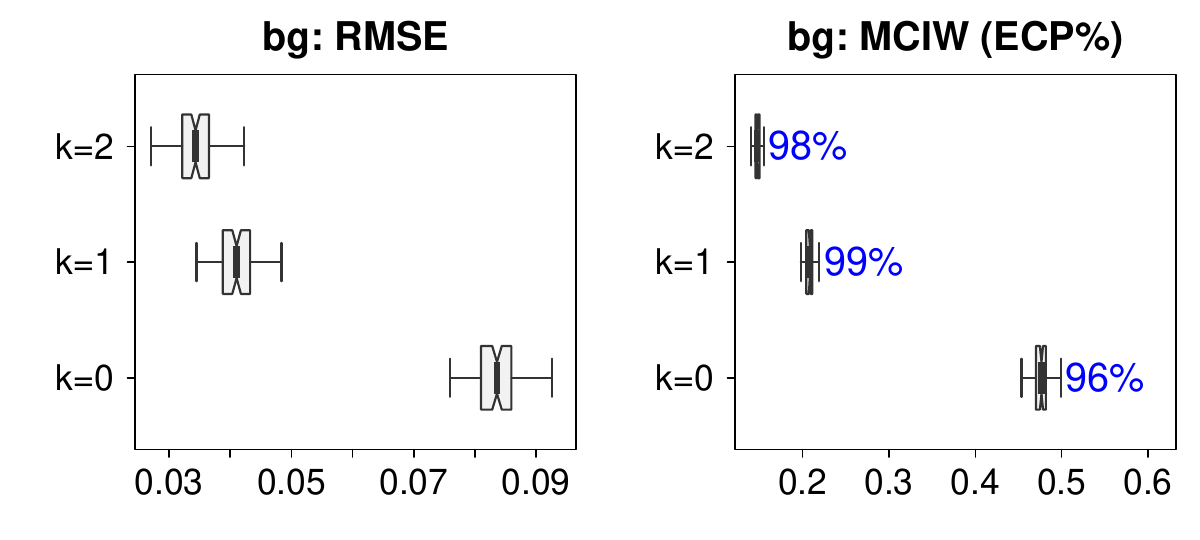}
    \caption{Performance of the horseshoe GDSP for \texttt{blocks} (top) and \texttt{bg} (bottom) trends on a $50 \times 50$ lattice (50\% missingness) across $k\in \{0,1,2\}$. The higher order penalties (larger $k$) encourage greater smoothness, provide narrower intervals, and perform best on smooth trends. The choice of $k=1$ from the main paper is competitive for both smooth and non-smooth trends.} 
    \label{fig:many-k}
\end{figure}

\begin{figure}[h]
    \centering
    \includegraphics[width=\textwidth]{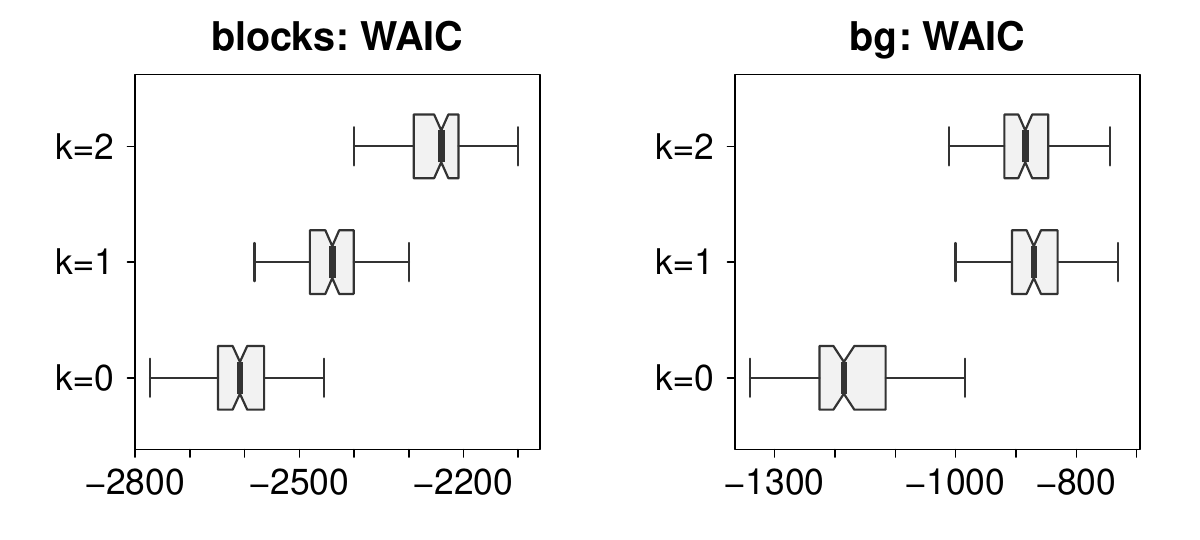}
    \caption{WAIC for the horseshoe GDSP across $k \in \{0,1,2\}$ for the
    \texttt{blocks} (left) and \texttt{bg} (right) trends on a $50 \times 50$
    lattice with 50\% missingness. Lower WAIC indicates better out-of-sample prediction.}
    \label{fig:many-k-waic}
\end{figure}

\subsection{Time series of images}
We complement the spatio-temporal results of Section~\ref{sec:st_sims} for the remaining trends. Figure~\ref{fig:st-summary-supp} shows the  forecast RMSEs and MCIW/ECP  at $t = T$ for \texttt{bg} and \texttt{blocks+}, using the same simulation setup ($d_1 = d_2 = 40$, $T = 4$, $\mathrm{RSNR} = 3$). The findings are consistent with those for \texttt{lin} and \texttt{blocks}: the proposed horseshoe GDSP achieves the lowest forecast RMSE  for both trends; for \texttt{bg} it also provides the narrowest (calibrated) intervals, while for \texttt{blocks+} the interval widths are all similar. All Bayesian methods significantly outperform GL in point estimation, with the advantage being particularly pronounced under missingness.

\begin{figure}[h]
    \centering
    \includegraphics[width=\linewidth]{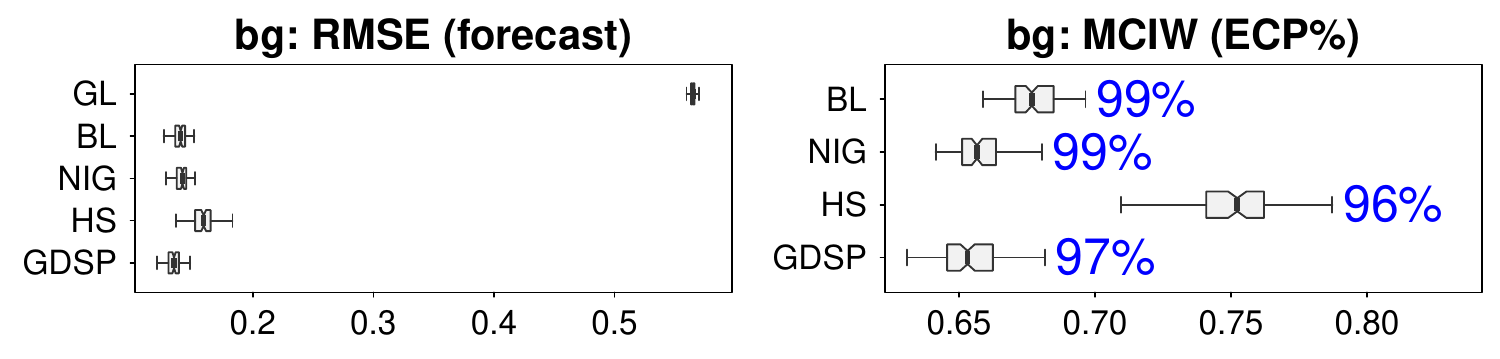}
    \vspace{2pt}
    \includegraphics[width=\linewidth]{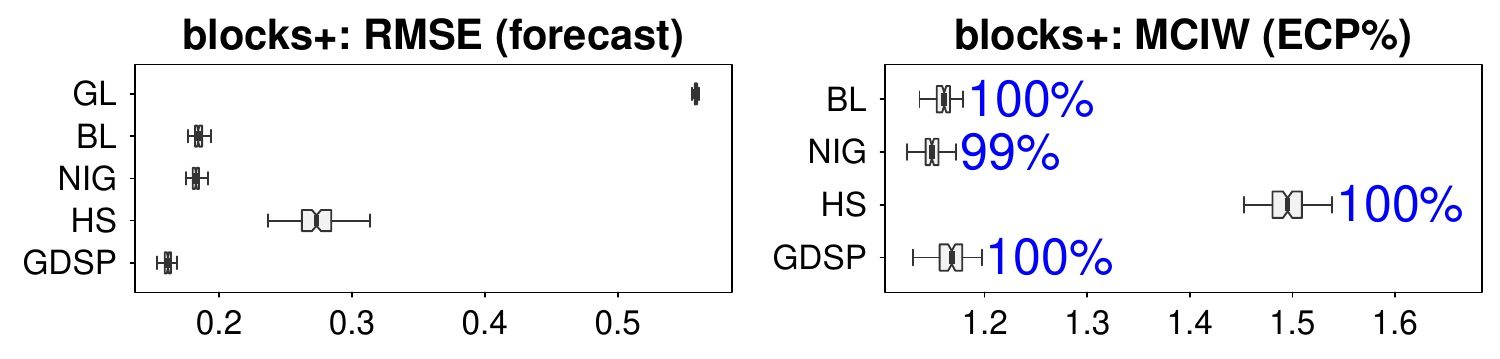}
    \caption{Forecast RMSEs and MCIWs with ECPs (blue annotations) at $t = T$ for the \texttt{bg} and \texttt{blocks+} spatio-temporal trends ($d_1 = d_2 = 40$, $T = 4$, $\mathrm{RSNR} = 3$); each panel title reports the trend and the metric. The proposed horseshoe GDSP provides accurate point estimates and narrow intervals that achieve the nominal coverage.}
    \label{fig:st-summary-supp}
\end{figure}

\section{Additional results and MCMC diagnostics for the U.S.\ unemployment analysis}\label{app-data}
In Figure~\ref{fig:us-traceplot}, we report trace plots of the trend $\beta$ for three representative counties: San Juan County, UT (FIPS 49037; degree 14, the most connected, hence a potentially challenging parameter for MCMC), Los Angeles County, CA (FIPS 06037; large metro), and Washington County, KY (FIPS 21229; from an anomalous block of Kentucky counties), for one observed month (2020-06) and the forecast month (2020-07). The most-connected county corresponds to greater dependencies in the trend and in the log-shrinkage, and thus represents a potentially challenging case for MCMC. 
The observed county-months mix well for all three counties. Mixing degrades as the information at a node decreases. Bulk effective sample sizes (ESS) per county-month for the GDSP trend have median $4688$ of the $5000$ stored draws at the observed county-months, and $2349$ at the held-out county-months of the imputation task, whose spatial and temporal neighbors are largely observed. On the forecast slice, which is informed only through the imputation step (Section~\ref{mcmc-imp}) and the graph, the median ESS is $414$ and the 5th percentile is 136. The implied Monte Carlo error is small relative to the reported posterior uncertainty. The Monte Carlo standard error of a county's posterior-mean forecast has median $1.2\%$ of its $95\%$ forecast-interval width ($2.1\%$ for $95\%$ of counties), and $0.5\%$ at the imputation targets.

\begin{figure}
    \centering
    \includegraphics[width=\linewidth]{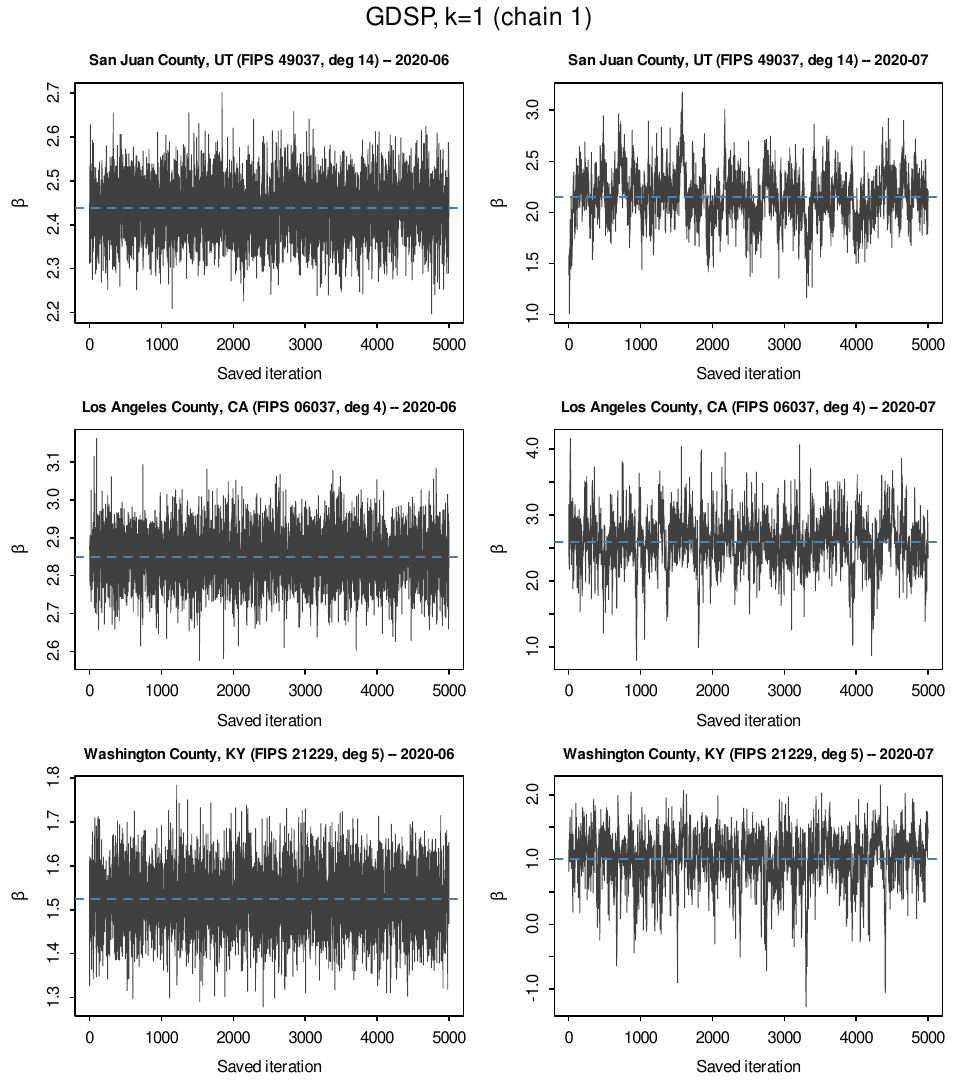}
    \caption{Trace plots of the trend $\beta$ for three representative U.S.\ counties (rows): San Juan County, UT (degree 14, the most connected); Los Angeles County, CA; and Washington County, KY. Columns: the last observed month (2020-06) and the forecast month (2020-07). Each panel shows the $5000$ stored draws of the single chain described in Section~\ref{sec:us_unemployment}. Mixing is fast for the observed months; for the forecast month, which has no data, the chain explores a wider range, and the intervals in Figure~\ref{fig:us_profile} are correspondingly wider.}
    \label{fig:us-traceplot}
\end{figure}

Figure~\ref{fig:us_profile} complements the maps of
Figure~\ref{fig:us_forecast} with the full county-level profile: all
$3{,}108$ counties ordered by FIPS code, with the realized
log-unemployment, the GDSP fit, and its $95\%$ intervals for each month of
the window. On the observed months the fit tracks the jagged cross-county
profile with narrow intervals, and the oracle-tuned GL (red) deviates
visibly at the held-out county-months it must impute; on the forecast month
the bands widen honestly, by an amount that adapts across regions, with the
realized July values falling inside them almost everywhere.

\begin{figure}[t!]
  \centering
  \includegraphics[width=0.94\linewidth]{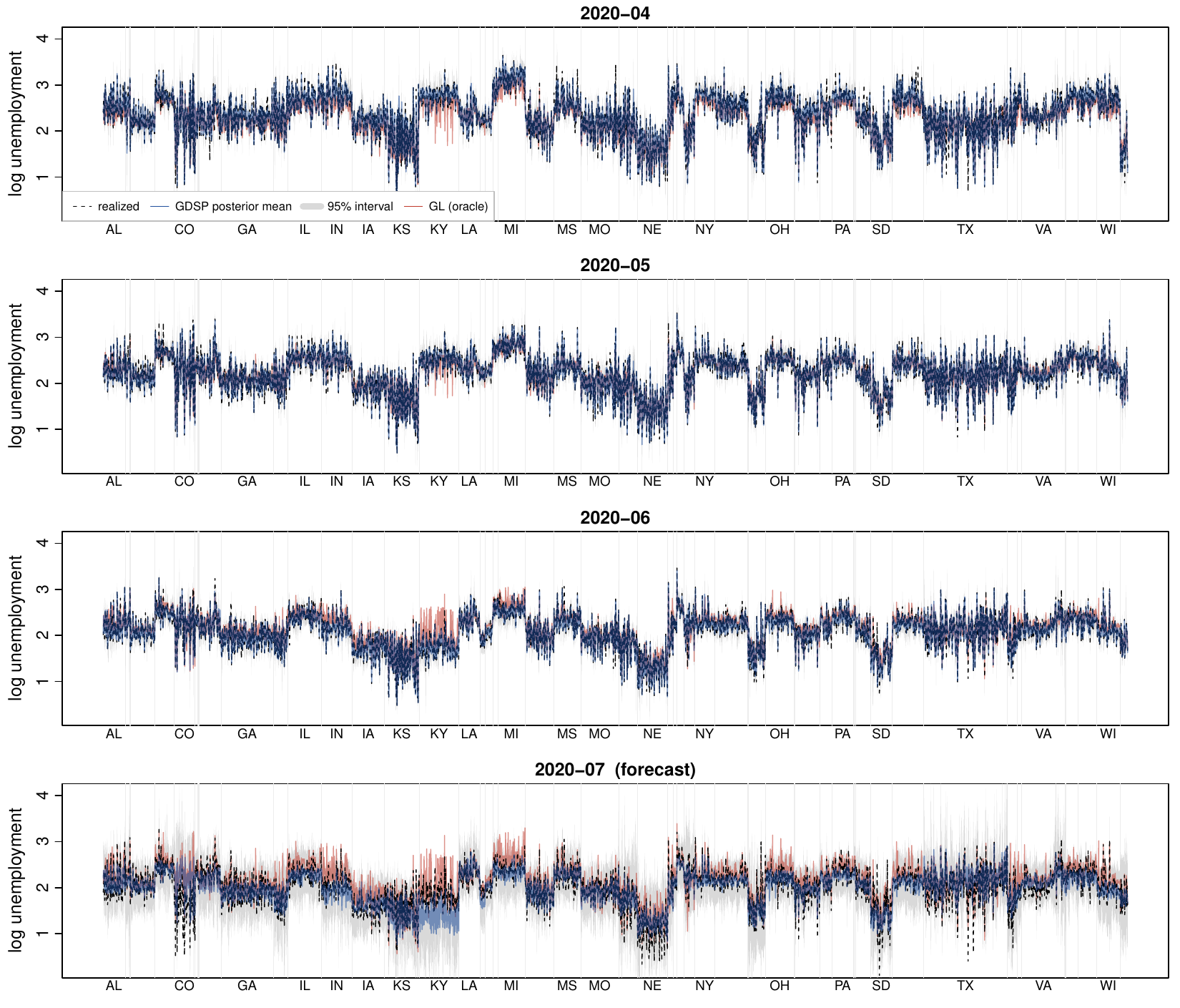}
  \caption{U.S.\ log-unemployment data with all $3108$ counties on
  the horizontal axis, ordered by FIPS code (light vertical lines mark state boundaries; large states labeled) and plotted across different months (top to bottom). Each panel shows realized
  log-unemployment (dashed black), the GDSP posterior mean (blue), $95\%$
  intervals (gray), and the oracle-tuned GL estimate (red). 
  Under the design of
  Section~\ref{sec:us_unemployment}, $25\%$ of the county-months in the
  top three panels are held out and imputed. The bottom
  panel is the out-of-sample July forecast.}
  \label{fig:us_profile}
\end{figure}

\end{document}